\documentclass[a4paper, 4pt]{elsarticle}
\usepackage{mathtools}   
\usepackage{extarrows}
\usepackage{amsmath}     
\usepackage{amssymb}     
\usepackage{hyperref}  
\usepackage{lipsum}
\usepackage{arydshln}
\usepackage[T1]{fontenc}
\usepackage{amssymb}
\usepackage{graphicx}
\usepackage{amsfonts}
\usepackage{booktabs}   
\usepackage{amsmath}    
\usepackage{geometry}
\usepackage{mathtools}
\usepackage{epsfig}
\usepackage{multicol}
\usepackage{amsmath, mathtools}
\usepackage{CJK}
\usepackage{amsmath}
\usepackage{amsthm}
\usepackage{setspace}
\usepackage{graphicx}
\usepackage{epstopdf}
\usepackage{subeqnarray}
\usepackage{mathrsfs}
\usepackage{bbding}
\usepackage{cases}
\usepackage{caption}
\usepackage{subcaption}
\usepackage{mathrsfs}
\usepackage{float}
\usepackage{tikz}
\usepackage{geometry}
\usepackage{lineno,hyperref}
\usepackage{titlesec}

\modulolinenumbers[5]
\journal{the journal}

\newcounter{mycounter}
\newcommand{\ssn}{\mathrm{sn}}

\newtheorem{theorem}{Theorem}
\newtheorem{lemma}{Lemma}
\newtheorem{rmk}{Remark}
\newtheorem{prop}{Proposition}

\biboptions{sort&compress}

\begin{document}
	\title{On the $N$-elliptic localized solutions to the derivative nonlinear Schr\"odinger equation and their asymptotic analysis}
\author[1]{Liming Ling}
\ead{lmling@scut.edu.cn}
\author[1]{Wang Tang\corref{cor1}}  
\ead{terrencet@scut.edu.cn}
\cortext[cor1]{Corresponding author}
\address[1]{School of Mathematics, South China University of Technology, Guangzhou, China}
	\begin{abstract}
		
	We parameterize the elliptic function solutions to the derivative nonlinear Schr\"odinger (DNLS) equation with four independent parameters and generate two equivalent forms of $N$-elliptic localized solutions to the DNLS equation through the Darboux-B\"acklund transformation.
	The $N$-elliptic localized solutions are expressed as (the derivative of) the ratios of determinants with entries in terms of Weierstrass sigma functions. Moreover, the asymptotic behaviors of both forms of $N$-elliptic localized solutions are analyzed along and between the propagation directions as $t\rightarrow \pm\infty$, which verify that the collisions between elliptic-solutions are elastic.
	We prove that the solution tends to a simple elliptic localized solution along each propagation direction. Between the propagation directions, the solution asymptotically approaches a shifted background. Furthermore, we establish sufficient conditions for strictly elastic collisions. The dynamic behaviors of the solutions are systematically investigated, with analytical results visualized through graphical illustrations. The asymptotic analysis of these solutions confirms that they exhibit the behavior predicted by the generalized soliton resolution conjecture on the elliptic function background.
	\end{abstract}

	\maketitle	
	\section{Introduction}
In this work, we focus on $N$-elliptic localized solutions for the derivative nonlinear Schr\"odinger (DNLS) equation \cite{Mjolhus_1978}
\begin{align}\label{Eq-DNLS-equation}
	\mathrm{i}u_t + u_{xx} + 2\mathrm{i}\big(|u|^2u\big)_x = 0,
\end{align}
where $u(x,t): \mathbb{R} \times \mathbb{R} \mapsto \mathbb{C}$.

The DNLS equation \eqref{Eq-DNLS-equation} was first proposed by Rogister in 1971 \cite{Rogister1971ParallelPO} and later derived independently by Mio et al. \cite{DNLS-derivation-1} and Mjolhus \cite{DNLS-derivation-2}. It has applications in diverse fields, including nonlinear optical pulse propagation in fibers \cite{Agrawal2008Applications}, modeling of subpicosecond or femtosecond pulses in single-mode fibers \cite{1987Soliton}, and the study of long-wavelength dynamics of dispersive Alfv\'{e}n waves in plasma physics \cite{Rogister1971ParallelPO}. Kaup and Newell derived its Lax pair using inverse scattering transformation \cite{Kaup:1978hw}. Subsequent methods such as Hirota's bilinear method \cite{1995Bilinearization} and Darboux transformation \cite{1988Darboux,2005Darboux,DNLS-DT,Cieli2009Algebraic,DNLS-DT-2003,DNLS-DT-2011} have been utilized to generate soliton and rational solutions on constant backgrounds.

Recently, Grinevich and Santini \cite{Grinevich_2018} revealed that genus-$1$ or genus-$2$ solutions may develop from the Cauchy problem of the NLS equation with a small perturbation. Prior to the formation of rogue waves and Akhmediev breathers, a nonzero background represented by elliptic functions emerges in the genus-$1$ case. Thus, exact solutions over complicated backgrounds, particularly elliptic function backgrounds, are more likely to arise in experiments and physical settings. Consequently, studying explicit solutions over elliptic backgrounds is of significant importance. We now review the history of studies on explicit solutions on elliptic backgrounds. In \cite{1994-Algebro}, multi-phase modulations of cnoidal waves for the defocusing NLS equation were considered through degeneration of finite-gap solutions. Multi-Akhmediev breather solutions on dnoidal backgrounds were also investigated using complex coordinate transformations \cite{1994-Algebro}. Later, Shin derived cnoidal wave solutions \cite{2003Squared} and soliton solutions on cnoidal backgrounds \cite{Shin-coupled-NLS} for coupled nonlinear Schr\"odinger equations via Darboux-B\"acklund transformation. A subsequent study employed Darboux-B\"acklund transformation to obtain dark solitons on cnoidal backgrounds for the defocusing NLS equation \cite{Shin}. Solutions over elliptic function backgrounds were also obtained through nonlocal symmetry localization \cite{2012Explicit} and consistent Riccati expansion \cite{Consistence-Ricatti}. Higher-order breathers and rogue waves on cnoidal and dnoidal backgrounds were numerically established via Darboux-B\"acklund transformation in \cite{2022Higher}. The inverse scattering transformation was utilized to generate multi-soliton solutions in terms of theta functions over elliptic backgrounds \cite{Takahashi_2012,Takahashi_2016}. Rogue waves and high-order rogue waves on cnoidal backgrounds for NLS equation were constructed using Darboux-B\"acklund transformation and numerical method \cite{K.D.J.}. Recently, elliptic function solutions for a family of Boussinesq-type equations were constructed by a direct linearisation scheme \cite{Boussinesq}.

Several systematic methods have recently emerged for constructing exact solutions over elliptic backgrounds via Darboux-B\"acklund transformations. Based on nonlinearization theory of Lax pairs \cite{1999Relation}, Chen and Pelinovsky derived algebraically decaying solitons, rogue waves, and other localized waves over elliptic backgrounds for the NLS equation \cite{CP-NLS}. This approach was extended to the mKdV and DNLS equations \cite{CP-DNLS,CPU-DNLS}. Employing Bargmann constraints with two eigenfunctions further developed the algebraic method to construct exact solutions for NLS equation on doubly-periodic backgrounds \cite{CPW-NLS}. This enhanced version enabled the construction of exact solutions over broader classes of elliptic backgrounds in the mKdV equation \cite{CP-mKdV-2}. Recent advances extended the method to discrete systems \cite{CP-dmKdV,CP-AL}. The algebraic approach now finds widespread application across numerous integrable equations \cite{Zhaqilao,Chen-Fa,K2,KMM,5th-cmKdV,5thNLS}. Using Baker-Akhiezer functions on algebraic curves, the inverse algebraic-geometric (IAG) approach was introduced to construct exact solutions over elliptic backgrounds for many integrable systems with $3 \times 3$ or higher-order Lax pairs, such as the vector Geng-Li model \cite{LRM-GL}, the Yajima-Oikawa equation \cite{LRM-YO}, and so on \cite{LRM-FL}. The IAG approach was later extended to semi-discrete systems as well \cite{LRM-semi}. Another systematic method was developed to construct multi-breathers and rogue waves over elliptic backgrounds for the AKNS equations \cite{LS-NLS,LS-mKdV-solution,LS-sG,Ling-NLS,thetaHirota,thetaLPD} in the compact form with determinant structures. By introducing a uniform parameter $z$ as a substitution for the spectral parameter, multi-breather solutions in terms of theta functions were derived over various elliptic backgrounds. Furthermore, higher-order rogue wave solutions over elliptic backgrounds were obtained for these systems using generalized Darboux transformation \cite{LS-RW}.

Soliton resolution constitutes a fundamental research direction in the theory of integrable systems and dispersive equations. The conjecture posits that solutions with generic initial data of weakly nonlinear dispersive equations asymptotically decompose into a finite sum of solitons, a radiative component governed by a linear equation, and a term vanishing as $t \rightarrow +\infty$ \cite{2008Why}. First emerging during the 1970s-1980s through studies of the Korteweg-de Vries (KdV) equation's integrability \cite{Miura_1976,Segur_1973,Schuur_1986} and numerical experiments \cite{Ivancevic_2010}, this conjecture found its initial rigorous verification via inverse scattering transformations \cite{Segur_1976_I,Segur_1976_II,Novokshenov_1980}. Substantial theoretical progress was later achieved through modern analytical methods, particularly through Deift and Zhou's nonlinear steepest descent approach for Riemann-Hilbert problems \cite{Deift-Zhou}, which established soliton resolution for numerous integrable systems \cite{YYL,LJQ}.

While first- and second-order rogue waves for the DNLS equation were constructed in \cite{CP-DNLS}, exact solutions derived from $N$-fold Darboux transformation remain systematically unexplored in compact form.  Building on \cite{LS-mKdV-solution,LS-NLS,LS-sG,effective-integrationeffective-integration,CP-DNLS}, we construct explicit $N$-elliptic localized solutions for \eqref{Eq-DNLS-equation} using Weierstrass functions via Darboux-B\"acklund transformations and addition formulas, followed by rigorous asymptotic analysis. Our main contributions are summarized below.

Unlike prior studies \cite{LS-NLS,LS-mKdV-solution,LS-sG,Ling-NLS,thetaHirota,thetaLPD}, our solutions utilize Weierstrass functions instead of Jacobi theta functions. This formulation offers distinct advantages: the addition formula \eqref{Eq-addition formulas of the sigma functions} for Weierstrass sigma functions eliminates indices, yielding more concise expressions than those involving Jacobi theta functions. Structures such as $\wp(\theta_1) - \wp(\theta_2)$ and $\zeta(\theta_1) + \zeta(\theta_2) + \zeta(\theta_3) - \zeta(\theta_1 + \theta_2 + \theta_3)$ naturally emerge during construction. These directly correspond to known Weierstrass function identities, enabling computational simplification. Such structures are closely related to the Weierstrass sigma addition formula and frequently appear in constructing elliptic function solutions for AKNS hierarchy equations \cite{LS-NLS,LS-mKdV-solution,LS-sG}. While Jacobian elliptic functions require specialized techniques, Weierstrass sigma functions directly facilitate calculations involving these structures, significantly simplifying computations compared to existing approaches.

Combining Jacobi theta addition formulas with Darboux-B\"acklund transformations yields compact elliptic localized solutions for AKNS hierarchy equations \cite{LS-NLS,LS-mKdV-solution,LS-sG,Ling-NLS,thetaHirota,thetaLPD}. This motivates applying similar techniques to Kaup-Newell hierarchy equations. However, the Darboux-B\"acklund transformation for the DNLS equation exhibits greater complexity, necessitating supplementary algebraic methods when integrated with addition formulas. Using the transformation derived from expansion at $\lambda = 0$, we obtain the $N$-elliptic localized solution as the derivative of a ratio of two determinants with elliptic function entries. Although intricate, the B\"acklund transformation from expansion at $\lambda = \infty$ aligns naturally with the addition formula, providing an equivalent derivative-free representation. Direct verification confirms the equivalence between both solution forms, suggesting potential applicability to other Kaup-Newell hierarchy equations.

Inspired by \cite{Takahashi_2016,Ling-NLS}, we introduce a sigma-type Cauchy matrix with corresponding determinant formula and rigorous proof. Leveraging this framework, we analytically investigate solution asymptotics. As $t \to \pm\infty$, we derive asymptotic expressions along propagation directions $L_i^{\pm}$ and within regions $R_i^{\pm}$ for $i=1,2,\ldots,N$. We prove that asymptotic solutions along propagation directions emerge from the background via one-fold Darboux transformation, while intermediate regions correspond to shifted backgrounds. Furthermore, \cite{LS-mKdV-solution} established $u(\xi,t) = u(-\xi,-t)$ symmetry for $N$-elliptic localized solutions to the mKdV equation under specific parameter constraints. We extend this to $u(\xi,t) = u^*(-\xi,-t)$ for DNLS equation, enabling sufficient conditions for strictly elastic collisions.

Recent work \cite{Ling-NLS} extends the soliton resolution framework by demonstrating that the NLS equation's $N$-breather solutions over elliptic backgrounds asymptotically decompose into single-elliptic localized solutions as $t \rightarrow \pm\infty$.
This work demonstrates that $N$-elliptic localized solutions to DNLS equation  decompose into single elliptic localized waves propagating at distinct velocities as $t \rightarrow \pm\infty$, consistent with the generalized soliton resolution conjecture over elliptic backgrounds. This provides the first verification of soliton resolution for exact solutions on complex backgrounds in non-AKNS integrable systems.

The paper is organized as follows. Section~2 revisits Darboux-B\"acklund transformations for DNLS and rigorously establishes equivalence between their distinct forms. Section~3 constructs elliptic function solutions for DNLS and derives the fundamental solution matrix for its Lax pair. Section~4 establishes $N$-elliptic localized solutions. Section~5 introduces the sigma-formulation of Cauchy determinants to enable asymptotic analysis. Section~6 provides an equivalent derivative-free representation via alternative B\"acklund transformations with corresponding asymptotics. Finally, Section~7 examines solution dynamics through analytical results and graphical illustrations.

	\section{The two equivalent forms of Darboux-B\"acklund transformation}
	In this section, we revisit the Darboux-B\"acklund transformation for the DNLS equation \eqref{Eq-DNLS-equation}\cite{DNLS-DT-2003,DNLS-DT-2011,DNLS-DT}.
	Expanding the Darboux matrix at distinct values of the spectral parameter yields different forms of B\"acklund transformations. 
	To the best of our knowledge, the equivalence between the two forms of B\"acklund transformations obtained by expanding the Darboux matrix at $\lambda=0$ and $\lambda=\infty$ respectively remains unproven.
	We denote these forms as $BT_0$ and $BT_\infty$, respectively. 
	Using the Lax pair structure, we rigorously establish their equivalence. 
	
	The DNLS equation \eqref{Eq-DNLS-equation} admits the following Lax pair
	\begin{equation}\label{Eq-Lax pair}
		\begin{split}
			&\Phi_x(x,t;\lambda)=\textbf{U}(\textbf{Q};\lambda)\Phi(x,t;\lambda),\quad \textbf{U}(\textbf{Q};\lambda)=-2\mathrm{i}\sigma_3\lambda^2+2\textbf{Q}\lambda ,\\
			& \Phi_t(x,t;\lambda)=\textbf{V}(\textbf{Q};\lambda)\Phi(x,t;\lambda),\quad \textbf{V}(\textbf{Q};\lambda)=\left(4\lambda^2+2\textbf{Q}^2\right)\textbf{U}(\textbf{Q};\lambda)+2\mathrm{i}\sigma_3\textbf{Q}_x\lambda,
		\end{split}
	\end{equation}
	where $\lambda\in \mathbb{C}\cup\{\infty\}$ is the spectral parameter and 
	\begin{align}\label{Eq-sigma3 and Q}
		\sigma_3= \begin{pmatrix}
			1 & 0\\
			0 & -1\\
		\end{pmatrix},\quad \textbf{Q}=\begin{pmatrix}
			0 & \mathrm{i}u\\
			\mathrm{i}u^* & 0\\
		\end{pmatrix},
	\end{align}
	with $^*$ denoting complex conjugation.
	The compatibility condition of above Lax pair \eqref{Eq-Lax pair}:
	\begin{equation}\label{Eq-compatibility-condition}
		\mathbf{U}_t-\mathbf{V}_x+[\mathbf{U},\mathbf{V}]=0,
	\end{equation}
	yields the DNLS equation \eqref{Eq-DNLS-equation}. Matrices  $\textbf{U}(\textbf{Q};\lambda)$
	and $\textbf{V}(\textbf{Q};\lambda)$ involved in the Lax pair \eqref{Eq-Lax pair} satisfy the $\mathrm{su}(2)$-symmetry
	\begin{equation}\label{Eq-symmetry-1}
		\textbf{U}^{\dagger}(\textbf{Q};\lambda^*)=-\textbf{U}(\textbf{Q};\lambda),\quad \textbf{V}^{\dagger}(\textbf{Q};\lambda^*)=-\textbf{V}(\textbf{Q};\lambda),
	\end{equation}
	where $^{\dagger}$ represents the Hermitian conjugation. Additionally, they exhibit the twist symmetry 
	\begin{equation}\label{Eq-symmetry-2}
		\textbf{U}(\textbf{Q};-\lambda)=\sigma_3\textbf{U}(\textbf{Q};\lambda)\sigma_3, \quad \textbf{V}(\textbf{Q};-\lambda)=\sigma_3\textbf{V}(\textbf{Q};\lambda)\sigma_3.
	\end{equation}
	Consequently, $\sigma_3\Phi(x,t;-\lambda)\sigma_3$ also solves \eqref{Eq-Lax pair}, while $\Phi^{\dagger}(x,t;\lambda^*)$ solves the adjoint Lax pair
	\begin{align}\label{Eq-adjoint-Lax-pair}
		\Psi_x(x,t;\lambda)=-\Psi(x,t;\lambda) \textbf{U}(x,t;\lambda),\quad \Psi_t(x,t;\lambda)=-\Psi(x,t;\lambda) \textbf{V}(x,t;\lambda),
	\end{align}
	if $\Phi(x,t;\lambda)$ is a solution matrix to  \eqref{Eq-Lax pair}. Direct verification shows that  $\Phi^{-1}(x,t;\lambda)$ also solves the adjoint Lax pair \eqref{Eq-adjoint-Lax-pair} . By the existence and uniqueness theorem for ODEs, the solution matrix $\Phi(x,t;\lambda)$ therefore satisfies the symmetries:
	\begin{align}\label{Eq-symmetry-phi}
		\Phi^{-1}(x,t;\lambda)=\Phi^\dagger(x,t;\lambda^*),\quad \Phi(x,t;-\lambda)=\sigma_3\Phi(x,t;\lambda)\sigma_3.
	\end{align}
	We propose to find the $N$-fold Darboux matrix $\textbf{T}^{[N]}(x,t;\lambda)$ that keeps the symmetries  \eqref{Eq-symmetry-phi}. To be specific, $\textbf{T}^{[N]}(x,t;\lambda)$ needs to satisfy
	\begin{align}\label{Eq-T-symmetry}
		(\textbf{T}^{[N]})^{-1}(x,t;\lambda)=(\textbf{T}^{[N]})^\dagger(x,t;\lambda^*),\quad \textbf{T}^{[N]}(x,t;-\lambda)=\sigma_3\textbf{T}^{[N]}(x,t;\lambda)\sigma_3.	
	\end{align}
	Therefore, along the lines of the reference \cite{DNLS-DT} we could assume that
	\begin{align}\label{Eq-T}
		\textbf{T}^{[N]}(x,t;\lambda)=\mathbb{I}_{2}+\sum_{i=1}^{N}\left(\frac{\textbf{A}_i(x,t)}{\lambda_i^*-\lambda}-\frac{\sigma_3\textbf{A}_i(x,t)\sigma_3}{\lambda+\lambda_i^*}\right)\lambda_i^*\lambda,
	\end{align}
	where $\mathbb{I}_2$ denotes the 2$\times$2 identity matrix and $\mathbf{A}_i(x,t)=\left|x_i(x,t)\right\rangle\left\langle y_i(x,t)\right|$ for $1\leq i\leq N$. Here, $\left|x_i\right\rangle=\big(x_i^{(1)},x_i^{(2)}\big)^T$ and $\left\langle y_i\right|=\big(y_i^{(1)},y_i^{(2)}\big)$ are two-dimensional column and row vectors to be determined respectively. Suppose $u\in L^{\infty}(\mathbb{R})\cap C^{\infty}(\mathbb{R})$ is a solution to the DNLS equation \eqref{Eq-DNLS-equation}, and let $\varphi_i=(\varphi_{i}^{(1)},\varphi_{i}^{(2)})^T, i=1,2,\dots,N$ be $N$ solutions to the corresponding Lax pair \eqref{Eq-Lax pair} at eigenvalues $\lambda_i$. Via the Darboux transformation
	\begin{align}\label{Eq-Darboux-transformation}
		\Phi^{[N]}( x, t;\lambda)=\mathbf{T}^{[N]}( x, t;\lambda) \Phi(x, t;\lambda),
	\end{align}
	the Lax pair \eqref{Eq-Lax pair} is converted into a new one 
	\begin{equation}\label{Eq-New-Lax-pair}
		\begin{split}
			\begin{cases}\Phi^{[N]}_x(x,t;\lambda)=\mathbf{U}(\textbf{Q}^{[N]};\lambda) \Phi^{[N]}(x,t;\lambda),  \\ \Phi^{[N]}_t(x,t;\lambda)=\mathbf{V}(\textbf{Q}^{[N]};\lambda) \Phi^{[N]}(x,t;\lambda), \end{cases}
		\end{split}
	\end{equation}
	where 
	\begin{align}\label{Eq-QN}
		\mathbf{Q}^{[N]}(x,t)=\begin{pmatrix} 0 & \mathrm{i}u^{[N]}\\
			\mathrm{i}(u^{[N]})^* & 0\end{pmatrix}.
	\end{align}
	
	From \eqref{Eq-Lax pair}, \eqref{Eq-New-Lax-pair} and \eqref{Eq-QN}, it follows that $u^{[N]}$ constitutes a new solution to the DNLS equation \eqref{Eq-DNLS-equation}. Substituting the Darboux transformation \eqref{Eq-Darboux-transformation} into \eqref{Eq-New-Lax-pair} and using \eqref{Eq-Lax pair} yields relations between the transformed Lax matrices $\mathbf{U}(\textbf{Q}^{[N]};\lambda),\mathbf{V}(\textbf{Q}^{[N]};\lambda)$ and the original Lax matrices $\mathbf{U}(\textbf{Q};\lambda),\mathbf{V}(\textbf{Q};\lambda)$:
	\begin{equation}\label{Eq-QN-Q}
		\begin{split}
			\mathbf{U}(\textbf{Q}^{[N]};\lambda)\mathbf{T}^{[N]}(x,t;\lambda)=\mathbf{T}^{[N]}_x(x,t;\lambda)+\mathbf{T}^{[N]}(x,t;\lambda)\mathbf{U}(\textbf{Q};\lambda),\\
			\mathbf{V}(\textbf{Q}^{[N]};\lambda)\mathbf{T}^{[N]}(x,t;\lambda)=\mathbf{T}^{[N]}_t(x,t;\lambda)+\mathbf{T}^{[N]}(x,t;\lambda)\mathbf{V}(\textbf{Q};\lambda).
		\end{split}
	\end{equation}
	These relations \eqref{Eq-QN-Q} are crucial for determining the corresponding B\"acklund transformation. Subsequent proposition presents the B\"acklund transformation $BT_0$, which is derived by expanding the Darboux matrix $\mathbf{T}^{[N]}$ at $\lambda=0$. This formulation of B\"acklund transformation is frequently utilized in literature. Henceforth, we only present the conclusions here without giving the proof. We recommend the reference \cite{DNLS-DT} for details.

	\begin{prop}[\cite{DNLS-DT}The B\"acklund transformation $BT_0$]\label{thm:DT}
		Suppose that $u\in L^{\infty}(\mathbb{R})\cap C^{\infty}(\mathbb{R})$ is a solution to the DNLS equation \eqref{Eq-DNLS-equation}, while $\varphi_i=(\varphi_{i}^{(1)},\varphi_{i}^{(2)})^T ,i=1,2,...,N$ are solutions to the corresponding Lax pair \eqref{Eq-Lax pair} at eigenvalues $\lambda_i$  respectively. Then
		\begin{align}\label{Eq-BT0-1}
			u^{[N]}_0=u+\mathrm{i}(\varphi^{(2)\dagger}\mathbf{M}^{-1}\varphi^{(1)})_x,
		\end{align}
		is a new solution to the DNLS equation \eqref{Eq-DNLS-equation}, where $\varphi^{(i)}=(\varphi_{1}^{(i)},\varphi_{2}^{(i)},...,\varphi_{N}^{(i)})^T,i=1,2$  and the matrix elements of the $N\times N$ matrix $\mathbf{M}_{ij}$ are defined by
		\begin{equation}\label{Eq-BT0-2}
			\mathbf{M}_{ij}=\left(\frac{ \varphi_j^{\dagger} \varphi_i}{\lambda_j^{*}-\lambda_i}-\frac{ \varphi_j^{\dagger} \sigma_3 \varphi_i}{\lambda_j^{*}+\lambda_i}\right)\lambda_i\lambda_j^{*}.
		\end{equation}
	\end{prop}
	The following theorem establishes another form of B\"acklund transformation by expanding the Darboux matrix at $\lambda=\infty$. This formulation was also reported in \cite{DNLS-DT-2003,DNLS-DT-2011} with another determinant form. Here we present the $N$-fold $BT_\infty$ transformation in the compact form and provides its rigorous proof.
	\begin{theorem}\label{thm:DT-infty}
		Suppose that $u\in L^{\infty}(\mathbb{R})\cap C^{\infty}(\mathbb{R})$ is a solution to the DNLS equation \eqref{Eq-DNLS-equation} while $\varphi_i=(\varphi_{i}^{(1)},\varphi_{i}^{(2)})^T,i=1,2,...,N$ are solutions to the corresponding Lax pair \eqref{Eq-Lax pair} at eigenvalues  $\lambda_i$  respectively. Then 	
		\begin{align}\label{Eq-DT-infty}
			u^{[N]}_\infty=\frac{1+2\varphi^{(1)\dagger}\mathbf{\Lambda}^\dagger\mathbf{M}^{-1}\varphi^{(1)}}{1-2\varphi^{(2)\dagger}\mathbf{\Lambda}^\dagger\left(\mathbf{M}^\dagger\right)^{-1}\varphi^{(2)}}u+\frac{4\varphi^{(2)\dagger}\big(\mathbf{\Lambda^{\dagger}}\big)^2\mathbf{M}^{-1}\varphi^{(1)}}{1-2\varphi^{(2)\dagger}\mathbf{\Lambda}^\dagger\left(\mathbf{M}^\dagger\right)^{-1}\varphi^{(2)}},	
		\end{align}
		is a new solution to the DNLS equation \eqref{Eq-DNLS-equation},
		where $\varphi^{(i)}=(\varphi_{1}^{(i)},\varphi_{2}^{(i)},...,\varphi_{N}^{(i)})^T,i=1,2$ and  $\mathbf{\Lambda}=\mathrm{diag}(\lambda_i),1\leq i\leq N$. Moreover, the entries of $\mathbf{M}$ are given in \eqref{Eq-BT0-2}. 
	\end{theorem}
	\begin{proof}
		As $\lambda\rightarrow \infty$, we have the following expansion
		\begin{equation}\label{Eq-T-expansion-infty}
			\begin{split}
				\textbf{T}^{[N]}(x,t;\lambda)&=\left(\textbf{I}-\sum_{i=1}^N(\textbf{A}_i+\sigma_3\textbf{A}_i\sigma_3)\lambda_i^*\right)+\sum_{i=1}^N(-\textbf{A}_i+\sigma_3\textbf{A}_i\sigma_3){\lambda_i^{*}}^{2}\lambda^{-1}+\mathcal{O}(\lambda^{-2})\\
				&:=\mathbf{T}_{0}^{[N]}+\mathbf{T}_{-1}^{[N]}\lambda^{-1}+\mathcal{O}(\lambda^{-2}).
			\end{split}
		\end{equation}
		The matrices $\mathbf{T}_{0}^{[N]}$ and $\mathbf{T}_{-1}^{[N]}$ can be expressed through their entries as
		\begin{equation}\label{Eq-T0-T1-entries}
			\begin{split}
				\mathbf{T}_{0}^{[N]}&=\begin{pmatrix}
					1-2\sum_{i=1}^N x_i^{(1)}y_i^{(1)}\lambda_i^* & 0\\
					0& 1-2\sum_{i=1}^N x_i^{(2)}y_i^{(2)}\lambda_i^*
				\end{pmatrix},\\
				\mathbf{T}_{-1}^{[N]}&=\begin{pmatrix}
					0 & -2\sum_{i=1}^N x_i^{(1)}y_i^{(2)}\left(\lambda_i^*\right)^2 \\
					-2\sum_{i=1}^N x_i^{(2)}y_i^{(1)}\left(\lambda_i^*\right)^2
					& 0\end{pmatrix}.
			\end{split}
		\end{equation}
		Plugging the expansion \eqref{Eq-T-expansion-infty}-\eqref{Eq-T0-T1-entries} into \eqref{Eq-QN-Q} and making utilize of \eqref{Eq-Lax pair}, we arrive at 
		\begin{equation}\label{Eq-BT-infty-scalar}
			\left(1-2\sum_{i=1}^Nx_i^{(2)}y_i^{(2)}\lambda_i^{*}\right)u^{[N]}_{\infty}=-4\sum_{i=1}^Nx_i^{(1)}y_i^{(2)}\left(\lambda_i^{*}\right)^2+\left(1-2\sum_{i=1}^Nx_i^{(1)}y_i^{(1)}\lambda_i^{*}\right)u,
		\end{equation}
		by comparing the coefficients of $\lambda$ on both sides of \eqref{Eq-QN-Q}. 
		According to \eqref{Eq-T}, the inverse of the Darboux matrix could be written as 
		\begin{align}\label{Eq-TN-inverse}
			(\mathbf{T}^{[N]})^{-1}=\mathbb{I}_2+\sum_{i=1}^N\left(\frac{\textbf{A}_i^\dagger}{\lambda_i-\lambda}-\frac{\sigma_3\textbf{A}_i^\dagger\sigma_3}{\lambda_i+\lambda}\right)\lambda_i\lambda.
		\end{align}
		It follows from $\mathbf{T}^{[N]}\left(\mathbf{T}^{[N]}\right)^{-1}=\mathbb{I}_2$ that
		\begin{equation}\label{Eq-yi-kernal}
			\left\langle y_i\right|(\mathbf{T}^{[N]})^{-1}|_{\lambda=\lambda_i^*}=0,\quad 1\leq i\leq N,
		\end{equation}
		where we take account of the residue of $\textbf{T}^{[N]}\left(\mathbf{T}^{[N]}\right)^{-1}$ at $\lambda=\lambda_i^*$ is zero. On the other hand, according to the Darboux transformation theory, it holds that 
		\begin{equation}\label{Eq-phii-kernel}
			\varphi_i^{\dagger}(\mathbf{T}^{[N]})^{-1}|_{\lambda=\lambda_i^*}=0,\quad 1\leq i\leq N.
		\end{equation}
		Consequently, we could set $\left\langle y_i\right|=\varphi_i^{\dagger},1\leq i\leq N$ without loss of generality.
		Henceforth, combining with \eqref{Eq-TN-inverse}-\eqref{Eq-phii-kernel}, we arrive at
		\begin{equation}\label{Eq-DT-algebraic-equations}
			\begin{split}
				&\varphi_{i }^{(1)}+\sum_{j=1}^N \left(\frac{x_{j}^{(1)} \varphi_j^{\dagger} \varphi_i}{\lambda_j^{*}-\lambda_i}-\frac{x_{j}^{(1)} \varphi_j^{\dagger} \sigma_3 \varphi_i}{\lambda_j^{*}+\lambda_i}\right)\lambda_i\lambda_j^{*}=0, \\
				&\varphi_{i }^{(2)}+\sum_{j=1}^N \left(\frac{x_{j}^{(2)} \varphi_j^{\dagger} \varphi_i}{\lambda_j^{*}-\lambda_i}+\frac{x_{j}^{(2)} \varphi_j^{\dagger} \sigma_3 \varphi_i}{\lambda_j^{*}+\lambda_i}\right)\lambda_i\lambda_j^{*}=0,\\
			\end{split}
		\end{equation}
		for $1\leq i\leq N$.	From the above equations \eqref{Eq-DT-algebraic-equations}, we indicate that 
		\begin{equation}\label{Eq-x1-x2}
			\begin{split}
				x^{(1)}&=-\textbf{M}^{-1}\varphi^{(1)},\\
				x^{(2)}&=\left(\mathbf{M}^{\dagger}\right)^{-1}\varphi^{(2)},
			\end{split}
		\end{equation}
		respectively, where we denote $x^{(i)}=(x^{(i)}_1,x^{(i)}_2,\ldots,x^{(i)}_N)^T,i=1,2$. Henceforth, the equations \eqref{Eq-BT-infty-scalar} and \eqref{Eq-x1-x2} lead to \eqref{Eq-DT-infty}. Thus we complete the proof.
	\end{proof}

	The following theorem establishes the equivalence between $BT_0$ and $BT_\infty$.
	\begin{theorem}\label{thm:DT-equivalence}
		The two forms of Darboux-B\"acklund transformation \eqref{Eq-BT0-1}-\eqref{Eq-DT-infty} are equivalent. Specifically, it holds that $u^{[N]}_0=u^{[N]}_\infty$.
	\end{theorem}
	\begin{proof}
		Making use of the Lax pair \eqref{Eq-Lax pair}, we arrive at
		\begin{equation}\label{Eq-matrix-Lax-pair}
			\begin{split}
				\varphi^{(1)}_x&=-2\mathrm{i}\mathbf{\Lambda}^2\varphi^{(1)}+2\mathrm{i}\mathbf{\Lambda} u\varphi^{(2)},\\
				\varphi^{(2)}_x&=2\mathrm{i}\mathbf{\Lambda} u^*\varphi^{(1)}+2\mathrm{i}\mathbf{\Lambda}^2\varphi^{(2)}.
			\end{split}
		\end{equation}
		Moreover, combining \eqref{Eq-Lax pair} and  \eqref{Eq-BT0-2} gives rise to 
		\begin{equation}\label{Eq-derivative-M}
			\mathbf{M}_x=4\mathrm{i}\mathbf{\Lambda}\left(\mathbf{\Lambda}\varphi^{(1)}\big(\varphi^{(1)}\big)^\dagger-u\varphi^{(2)}\big(\varphi^{(1)}\big)^\dagger-\varphi^{(2)}\big(\varphi^{(2)}\big)^\dagger\mathbf{\Lambda}^\dagger\right)\mathbf{\Lambda}^\dagger.
		\end{equation}
		The symmetries \eqref{Eq-T-symmetry} and expansion \eqref{Eq-T-expansion-infty} imply that $\mathbf{T}_0^{[N]}$ is unitary and diagonal. Consequently, it is satisfied that
		\begin{equation}\label{Eq-entries-of-T0-modulus}
			|1+2\varphi^{(1)\dagger}\mathbf{\Lambda}^\dagger\mathbf{M}^{-1}\varphi^{(1)}|=|1-2\varphi^{(2)\dagger}\mathbf{\Lambda}^\dagger\left(\mathbf{M}^{\dagger}\right)^{-1}\varphi^{(2)}|=1.
		\end{equation}
		Furthermore, plugging  \eqref{Eq-T-expansion-infty}-\eqref{Eq-T0-T1-entries} into \eqref{Eq-QN-Q} and making utilize of \eqref{Eq-Lax pair} gives rise to
		\begin{equation}\label{Eq-matrix-BT-infty}
			\begin{split}
				\big(1+2\varphi^{(1)\dagger}\mathbf{\Lambda}^\dagger\mathbf{M}^{-1}\varphi^{(1)}\big)(u^{[N]}_\infty)^*&=4	\varphi^{(1)\dagger}\big(\mathbf{\Lambda}^\dagger\big)^2\big(\mathbf{M}^{\dagger}\big)^{-1}\varphi^{(2)}+\big(1-2\varphi^{(2)\dagger}\mathbf{\Lambda}^\dagger\big(\mathbf{M}^{\dagger}\big)^{-1}\varphi^{(2)}\big)u^*,\\
				\big(1-2\varphi^{(2)\dagger}\mathbf{\Lambda}^\dagger\big(\mathbf{M}^\dagger\big)^{-1}\varphi^{(2)}\big)u^{[N]}_\infty&=4	\varphi^{(2)\dagger}\big(\mathbf{\Lambda}^\dagger\big)^2\mathbf{M}^{-1}\varphi^{(1)}+\big(1+2\varphi^{(1)\dagger}\mathbf{\Lambda}^\dagger\mathbf{M}^{-1}\varphi^{(1)}\big)u.\\
			\end{split}
		\end{equation}
		Taking the complex conjugation of the second identity in \eqref{Eq-matrix-BT-infty} yields
		\begin{equation}\label{Eq-matrix-BT-infty0-2}
			\big(1-2\varphi^{(2)\dagger}\mathbf{M}^{-1}\mathbf{\Lambda}\varphi^{(2)}\big)(u^{[N]}_\infty)^*=4	\varphi^{(1)\dagger}\big(\mathbf{M}^{\dagger}\big)^{-1}\mathbf{\Lambda}^2\varphi^{(2)}+\big(1+2\varphi^{(1)\dagger}\big(\mathbf{M}^{\dagger}\big)^{-1}\mathbf{\Lambda}\varphi^{(1)}\big)u^*.
		\end{equation}
		Using \eqref{Eq-entries-of-T0-modulus}, we multiply the first identity in \eqref{Eq-matrix-BT-infty} by $1+2\varphi^{(1)\dagger}\big(\mathbf{M}^{\dagger}\big)^{-1}\mathbf{\Lambda}\varphi^{(1)}$ and \eqref{Eq-matrix-BT-infty0-2} by $1-2\varphi^{(2)\dagger}\mathbf{\Lambda}^\dagger\big(\mathbf{M}^{\dagger}\big)^{-1}\varphi^{(2)}$ respectively, then subtract the resulting identities and take the complex conjugation to obtain 
		\begin{equation}\label{Eq-identity}
			\varphi^{(2)\dagger}\mathbf{M}^{-1}\mathbf{\Lambda}^2\varphi^{(1)}\big(1+2\varphi^{(1)\dagger}\mathbf{\Lambda}^\dagger\mathbf{M}^{-1}\varphi^{(1)}\big)=\varphi^{(2)\dagger}(\mathbf{\Lambda}^{\dagger})^2\mathbf{M}^{-1}\varphi^{(1)}\big(1-2\varphi^{(2)\dagger}\mathbf{M}^{-1}\mathbf{\Lambda}\varphi^{(2)}\big).
		\end{equation}
		Combining \eqref{Eq-BT0-1}, \eqref{Eq-matrix-Lax-pair}-\eqref{Eq-entries-of-T0-modulus} and \eqref{Eq-identity},  we obtain
		\begin{equation}\label{Eq-BT0-simplication}
			\begin{split}
				&\quad u^{[N]}_0\\
				&=u+\mathrm{i}\varphi^{(2)\dagger}_x\mathbf{M}^{-1}\varphi^{(1)}-\mathrm{i}\varphi^{(2)\dagger}\mathbf{M}^{-1}\mathbf{M}_x\mathbf{M}^{-1}\varphi^{(1)}+\mathrm{i}\varphi^{(2)\dagger}\mathbf{M}^{-1}\varphi^{(1)}_x\\
				&= \big(1+2\varphi^{(1)\dagger}\mathbf{\Lambda}^\dagger\mathbf{M}^{-1}\varphi^{(1)}\big)\big(1-2\varphi^{(2)\dagger}\mathbf{M}^{-1}\mathbf{\Lambda}\varphi^{(2)}\big)u\\
				&\quad+2\varphi^{(2)\dagger}\big(\mathbf{\Lambda^{\dagger}}\big)^2\mathbf{M}^{-1}\varphi^{(1)}\big(1-2\varphi^{(2)\dagger}\mathbf{M}^{-1}\mathbf{\Lambda}\varphi^{(2)}\big)\\
				&\quad+2\varphi^{(2)\dagger}\mathbf{M}^{-1}\mathbf{\Lambda}^2\varphi^{(1)}\big(1+2\varphi^{(1)\dagger}\mathbf{\Lambda}^\dagger\mathbf{M}^{-1}\varphi^{(1)}\big)\\
				&\overset{\mathclap{\eqref{Eq-entries-of-T0-modulus}}}{=}\frac{1+2\varphi^{(1)\dagger}\mathbf{\Lambda}^\dagger\mathbf{M}^{-1}\varphi^{(1)}}{1-2\varphi^{(2)\dagger}\mathbf{\Lambda}^\dagger\left(\mathbf{M}^\dagger\right)^{-1}\varphi^{(2)}}u+\frac{2\varphi^{(2)\dagger}\big(\mathbf{\Lambda^{\dagger}}\big)^2\mathbf{M}^{-1}\varphi^{(1)}}{1-2\varphi^{(2)\dagger}\mathbf{\Lambda}^\dagger\left(\mathbf{M}^\dagger\right)^{-1}\varphi^{(2)}}+2\varphi^{(2)\dagger}\mathbf{M}^{-1}\mathbf{\Lambda}^2\varphi^{(1)}\big(1+2\varphi^{(1)\dagger}\mathbf{\Lambda}^\dagger\mathbf{M}^{-1}\varphi^{(1)}\big)\\
				&\overset{\mathclap{\eqref{Eq-identity}}}{=}\frac{1+2\varphi^{(1)\dagger}\mathbf{\Lambda}^\dagger\mathbf{M}^{-1}\varphi^{(1)}}{1-2\varphi^{(2)\dagger}\mathbf{\Lambda}^\dagger\left(\mathbf{M}^\dagger\right)^{-1}\varphi^{(2)}}u+\frac{2\varphi^{(2)\dagger}\big(\mathbf{\Lambda^{\dagger}}\big)^2\mathbf{M}^{-1}\varphi^{(1)}}{1-2\varphi^{(2)\dagger}\mathbf{\Lambda}^\dagger\left(\mathbf{M}^\dagger\right)^{-1}\varphi^{(2)}}+2\varphi^{(2)\dagger}(\mathbf{\Lambda}^{\dagger})^2\mathbf{M}^{-1}\varphi^{(1)}\big(1-2\varphi^{(2)\dagger}\mathbf{M}^{-1}\mathbf{\Lambda}\varphi^{(2)}\big)\\
				&\overset{\mathclap{\eqref{Eq-entries-of-T0-modulus}}}{=}\frac{1+2\varphi^{(1)\dagger}\mathbf{\Lambda}^\dagger\mathbf{M}^{-1}\varphi^{(1)}}{1-2\varphi^{(2)\dagger}\mathbf{\Lambda}^\dagger\left(\mathbf{M}^\dagger\right)^{-1}\varphi^{(2)}}u+\frac{4\varphi^{(2)\dagger}\big(\mathbf{\Lambda^{\dagger}}\big)^2\mathbf{M}^{-1}\varphi^{(1)}}{1-2\varphi^{(2)\dagger}\mathbf{\Lambda}^\dagger\left(\mathbf{M}^\dagger\right)^{-1}\varphi^{(2)}}=u^{[N]}_\infty.
			\end{split}
		\end{equation}
		Henceforth, we complete the proof.
	\end{proof}
	Combining with \eqref{Eq-DT-infty} and \eqref{Eq-entries-of-T0-modulus}, the modulus of the solution $u^{[N]}$
	\eqref{Eq-BT0-1} can be evaluated as
	\begin{align}\label{Eq-DT-modulus}
		|u^{[N]}|=\left|u+\frac{4\varphi^{(2)\dagger}\big(\mathbf{\Lambda^{\dagger}}\big)^2\mathbf{M}^{-1}\varphi^{(1)}}{1+2\varphi^{(1)\dagger}\mathbf{\Lambda}^\dagger\mathbf{M}^{-1}\varphi^{(1)}}\right|,
	\end{align}
	where we consistently denote both $u^{[N]}_0$ and $u^{[N]}_\infty$ by $u^{[N]}$ by virtue of the equivalence between $BT_0$ and $BT_\infty$. Throughout this work, $u^{[N]}$ represents any solution generated through the $N$-fold Darboux transformation, while the notation $u_N$ is reserved specifically for solutions arising from the elliptic function solutions under the $N$-fold Darboux-B\"acklund transformation.

	\section{The elliptic function solutions to the DNLS equation and its Lax pair}
In this section,
we derive elliptic function solutions to the DNLS equation \eqref{Eq-DNLS-equation}, 
then construct solutions to the associated Lax pair \eqref{Eq-Lax pair}. 
Unlike existing solutions \cite{Ling-NLS,LS-mKdV-solution}, 
these elliptic function solutions are explicitly expressed via Weierstrass sigma functions rather than Jacobi theta functions. 
Compared with the elliptic function solutions in \cite{effective-integrationeffective-integration}, 
we further observe that distinct types of squared moduli can be written uniformly if complex elliptic moduli are permitted. 
Moreover, we verify that the elliptic function solutions admit representations in terms of  Weierstrass functions of four parameters. 
This extends previous results \cite{Ling-NLS,LS-mKdV-solution}, 
where elliptic function solutions to AKNS equations were expressed via Jacobi elliptic functions of three parameters $\alpha,k$ and $l$. 
We verify that these four parameters are generically independent.
	
	\subsection{The squared moduli of the elliptic function solutions to the DNLS equation}
	We begin with the stationary zero-curvature equations
	\begin{align}\label{Eq-stationary-zero-curvature-equation}
		\textbf{L}_x(x,t;\lambda)=[\textbf{U}(\textbf{Q};\lambda),\textbf{L}(x,t;\lambda)],\quad     \textbf{L}_t(x,t;\lambda)=[\textbf{V}(\textbf{Q};\lambda),\textbf{L}(x,t;\lambda)].
	\end{align}
	The compatibility condition $\textbf{L}_{xt}(x,t;\lambda)=\textbf{L}_{tx}(x,t;\lambda)$ for \eqref{Eq-stationary-zero-curvature-equation}, enforced by the Jacobi identity, yields the DNLS equation \eqref{Eq-DNLS-equation}. Dependent variables are omitted where unambiguous in the following. Previous literature \cite{Ling-NLS} indicates that computing genus-one solutions for the NLS equation through coefficient comparison of powers of $\lambda$ entails tedious computations. To streamline this procedure, we observe that taking \(\mathbf{L}\) as a linear combination of the Lax pair matrices \(\mathbf{U}\) and \(\mathbf{V}\) suffices to obtain genus-one solutions for \eqref{Eq-DNLS-equation}. Specifically, we employ the ansatz
	\begin{align}\label{Eq-L-ansatz}
		\textbf{L}=-\frac{\mathrm{i}}{8}\textbf{V}+\frac{\mathrm{i}s_1}{4}\textbf{U}+\alpha_4\sigma_3,
	\end{align}
	with real parameters $s_1,\alpha_4\in\mathbb{R}$. Following the Lax pair  \eqref{Eq-Lax pair} and ansatz \eqref{Eq-L-ansatz}, the matrix $\mathbf{L}$ admits the equivalent representation
	\begin{equation}\label{Eq-Lax-matrix}
		\mathbf{L} = \begin{pmatrix}
			-\mathrm{i}f & g \\
			h &  \mathrm{i}f \\
		\end{pmatrix},
	\end{equation}
	whose components are specified by
	\begin{equation}\label{Eq-fgh}
		\begin{split}
			f &= -\mathrm{i}\lambda^4 + \frac{\mathrm{i}(s_1+\nu)}{2}\lambda^2 + \mathrm{i}\alpha_4, \\    
			g &= u\lambda^3 + \tfrac{\mathrm{i}}{4}u_x\lambda - \frac{s_1+\nu}{2}u\lambda=u\lambda(\lambda^2-\mu), \\
			h &= u^*\lambda^3 - \tfrac{\mathrm{i}}{4}u_x^*\lambda - \frac{s_1+\nu}{2}u^*\lambda=u^*\lambda(\lambda^2-\mu^*).
		\end{split}
	\end{equation}
	In these expressions, $\nu$
	defines the squared modulus of $u$ while the auxiliary spectral parameter is given by
	\begin{equation}\label{Eq-mu}
		\mu=-\frac{\mathrm{i}}{4}\frac{u_x}{u}+\frac{s_1+\nu}{2}.
	\end{equation}
	Introducing the polynomial $P(\lambda):=\det(\mathbf{L})$, 
	we reformulate it through \eqref{Eq-Lax-matrix}-\eqref{Eq-fgh} as
	\begin{equation}\label{Eq-P-ansatz}
		\begin{split}
			P(\lambda)&=-\left(\lambda^4-\frac{s_1+\nu}{2}\lambda^2-\alpha_4\right)^2-\lambda^2\nu(\lambda^2-\mu)(\lambda^2-\mu^*)\\
			&:=-\lambda^8+s_1 \lambda^6-s_2 \lambda^4+s_3 \lambda^2-s_4.
		\end{split}
	\end{equation}
	Using the fundamental theory of linear algebra and the spatial part of the stationary zero curvature equations \eqref{Eq-stationary-zero-curvature-equation}, we deduce that
	\begin{equation}
		\left( \det(\mathbf{L})\right)_x= \det(\mathbf{L}) \cdot \mathrm{tr}\left( \mathbf{L}^{-1}\mathbf{L}_x \right) = \det(\mathbf{L}) \cdot \mathrm{tr}\left( \mathbf{L}^{-1}(\mathbf{U}\mathbf{L} - \mathbf{L}\mathbf{U}) \right) = \det(\mathbf{L}) \cdot \mathrm{tr}\left( \mathbf{L}^{-1}\mathbf{U}\mathbf{L} - \mathbf{U} \right) = 0.
	\end{equation}
	Similarly we also have $	\left( \det(\mathbf{L})\right)_t=0$. Henceforth, the polynomial $P(\lambda)$ is independent of $x$ and $t$. 
	The following proposition reveals that $\nu$ is a travelling wave with respect to the variable $\xi=x+2s_1t$.
	\begin{prop}\label{Prop-ue}
		The squared modulus $\nu$ of the elliptic function solutions to DNLS equation \eqref{Eq-DNLS-equation} satisfies
		\begin{align}\label{Eq-nu-travelling-2}
			\nu_t = 2s_1\nu_x.
		\end{align}
	\end{prop}
	\begin{proof}
		We substitute \eqref{Eq-L-ansatz} into the spatial part of the stationary zero curvature equation \eqref{Eq-stationary-zero-curvature-equation} and obtain
		\begin{equation}
			-\frac{\mathrm{i}}{8}\mathbf{V}_x+\frac{\mathrm{i}s_1}{4}\mathbf{U}_x=-\frac{\mathrm{i}}{8}[\mathbf{U},\mathbf{V}]+[\mathbf{U},\alpha_4\sigma_3].
		\end{equation}
		Using \eqref{Eq-compatibility-condition}, the above equation can be converted to 
		\begin{align}\label{Eq-prop1-1}
			-\mathrm{i}\mathbf{U}_t+2\mathrm{i}s_1\mathbf{U}_x+32\lambda\alpha_4\sigma_3\mathbf{Q}=0.
		\end{align}
		From \eqref{Eq-Lax pair} we derive
		\begin{align}\label{Eq-U-diag-off}
			\mathbf{U}^{\mathrm{diag}}=-2\mathrm{i}\sigma_3\lambda^2,\quad \mathbf{U}^{\mathrm{off}}=2\mathbf{Q}\lambda.
		\end{align}
		Substituting \eqref{Eq-U-diag-off} into \eqref{Eq-prop1-1}, the off-diagonal part gives rise to
		\begin{align}\label{Eq-u-travelling}
			\textbf{Q}_t = -16\mathrm{i}\alpha_4\sigma_3\textbf{Q} + 2s_1\textbf{Q}_x.
		\end{align}
		Rewrite the above matrix equation in its entries, we arrive at 
		\begin{align}\label{Eq-u-travelling-2}
			u_t-2s_1 u_x+16\mathrm{i}\alpha_4 u=0.
		\end{align}
		Multiplying the above equation with $u^*$ and summing the resulting equation with its complex conjugation gives rise to \eqref{Eq-nu-travelling-2}. Henceforth we complete the proof.
	\end{proof}
	To establish elliptic function solutions for the DNLS equation, we examine the elliptic equation governing the squared moduli. Based on this, 
	we could express the squared moduli via elliptic functions.

	\begin{prop}
		The squared modulus $\nu$ satisfies the differential equation
		\begin{equation}\label{Eq-nu-x}
			\nu_\xi = \sqrt{-R(\nu)},
		\end{equation}
		where $R(\nu)$ is the quartic polynomial
		\begin{equation}\label{Eq-R}
			\begin{split}
				R(\nu)
				&= \nu^4 + 4s_1\nu^3 + (6s_1^2 - 8s_2 + 48\alpha_4)\nu^2\\
				&\quad - (-4s_1^3+16s_1s_2 - 64s_3 - 32s_1\alpha_4)\nu+\left( -s_1^2+4s_2+ 8\alpha_4 \right)^2.
			\end{split}
		\end{equation}
	\end{prop}
	\begin{proof}
		It follows from  \eqref{Eq-P-ansatz} that
		\begin{equation}\label{Eq-si}
			s_2=\left(\frac{s_1+\nu}{2}\right)^2-2\alpha_4-\nu\left(\mu+\mu^*\right),\quad
			s_3=-(s_1+\nu) \alpha_4-\nu |\mu|^2 ,
		\end{equation}
		yielding
		\begin{align}\label{Eq-mu-plus-mubar }
			\mu+\mu^*=-\frac{s_2-(\frac{s_1+\nu}{2})^2+2\alpha_4}{\nu},\quad |\mu|^2=\frac{s_3+(s_1+\nu)\alpha_4}{-\nu}.
		\end{align}
		From \eqref{Eq-mu-plus-mubar }, $\mu$ and $\mu^*$ are roots of a quadratic equation. Thus, without loss of generality,
		\begin{align}\label{Eq-mu-nu}
			\mu= -\frac{1}{8\nu}\big(4s_2+8\alpha_4-(\nu+s_1)^2+\mathrm{i}\sqrt{-R(\nu)}\big).
		\end{align}
		From \eqref{Eq-mu}, we obtain
		\begin{equation}\label{Eq-Im-mu}
			\mu - \mu^* = -\frac{\mathrm{i}}{4}\frac{\nu_\xi}{\nu}.
		\end{equation}
		Combining \eqref{Eq-mu-nu} with \eqref{Eq-Im-mu} yields \eqref{Eq-nu-x}, completing the proof.
	\end{proof}

We assume that 
the roots of the polynomials $P(\lambda)$ and $R(\nu)$ are $\pm\lambda^{(i)}$ and $\nu_i$ for $i=1,2,3,4$ respectively. 
The following proposition connects the roots of $P(\lambda)$ with those of $R(\nu)$. 
Using these connections, we classify the squared moduli into three types.
	\begin{prop}\label{Prop-connection-P-R}
		The roots of the polynomials $R(\nu)$ and $P(\lambda)$ satisfy the relations
		\begin{equation}\label{Eq-connection-P-R+}
			\begin{aligned}
				\nu_1 &= -\left(\lambda^{(1)} + \lambda^{(2)} + \lambda^{(3)} - \lambda^{(4)}\right)^2, \\ 
				\nu_2 &= -\left(\lambda^{(1)} + \lambda^{(2)} - \lambda^{(3)} + \lambda^{(4)}\right)^2, \\
				\nu_3 &= -\left(\lambda^{(1)} - \lambda^{(2)} + \lambda^{(3)} + \lambda^{(4)}\right)^2, \\ 
				\nu_4 &= -\left(-\lambda^{(1)} + \lambda^{(2)} + \lambda^{(3)} + \lambda^{(4)}\right)^2,
			\end{aligned}
		\end{equation}
		for $\alpha_4 = \prod_{i=1}^4 \lambda^{(i)}$, and
		\begin{equation}\label{Eq-connection-P-R-}
			\begin{aligned}
				\nu_1 &= -\left(\lambda^{(1)} + \lambda^{(2)} + \lambda^{(3)} + \lambda^{(4)}\right)^2, \\
				\nu_2 &= -\left(\lambda^{(1)} + \lambda^{(2)} - \lambda^{(3)} - \lambda^{(4)}\right)^2, \\
				\nu_3 &= -\left(\lambda^{(1)} - \lambda^{(2)} + \lambda^{(3)} - \lambda^{(4)}\right)^2, \\
				\nu_4 &= -\left(\lambda^{(1)} - \lambda^{(2)} - \lambda^{(3)} + \lambda^{(4)}\right)^2,
			\end{aligned}
		\end{equation}
		for $\alpha_4 = -\prod_{i=1}^4 \lambda^{(i)}$.
	\end{prop}
	\begin{proof}
		Combining \eqref{Eq-P-ansatz} with \eqref{Eq-mu-nu}, we deduce that the polynomial $P(\lambda)$ can be factored as 
		\begin{equation}
			\begin{split}
				P(\lambda)=-\left(\lambda^4-\frac{s_1+\nu_i}{2}\lambda^2-\alpha_4\right)^2+(\mathrm{i}\sqrt{\nu_i})^2\lambda^2(\lambda^2-\mu_i)^2=-P_i(\lambda)P_i(-\lambda),\quad i=1,2,3,4,
			\end{split}
		\end{equation}
		where 
		\begin{equation}
			\begin{split}
				\mu_i= -\frac{1}{8\nu_i}\big(4s_2+8\alpha_4-(\nu_i+s_1)^2\big),\quad 
				P_i(\lambda)=\lambda^4+\mathrm{i}\sqrt{\nu_i}\lambda^3-\frac{s_1+\nu_i}{2}\lambda^2-\mathrm{i}\mu_i\sqrt{\nu_i}\lambda-\alpha_4.
			\end{split}
		\end{equation}
		When $\alpha_4=\prod_{i=1}^4\lambda^{(i)}$, using the Vieta's theorem, the coefficients of $\lambda^3$ and $\lambda^0$ in $P_i(\lambda)$ gives rise to \eqref{Eq-connection-P-R+}. Similarly we can derive \eqref{Eq-connection-P-R-}. Henceforth, we finish the proof.
	\end{proof}
	
	To ensure \eqref{Eq-nu-x} admits non-degenerate, bounded and non-negative solutions, certain constraints must be imposed on $\nu_i$ (and consequently on  $\lambda^{(i)}$). Based on Proposition \ref{Prop-connection-P-R}, we identify three admissible types:
	\begin{itemize}
		\item[-]Type 1: For  pairs $\lambda^{(1)} = -\lambda^{(2)*} = a - b\mathrm{i}$ and $\lambda^{(3)} = -\lambda^{(4)*} = -c - d\mathrm{i}$ with $a,b,c,d >0$, the relation \eqref{Eq-connection-P-R-} produces four distinct real roots of $R(\nu)$ \eqref{Eq-R} satisfying $\nu_1 >\nu_2 \geq  0 \geq \nu_3 >\nu_4$. For pairs $\lambda^{(1)} = \lambda^{(2)*} = a - b\mathrm{i}$ and $\lambda^{(3)} = -\lambda^{(4)*} = -c - d\mathrm{i}$ with $a,b,c,d>0$, the equation \eqref{Eq-connection-P-R+} generate the same $\nu_i,i=1,2,3,4$ after a proper rearrangement. 
		
		\item[-]Type 2: When all $\lambda^{(i)}$ are purely imaginary ($\lambda^{(i)} = \mathrm{i}\tau_i$) with ordered parameters $\tau_1 > \tau_2> \tau_3 > \tau_4 \geq 0$, both \eqref{Eq-connection-P-R+} and \eqref{Eq-connection-P-R-} generate four distinct real roots of $R(\nu)$ \eqref{Eq-R} with $\nu_1> \nu_2> \nu_3> \nu_4 \geq 0$.
		
		\item[-]Type 3: When $\lambda^{(1,2)}=\pm a+\mathrm{i}b,\lambda^{(3)}=\mathrm{i}c,\lambda^{(4)}=\mathrm{i}d$ with $a,b,c>0$ and $d\geq 0$, \eqref{Eq-connection-P-R+} provides $\nu_1>\nu_2\geq 0$ and $\nu_3=\nu_4^*$ when $c>d$ while \eqref{Eq-connection-P-R-} always provides $\nu_1>\nu_2\geq 0$ and $\nu_3=\nu_4^*$.
		
	\end{itemize}

	\subsection{The parameterization of the squared moduli}
	In the subsection, we aim to parameterize the three types of squared moduli  uniformly. As is revealed by \eqref{Eq-nu-x}, the squared moduli $\nu(\xi)$ oscillate in the intervals $[\nu_2,\nu_1]$ or $[\nu_4,\nu_3]$. Without loss of generality, we can set the value of $\nu(0)$ of the squared moduli oscillating within $[\nu_2,\nu_1]$ and $[\nu_4,\nu_3]$ as $\nu_1$ and $\nu_4$ respectively. We consider elliptic function solutions with squared modulus with $\nu(0)=\nu_1$ first.
	We introduce an elliptic curve
	\begin{align}\label{Eq-K1} \mathcal{K}_1:=\{(\Lambda_1,Y_1)|Y_1^2=R(\Lambda_1)\}.
	\end{align}
	According to \eqref{Eq-nu-x}, the point  $(\nu,\mathrm{i}\nu_\xi)$ lies on the curve $\mathcal{K}_{1}$. Based on the theory of elliptic curves, $\mathcal{K}_{1}$ can be associated with a normalized Weierstrass elliptic curve:
	\begin{align}\label{Eq-Kn}
		\mathcal{K}_{(n)}:=\{(\Lambda,Y)|Y^2=4\left(\Lambda-e_1\right)\left(\Lambda-e_2\right)\left(\Lambda-e_3\right)\},
	\end{align}
	via a birational mapping
	\begin{equation}\label{Eq-birational-mapping}
		\big(\Lambda_1,	Y_1\big)=\left(\nu_1 + \frac{1}{b\Lambda+c},\frac{bY}{(b\Lambda+c)^2}\right),
	\end{equation}
	where
	\begin{equation}\label{Eq-ei-and-vi}
		\begin{split}
			e_1&=-\frac{1}{12}(\nu_1+\nu_2)(\nu_3+\nu_4)+\frac{1}{6}(\nu_1\nu_2+\nu_3\nu_4),\\
			e_2&=-\frac{1}{12}(\nu_1+\nu_3)(\nu_2+\nu_4)+\frac{1}{6}(\nu_1\nu_3+\nu_2\nu_4),\\
			e_3&=-\frac{1}{12}(\nu_1+\nu_4)(\nu_2+\nu_3)+\frac{1}{6}(\nu_1\nu_4+\nu_2\nu_3),
		\end{split}
	\end{equation}
	and
	\begin{equation}\label{Eq-b-c}
		\begin{split}
			b&=-\frac{4}{(\nu_1-\nu_2)(\nu_1-\nu_3)(\nu_1-\nu_4)},\\
			c&=-\frac{1}{3}\left(\frac{1}{\nu_1-\nu_2}+\frac{1}{\nu_1-\nu_3}+\frac{1}{\nu_1-\nu_4}\right).
		\end{split}
	\end{equation}
	The curve $\mathcal{K}_{(n)}$ could be parameterized by $(\Lambda,Y)=\big(\wp(\xi),\wp'(\xi)\big)$, where $\wp(\xi)$ satisfies the normalized Weierstrass elliptic equation
	\begin{align}\label{Eq-Weierstrass-elliptic-equation}
		\big(\wp'(\xi)\big)^2=4\big(\wp(\xi)-e_1\big)\big(\wp(\xi)-e_2\big)\big(\wp(\xi)-e_3\big).
	\end{align}
	Recall that the point $(\nu,\mathrm{i}\nu_\xi)$ lies on the curve $\mathcal{K}_{1}$. Applying \eqref{Eq-birational-mapping} and \eqref{Eq-b-c} yields
	\begin{equation}\label{Eq-nu-parameterization}
		\nu(\xi)=\nu_1\frac{\wp(\xi)-\wp(\rho)}{\wp(\xi)-\wp(\kappa)},
	\end{equation}
	as a solution to \eqref{Eq-nu-x}, where $\kappa$ and $\rho$ satisfy
	\begin{equation}\label{Eq-kappa-rho}
		\begin{split}
			\wp(\kappa)&=-\frac{1}{12}\sum_{i=2}^4\prod_{j=2,j\neq i}^4(\nu_1-\nu_j),\\
			\wp(\rho)	&=-\frac{\sum_{i=2}^4\nu_i\prod_{j=2,j\neq i}^4(\nu_1-\nu_j)}{12\nu_1}. 
		\end{split}
	\end{equation}
	We adapt the settings of the Weierstrass elliptic functions in the Appendix A. Then the parameters $\kappa$ and $\rho$ can be selected within the following periodic parallelogram
	\begin{align}\label{Eq-S}
		S_0:=\left\{w\in\mathbb{C}\Bigg|-\omega_1< \Re\big(w\big)\leq \omega_1,\Im(\omega_3)\leq\Im\big(w\big)< -\Im(\omega_3) \right\},
	\end{align}
	without loss of generality.
	There exists two values of $\kappa$ and $\rho$ satisfying \eqref{Eq-kappa-rho} in $S_0$ respectively.  
	Taking $\xi=\kappa$ into \eqref{Eq-nu-x} and using \eqref{Eq-nu-parameterization}, we infer that 
	\begin{equation}
		\nu_1^2=-\frac{\big(\wp'(\kappa)\big)^2}{\big(\wp(\rho)-\wp(\kappa)\big)^2}.
	\end{equation}
	A comparison between \eqref{Eq-birational-mapping} and \eqref{Eq-nu-parameterization} yields
	\begin{equation}\label{Eq-rho-minus-kappa}
		\wp(\rho)=\wp(\kappa)-\frac{1}{b\nu_1},
	\end{equation}
	which further implies
	\begin{equation}\label{Eq-rho-greater-kappa}
		\wp(\rho)>\wp(\kappa).
	\end{equation}
	Henceforth, by selecting $\kappa\in S_0$ properly such that $\kappa$ satisfies \eqref{Eq-kappa-rho} with $\mathrm{i}\wp'(\kappa)>0$, we arrive at
	\begin{equation}\label{Eq-nu1-parameterization}
		\nu_1=	\frac{\mathrm{i}\wp'(\kappa)}{\wp(\rho)-\wp(\kappa)}.
	\end{equation}
	Such $\kappa$ is uniquely determined in $S_0$. Define 
	\begin{equation}\label{Eq-Constant-C}
		C:=-s_1^2+4s_2+8\alpha_4.
	\end{equation}
	Substituting $\xi=\rho$ into \eqref{Eq-nu-x} and applying \eqref{Eq-nu-parameterization}, we derive	
	\begin{align}\label{Eq-Csqr}
		C^2=	\nu_1\nu_2\nu_3\nu_4=\frac{\big(\wp'(\rho)\big)^2\big(\wp'(\kappa)\big)^2}{\big(\wp(\rho)-\wp(\kappa)\big)^4}.
	\end{align}
	Thus, $\rho$ may be chosen such that 
	\begin{align}\label{Eq-C-para}
		C=	\frac{\wp'(\rho)\wp'(\kappa)}{\big(\wp(\rho)-\wp(\kappa)\big)^2}.	
	\end{align}
	This uniquely determines $\rho$ in $S_0$. To obtain the squared moduli with  $\nu(0)=\nu_4$, it suffices to exchange $\nu_{i},i=1,2$ with $\nu_{5-i}$ in the above constructions. Moreover, $\kappa$ and $\rho$ are uniquely determined in $S_0$ in a similar way. 
	The above discussion yields the uniform explicit form for the squared moduli $\nu(\xi)$:
	\begin{equation}\label{Eq-uniform-expression-for-the-squared-modulus}
		\nu(\xi)=\nu_0\frac{\wp(\xi)-\wp(\rho)}{\wp(\xi)-\wp(\kappa)},
	\end{equation}
	where 
	\begin{equation}\label{Eq-nu-0}
		\nu_0=\nu(0)=\frac{\mathrm{i}\wp'(\kappa)}{\wp(\rho)-\wp(\kappa)}.
	\end{equation}

	Finally, we discuss the relations between the square moduli presented here and those given in \cite{CP-DNLS}. For the first two types of squared moduli, using \eqref{Eq-ei-and-vi},\eqref{Eq-Weierstrass-elliptic-equation}-\eqref{Eq-kappa-rho},\eqref{Eq-nu1-parameterization} and \eqref{Eq-relation-wp-sn}, we could express the squared moduli with $\nu_0=\nu_1$ as 
	\begin{equation}\label{Eq-snoidal-representation}
		\nu(\xi)=\nu_4+\frac{\left(\nu_1-\nu_4\right)\left(\nu_2-\nu_4\right)}{\left(\nu_2-\nu_4\right)+\left(\nu_1-\nu_2\right) \operatorname{sn}^2(\alpha_0 \xi ; k_0)},
	\end{equation}
	where
	\begin{equation}\label{Eq-alpha0-k0}
		\alpha_0=\frac{1}{2} \sqrt{\left(\nu_1-\nu_3\right)\left(\nu_2-\nu_4\right)}, \quad
		k_0=\frac{\sqrt{\left(\nu_1-\nu_2\right)\left(\nu_3-\nu_4\right)}}{\sqrt{\left(\nu_1-\nu_3\right)\left(\nu_2-\nu_4\right)}},	
	\end{equation}
	by direct computation. This representation \eqref{Eq-snoidal-representation}-\eqref{Eq-alpha0-k0} coincides with equations (29)-(30) in \cite{CP-DNLS}.
	For the square moduli of Type $3$, to circumvent complex elliptic moduli, we verify  
	\begin{equation}\label{Eq-cnoidal-representation}
		\nu(\xi)=\nu_1+\frac{\left(\nu_2-\nu_1\right)\left(1-\operatorname{cn}(\theta_2\xi ; \theta_3)\right)}{1+\theta_1+(\theta_1-1) \operatorname{cn}(\theta_2 \xi ; \theta_3)},
	\end{equation}
	where 
	\begin{equation}\label{Eq-theta-123}
		\begin{split}
			\theta_1 & =\frac{\sqrt{\left(\nu_2-\Re(\nu_3)\right)^2+\Im^2(\nu_3)}}{\sqrt{\left(\nu_1-\Re(\nu_3)\right)^2+\Im^2(\nu_3)}}, \,\,\,
			\theta_2  =\sqrt[4]{\left(\left(\nu_1-\Re(\nu_3)\right)^2+\Im^2(\nu_3)\right)\left(\left(\nu_2-\Re(\nu_3)\right)^2+\Im^2(\nu_3)\right)},\\ 2\theta_3^2&=1-\frac{\left(\nu_1-\Re(\nu_3)\right)\left(\nu_2-\Re(\nu_3)\right)+\Im^2(\nu_3)}{\sqrt{\left(\left(\nu_1-\Re(\nu_3)\right)^2+\Im^2(\nu_3)\right)\left(\left(\nu_2-\Re(\nu_3)\right)^2+\Im^2(\nu_3)\right)}},
		\end{split}
	\end{equation}
	using properties of Jacobi elliptic functions. Equations \eqref{Eq-cnoidal-representation}-\eqref{Eq-theta-123} match equations (39)-(41) in \cite{CP-DNLS}. Thus, the squared moduli with $\nu_0=\nu_1$ in \cite{CP-DNLS} are equivalent to those presented here. Similarly, the squared moduli with $\nu_0=\nu_4$ can be derived directly by exchanging $\nu_i,i=1,2$ with $\nu_{5-i}$ in the above settings, which are also compatible with the squared moduli therein \cite{CP-DNLS}.
	
	\subsection{The elliptic function solutions to the DNLS equation}
	In this subsection, based on the unified expression of the squared moduli \eqref{Eq-uniform-expression-for-the-squared-modulus}, we derive Weierstrass function expressions for the auxiliary spectrum. Then, through integration, we obtain elliptic function solutions to the DNLS equation.
	\begin{prop}\label{Prop-mu}
		The auxiliary spectrum $\mu(\xi)$ could be  written in terms of Weierstrass zeta functions as 
		\begin{align}\label{Eq-mu-parameterization}
			\mu(\xi)=\frac{\mathrm{i}}{4}\Big(\zeta(\kappa+\xi) - \zeta(\xi+\rho) + \zeta(\rho+\kappa) - \zeta(2\kappa)\Big).	
		\end{align}
	\end{prop}
	\begin{proof}
		Consider the squared moduli with $\nu_0=\nu_1$.  Combining \eqref{Eq-birational-mapping} and substituting $\xi=\omega_i,i=1,2,3$ into \eqref{Eq-nu-x} yields
		\begin{equation}
			\prod_{j=1}^4\left(\nu_1\frac{\wp(\omega_i)-\wp(\rho)}{\wp(\omega_i)-\wp(\kappa)}-\nu_j\right)=0,\quad i=1,2,3.
		\end{equation}
		The three equations above are each composed of four factors. Since the $\nu_i$, $i=1,2,3,4$ are distinct, each equation has precisely one vanishing factor. Specifically, we obtain that
		\begin{equation}\label{Eq-nu234-parameterization}
			\{\nu_i,i=2,3,4\}=\left\{\frac{\mathrm{i}\wp'(\kappa)}{\wp(\rho)-\wp(\kappa)}+\frac{\mathrm{i}\wp'(\kappa)}{\wp(\kappa)-\wp(\omega_i)}, i=1,2,3\right\}.
		\end{equation}
 For squared moduli with $\nu_0=\nu_4$, analogous arguments apply via exchange of $\nu_i,i=1,2$ and $\nu_{5-i}$. Therefore, we always have 
		\begin{equation}
			\sum_{i=1}^4\nu_i=\frac{4\mathrm{i}\wp'(\kappa)}{\wp(\rho)-\wp(\kappa)}+\sum_{i=1}^3\frac{\mathrm{i}\wp'(\kappa)}{\wp(\kappa)-\wp(\omega_i)}.
		\end{equation}
		Moreover, \eqref{Eq-Weierstrass-elliptic-equation} implies
		\begin{equation}\label{Eq-log-deri}
			\sum_{i=1}^3\frac{\wp'(\kappa)}{\wp(\kappa)-\wp(\omega_i)}=	\frac{2\wp''(\kappa)}{\wp'(\kappa)}.
		\end{equation}
		Given the form of \eqref{Eq-mu-nu}, it remains to represent $s_1$ in terms of Weierstrass functions to obtain the representation ofe $\mu$. 
		Combining \eqref{Eq-R}, \eqref{Eq-nu1-parameterization}, \eqref{Eq-nu234-parameterization} and \eqref{Eq-log-deri} gives 
		\begin{equation}\label{Eq-s1-parameterization}
			s_1=-\frac{\	\sum_{i=1}^4\nu_i}{4}=-\frac{\mathrm{i}\wp'(\kappa)}{\wp(\rho)-\wp(\kappa)} - \frac{\mathrm{i}\wp''(\kappa)}{2\wp'(\kappa)}.
		\end{equation}
		Substituting \eqref{Eq-nu-parameterization}, \eqref{Eq-C-para} and \eqref{Eq-s1-parameterization} into \eqref{Eq-mu-nu} yields
		\begin{equation}
			\begin{split}
				\mu(\xi) 
				& = \frac{\mathrm{i}}{8}\Bigg(\frac{\wp'(\kappa)-\wp'(\xi)}{\wp(\kappa)-\wp(\xi)} 
				- \frac{\wp'(\xi)-\wp'(\rho)}{\wp(\xi)-\wp(\rho)} 
				+ \frac{\wp'(\rho)-\wp'(\kappa)}{\wp(\rho)-\wp(\kappa)} 
				- \frac{\wp''(\kappa)}{\wp'(\kappa)}\Bigg).
			\end{split}
		\end{equation}
		Applying the second formula in \eqref{Eq-formulas of Weierstrass functions} and the half argument formula \eqref{Eq-half-argument-2} produces \eqref{Eq-mu-parameterization}, completing the proof.
	\end{proof}
	
	Building upon \eqref{Eq-mu}, we could establish the $\xi$-depenence of the elliptic function solutions. Moreover, combining with \eqref{Eq-u-travelling-2}, 
	we generate elliptic function solutions to the DNLS equation \eqref{Eq-DNLS-equation}.
	
	\begin{prop}
		The elliptic function solutions to the DNLS equation \eqref{Eq-DNLS-equation} can be expressed as 
		\begin{equation}\label{Eq-DNLS-elliptic-solution}
			\begin{split}
				u(\xi,t)=\frac{\sqrt{\nu_0}\sigma(\kappa)\sigma(\xi+\rho)\sigma(\xi+\kappa)}{\sigma(\rho)\sigma^2(\xi-\kappa)}e^{-F(\xi,t)},
			\end{split}
		\end{equation}
		where the function $F(\xi,t)$ is defined as 
		\begin{equation}\label{Eq-F} 
			F(\xi,t)=\big(\zeta(\rho+\kappa)+\zeta(2\kappa)\big)\xi+16\mathrm{i}\alpha_4t.
		\end{equation}
	\end{prop}
	\begin{proof}
		Rewriting \eqref{Eq-mu} yields 
		\begin{equation}\label{Eq-u-ODE}
			u_\xi=4\mathrm{i}\left(\mu(\xi)-\frac{s_1+\nu(\xi)}{2}\right)u.
		\end{equation}
		The components $s_1$, $\nu(\xi)$ and $\mu(\xi)$ in \eqref{Eq-u-ODE} are now explicitly expressed via Weierstrass functions. Direct integration then yields the solution. Combining \eqref{Eq-nu-parameterization},   \eqref{Eq-mu-parameterization} and \eqref{Eq-s1-parameterization} gives
		\begin{equation}
			\begin{split}
				4\mathrm{i}\left(\mu(\xi)-\frac{s_1+\nu(\xi)}{2}\right)=-\zeta(\kappa+\xi)+\zeta(\rho+\xi)-\zeta(\rho+\kappa)+\zeta(2\kappa)-\frac{\wp''(\kappa)}{\wp'(\kappa)}-\frac{2\wp'(\kappa)}{\wp(\xi)-\wp(\kappa)}.
			\end{split}
		\end{equation}
		Employing \eqref{Eq-derivative-sigma}, \eqref{Eq-half-argument-2}  and  \eqref{Eq-integration-formulas}, the integral evaluates to 
		\begin{equation}\label{Eq-integrand-1}
			\begin{split}
				\int_0^{\xi}4\mathrm{i}\left(\mu(\xi')-\frac{s_1+\nu(\xi')}{2}\right)d\xi'=\ln\left(\frac{\sigma(\xi+\rho)\sigma(\xi+\kappa)}{\sigma^2(\xi-\kappa)}\right)-\left(\zeta(2\kappa)+\zeta(\rho+\kappa)\right)\xi-\ln\left(\frac{\sigma(\rho)}{\sigma(\kappa)}\right).
			\end{split}
		\end{equation}
		As a consequence of \eqref{Eq-u-travelling-2}, we have 
		\begin{equation}\label{Eq-integrand-2}
			\begin{split}
				&	\quad \ln\big(u(\xi,t)\big)-\ln\big(\sqrt{\nu_0}\big)=	\int_0^{\xi}4\mathrm{i}\left(\mu(\xi')-\frac{s_1+\nu(\xi')}{2}\right)d\xi'+16\mathrm{i}\alpha_4t.\\
			\end{split}
		\end{equation}	
		Combining with \eqref{Eq-integrand-1} and \eqref{Eq-integrand-2}, we deduce that \eqref{Eq-DNLS-elliptic-solution}-\eqref{Eq-F} constitute a solution to the DNLS equation \eqref{Eq-DNLS-equation}, completing the proof.
	\end{proof}

	We have now obtained a unified expression \eqref{Eq-DNLS-elliptic-solution} for the elliptic function solutions to the DNLS equation \eqref{Eq-DNLS-equation}. However, the parameter $\alpha_4$ involved in this expression remains unexpressed in terms of elliptic functions.  For theoretical completeness, we show that upon introducing two auxiliary parameters $\omega_{1,3}, \alpha_4$ can be written as Weierstrass functions. In fact,
	differentiating $g$ with respect to $\xi$ in \eqref{Eq-fgh} gives
	\begin{align}\label{Eq-g-xi-1}
		g_\xi=u_\xi\lambda(\lambda^2-\mu)-u\lambda\mu_\xi.
	\end{align}
	Substituting the spectral parameter $\lambda=\sqrt{\mu}$ into
	\eqref{Eq-g-xi-1} yields
	\begin{align}\label{Eq-g-xi-2}
		g_\xi=-u\sqrt{\mu}\mu_\xi.
	\end{align}
	Besides, the spatial part of the stationary zero curvature equation \eqref{Eq-stationary-zero-curvature-equation} implies 
	\begin{align}\label{Eq-g-xi-3}
		g_\xi=-4  \lambda u f-4 \mathrm{i} \lambda^2 g.	
	\end{align}
	Inserting $\lambda=\sqrt{\mu}$ into \eqref{Eq-g-xi-3} and employing \eqref{Eq-fgh} and \eqref{Eq-g-xi-1} yields
	\begin{align}\label{Eq-mu-xi}
		\mu_\xi=4f(\sqrt{\mu}).
	\end{align}
	Furthermore, combining \eqref{Eq-Lax-matrix} with \eqref{Eq-mu-xi}, we obtain
	\begin{equation}\label{Eq-mu-xi-sqr}
		\mu_\xi^2=16f^2(\sqrt{\mu})=16P(\sqrt{\mu}).
	\end{equation}
	Equation \eqref{Eq-mu-xi-sqr} establishes that $(\lambda^{(i)})^2, i=1,2,3,4$ correspond precisely to the values of $\mu(\xi)$ at the zeros of $\mu'(\xi)$ within a periodic parallelogram. Without loss of generality, these zeros are
	$-\frac{\kappa+\rho}{2}$, $-\frac{\kappa+\rho}{2}+\omega_1$, $-\frac{\kappa+\rho}{2}+\omega_3$ and $-\frac{\kappa+\rho}{2}+\omega_1+\omega_3$. The square root of the product of their $\mu$-values gives the parameterization of $\alpha_4$. Therefore, every elliptic solution \eqref{Eq-DNLS-elliptic-solution} to the DNLS equation can be parameterized by the four parameters $\kappa,\rho,\omega_1,\omega_3$.
	
We verify that $\kappa,\rho,\omega_1$ and $\omega_3$ are independent parameters in general in the following. To this end, we establish a one-to-one correspondence between the roots $\nu_i$ of $R(\nu)$ \eqref{Eq-R} and parameters $\kappa,\rho,\omega_1,\omega_3$. When $\nu_0=\nu_1$, equations \eqref{Eq-ei-and-vi} and \eqref{Eq-kappa-rho} yield the Jacobian determinant:
	\begin{equation}\label{Eq-Jacobian}
		\frac{\partial\left(\wp(\kappa),\wp(\rho),\wp(\omega_1),\wp(\omega_3)\right)}{\partial\left(\nu_1,\nu_2,\nu_3,\nu_4\right)}=\frac{(\nu_1-\nu_3)^2(\nu_1-\nu_2)(\nu_1-\nu_4)\big(2(\nu_1-\nu_4)(\nu_2-\nu_4)+\nu_4(\nu_3-\nu_4)\big)}{768\nu_1^2}.
	\end{equation}
	For squared moduli of Type 1,
	\begin{equation}
		2(\nu_1-\nu_4)(\nu_2-\nu_4)+\nu_4(\nu_3-\nu_4)=2\nu_1\nu_2+\nu_4^2+\nu_3\nu_4-2\nu_1\nu_4-2\nu_2\nu_4>0,
	\end{equation}
	since all right-hand side terms are nonnegative and $\nu_4^2>0$ holds strictly. For Type 2, it is evidently that
	\begin{equation}
		2(\nu_1-\nu_4)(\nu_2-\nu_4)+\nu_4(\nu_3-\nu_4)>0.
	\end{equation}
	For Type 3, the Jacobian determinant \eqref{Eq-Jacobian} vanishes if and only if 
	\begin{equation}
		\begin{split}
			\big(\Re(\nu_3)-\nu_1\big)\big(\Re(\nu_3)-\nu_2\big)&=0,\\
			\Im(\nu_3)\big(\nu_1+\nu_2-2\Re(\nu_3)\big)&=0.
		\end{split}
	\end{equation}
	These identities cannot simultaneously hold. By the inverse function theorem, the correspondence between $(\nu_1,\nu_2,\nu_3,\nu_4)$ and $\left(\wp(\kappa),\wp(\rho),\wp(\omega_1),\wp(\omega_3)\right)$ is one-to-one. Thus, by construction, a one-to-one correspondence exists between $\nu_i,i=1,2,3,4$ and  $\kappa,\rho,\omega_1,\omega_3$. However,
	when $\nu_0=\nu_4$, the condition $	2(\nu_4-\nu_1)(\nu_3-\nu_1)+\nu_1(\nu_2-\nu_1)\neq 0$ should be imposed on $\nu_i,i=1,2,3,4$ such that the Jacobi determinant does not vanish. 
	
	As is verified above, these four parameters are generally independent. For comparison, elliptic function solutions of the NLS equation \cite{Ling-NLS}, the mKdV equation \cite{LS-mKdV-solution}, and the sine-Gordon equation \cite{LS-sG} are each uniquely parameterized by three independent parameters $\alpha,k$ and $l$. The parameterization here thus extends prior results. Analogously, the methodology presented here are potential to be extended to other Kaup-Newell equations.

	\subsection{The fundamental solution matrix to the Lax pair }
	In this subsection, we construct the fundamental solution matrix for the Lax pair \eqref{Eq-Lax pair}. Our approach introduces a uniform parameter $z$ replacing the spectral parameter $\lambda$, which is essential for both constructing the fundamental solution matrix and obtaining $N$-elliptic localized solutions. To achieve this, we define the elliptic curve
	\begin{align}\label{Eq-K2} \mathcal{K}_2:=\{(\lambda,y)|y^2=P(\lambda)\}.
	\end{align}
	Equation \eqref{Eq-mu-xi-sqr} implies that $(\sqrt{\mu},\frac{\mu_\xi}{4})$ lies on $\mathcal{K}_2$. We therefore introduce the parameterization 
	\begin{equation}\label{Eq-parameterization} \big(\lambda(z),y(z)\big)=\big(\sqrt{\mu(z)},\frac{\mu'(z)}{4}\big),
	\end{equation}
	using \eqref{Eq-mu-parameterization}. Explicitly,
	\begin{equation}\label{Eq-parameterization-lambda-y}
		\big(\lambda(z),y(z)\big)=\left(-\frac{\sigma(z-\kappa)\sigma(\rho)}{2\sqrt{\nu_0}\sigma(\rho+z)\sigma(-\rho-\kappa)\sigma(\kappa)}\frac{d_0(z)}{d_0(\hat{z})},\frac{\mathrm{i}}{16}\left(\wp(z+\rho)-\wp(z+\kappa)\right)\right),
	\end{equation}
	where 
	\begin{equation}\label{Eq-nu0-d0}
		d_0(z)=\sqrt{\sigma(\rho+z)\sigma(z+2\kappa+\rho)},\quad \hat{z}=-\kappa-\rho-z.
	\end{equation}
	From \eqref{Eq-parameterization-lambda-y}-\eqref{Eq-nu0-d0}, through direct computations we obtain 
	\begin{equation}\label{Eq-shift}
		\lambda(z)=\lambda(\hat{z}),\quad y(z)=-y(\hat{z}),
	\end{equation}
	indicating that the transformation $z \mapsto \hat{z}$ preserves $\lambda$ while reversing the sign of $y$.
	\begin{rmk}
		Substituting $\xi=z$ into \eqref{Eq-mu-parameterization} and taking the square root does not immediately yield the parametric expression for $\lambda$ in \eqref{Eq-parameterization-lambda-y}. To demonstrate their equivalence, one can employ the first formula in \eqref{Eq-formulas of Weierstrass functions} and the half argument formula \eqref{Eq-half-argument-1}. We adopt this particular form for $\lambda(z)$ \eqref{Eq-parameterization-lambda-y} because it facilitates applying the addition formula \eqref{Eq-addition formulas of the sigma functions} in subsequent constructions for $N$-elliptic localized solutions.
	\end{rmk}	
	\begin{rmk}
		Based on the theory of elliptic functions, the mapping 
		\begin{equation}
			\begin{split}
				S_1&\rightarrow \mathcal{K}_2\\
				z&\mapsto \big(\lambda(z),y(z)\big),
			\end{split}
		\end{equation}
		is a one-to-one correspondence, where $S_1$ can be selected as arbitrary periodic parallelogram of the function $\wp(z)$.
	\end{rmk}
	
	Clearly, $\pm \mathrm{i}y(z)$ are eigenvalues of $\mathbf{L}\big(x,t;z\big)$. Let $\left(1,r_\pm(\xi,t;z)\right)^T$ denote the kernels of $\pm\mathrm{i}y(z)\mathbb{I}-\mathbf{L}(\xi,t;z)$, respectively. The existence and uniqueness of solutions to ODEs guarantee two linearly independent solutions $\big(\phi_{\pm}(\xi,t;z),\psi_{\pm}(\xi,t;z)\big)$ to the Lax pair \eqref{Eq-Lax pair} such that
	\begin{align}\label{Eq-rpm}
		r_{\pm}(\xi,t;z)=\frac{\psi_{\pm}(\xi,t;z)}{\phi_{\pm}(\xi,t;z)}=\frac{\mathrm{i}\big(f(\xi,t;z)\pm y(z)\big)}{g(\xi,t;z)}=\frac{\mathrm{i}h(\xi,t;z)}{f(\xi,t;z)\mp y(z)}.
	\end{align}
	Based on these preliminaries, we obtain the following lemmas. 
	\begin{lemma}\label{Prop-psi+}
		The function $\phi_{+}(\xi,t;z)$ admits the explicit representation
		\begin{align}\label{Eq-phi-p}
			\phi_+(\xi,t;z)=\frac{\sigma(z-\xi)}{\sigma(\xi-\kappa)}e^{\left(\frac{1}{2}\zeta(\kappa+z)+\frac{1}{2}\zeta(z+\rho)-\frac{1}{2}\zeta(\kappa+\rho)-\frac{1}{2}\zeta(2\kappa)\right)\xi-8\big(\mathrm{i}\alpha_4+y(z)\big)t}.
		\end{align}
		The corresponding function $\phi_{-}$ is obtained by substituting $z$ with $\hat{z}$ in \eqref{Eq-phi-p}.
	\end{lemma}
	\begin{proof}
		Utilizing \eqref{Eq-Lax pair}, \eqref{Eq-fgh}, \eqref{Eq-Im-mu} and \eqref{Eq-rpm}, we obtain that the function $\phi_+$ satisfies  
		\begin{equation}\label{Eq-spatial}
			\begin{split}
				\phi_{+,\xi}&\xlongequal{\eqref{Eq-Lax pair}}-2\mathrm{i}\lambda^2\phi_++2\lambda u r_{+}\phi_+\\
				&\xlongequal[\eqref{Eq-rpm}]{\eqref{Eq-fgh}}-2\mathrm{i}\lambda^2\left(1+\frac{\mathrm{i}\nu(\xi)\big(\lambda^2-\mu^*(\xi)\big)}{-\frac{\mathrm{i}}{2}\nu(\xi)\lambda^2+\beta_{+}}\right)\phi_+\\
				&=-2\mathrm{i}\lambda^2\left(1+\frac{\mathrm{i}\nu(\xi)\lambda^2-2\beta_{+}+2\beta_{+}+\mathrm{i}\nu(\xi)\big(\mu(\xi)-\mu^*(\xi)\big)-\mathrm{i}\nu(\xi)\mu(\xi)}{-\frac{\mathrm{i}}{2}\nu(\xi)\lambda^2+\beta_{+}}\right)\phi_+\\
				&\xlongequal{\eqref{Eq-Im-mu}}-2\mathrm{i}\lambda^2\left(-1+\frac{2\beta_{+}-\mathrm{i}\nu(\xi)\mu(\xi)}{-\frac{\mathrm{i}}{2}\nu(\xi)\lambda^2+\beta_{+}}+\frac{\nu'(\xi)}{4\left(-\frac{\mathrm{i}}{2}\nu(\xi)\lambda^2+\beta_{+}\right)}\right)\phi_+\\
				&:=\big(\uppercase\expandafter{\romannumeral1}(\xi)+\uppercase\expandafter{\romannumeral2}(\xi)+\uppercase\expandafter{\romannumeral3}(\xi)\big)\phi_+,
			\end{split}
		\end{equation}
		where we introduce the notation	\begin{align}\label{Eq-def-beta-p}
			\beta_{+}= y+\mathrm{i}\lambda^4-\frac{\mathrm{i}s_1}{2}\lambda^2-\mathrm{i}\alpha_4,
		\end{align}
	for simplicity. Furthermore, it follows from \eqref{Eq-Lax pair}, \eqref{Eq-fgh}, \eqref{Eq-mu}  and \eqref{Eq-rpm} that 
	\begin{equation}\label{Eq-phi-t}
		\begin{split}
			\phi_{+,t}&\xlongequal{\eqref{Eq-Lax pair}}\big(4\lambda^2-2s_1-2\nu(\xi)\big)\big(-2\mathrm{i}\lambda^2+2\mathrm{i}\lambda u r_+\big)\phi_+-2\lambda u_\xi r_+\phi_+\\
			&= -2\mathrm{i}\lambda^2\big(4\lambda^2-2s_1-2\nu(\xi)\big)\phi_++2\lambda\left(\mathrm{i}\left(4\lambda^2-2s_1-2\nu(\xi)\right)u-u_\xi\right)r_+\phi_+\\
			&\xlongequal{\eqref{Eq-fgh}}-8\big(f(\xi)+y(z)\big)\phi_++8\mathrm{i}\lambda u\left(\left(\lambda^2-\frac{1}{2}\nu(\xi)-\frac{1}{2}s_1\right)+\frac{\mathrm{i}}{4}\frac{u_\xi}{u}\right)\frac{\mathrm{i}\big(f(\xi)+y(z)\big)}{\lambda u\big(\lambda^2-\mu(\xi)\big)}\phi_+\\
			&\xlongequal{\eqref{Eq-mu}}-8\big(f(\xi)+y(z)\big)\phi_++8\big(f(\xi)-\mathrm{i}\alpha_4\big)\phi_+=-8\big(y(z)+\mathrm{i}\alpha_4\big)\phi_+.
		\end{split}
	\end{equation}
As a consequence, it is satisfied that
\begin{equation}\label{Eq-line-integral}
	\begin{split}
		&\quad \ln\big(\phi_+(\xi,t)\big)-\ln\big(\phi_+(0,0)\big)=\int_{(0,0)}^{(\xi,t)}\frac{\phi_{+,\xi'}(\xi',t')}{\phi_{+}(\xi',t')} d\xi'+ \frac{\phi_{+,t'}(\xi',t')}{\phi_{+}(\xi',t')}  dt'\\
	&=\int_0^{\xi}\frac{\phi_{+,\xi'}(\xi',0)}{\phi_{+}(\xi',0)} d\xi'+\int_0^{t}\frac{\phi_{+,t'}(\xi,t')}{\phi_{+}(\xi,t')} dt'\\
	&=\int_0^{\xi} \big(I(\xi')+II(\xi')+III(\xi')\big)d\xi'-8\int_0^{t}\big(y(z)+\mathrm{i}\alpha_4\big)dt'.
	\end{split}
\end{equation}
	Henceforth, it suffices to evaluate the line integral to obtain the representation of function $\phi_+$. We proceed by parameterizing $\beta_+$ with $z$ and henceforth the functions $II(\xi)$ and $III(\xi)$ in \eqref{Eq-spatial} can be integrated. It follows from \eqref{Eq-fgh} and \eqref{Eq-K2}-\eqref{Eq-parameterization-lambda-y} that 
	\begin{equation}
		y^2(z)=P\big(\lambda(z)\big)=f^2\big(\lambda(z)\big)-g\big(\lambda(z)\big)h\big(\lambda(z)\big)=f^2\big(\lambda(z)\big).
	\end{equation}
As a consequence, we could set
		\begin{align}\label{Eq-y-parameterization}
			y(z)=f\left(\lambda(z)\right),
		\end{align}
	without loss of generality. Actually, if otherwise $y(z)=-f\left(\lambda(z)\right)$, we could substitute $z$ with $\hat{z}$ and still denote $\hat{z}$ as $z$. Consequently, using \eqref{Eq-fgh}, \eqref{Eq-def-beta-p} and \eqref{Eq-y-parameterization}, $\beta_{+}$ is parameterized by $z$ as 
		\begin{align}\label{Eq-beta-plus}
			\beta_{+}(z)=\frac{\mathrm{i}}{2}\nu(z)\mu(z).
		\end{align}
 Using the parameterization \eqref{Eq-beta-plus}, we split $II(\xi)$ into three terms:
		\begin{align}\label{Eq-II}
				II(\xi)=\frac{-4\mathrm{i}\left(\nu(z)\mu(z)-\nu(\xi)\mu(\xi)\right)}{\left(\nu(z)-\nu(\xi)\right)}=-\frac{\mathrm{i}}{2}\left(2s_1+\nu(z)+\nu(\xi)-\frac{\mathrm{i}\big(\nu'(z)-\nu'(\xi)\big)}{\nu(z)-\nu(\xi)}\right).
		\end{align}
We integrate the function $II(\xi)$ by terms. Using \eqref{Eq-nu1-parameterization} and the integration formulas \eqref{Eq-integration-formulas}, we arrive at
		\begin{equation} \label{Eq-II-1}
			\begin{split}
			\int_0^\xi \nu(s)ds
				=\int_0^\xi \frac{\mathrm{i}\wp'(\kappa)}{\wp(\rho)-\wp(\kappa)}+\frac{\mathrm{i}\wp'(\kappa)}{\wp(\kappa)-\wp(s)}ds
				=\nu_0\xi-\mathrm{i}\left(\ln\frac{\sigma(\kappa-\xi)}{\sigma(\xi+\kappa)}+2\xi\zeta(\kappa)\right),
			\end{split}
		\end{equation}
		and
		\begin{equation}\label{Eq-II-2}
			\begin{split}
				&\quad \int_0^\xi\frac{\nu'(s)-\nu'(z)}{\nu(s)-\nu(z)}ds\\
				&\xlongequal{\eqref{Eq-nu-parameterization}}\quad\int_0^\xi\frac{\wp'(z)\big(\wp(\rho)-\wp(\kappa)\big)\big(\wp(s)-\wp(\kappa)\big)^2-\wp'(s)\big(\wp(\rho)-\wp(\kappa)\big)\big(\wp(z)-\wp(\kappa)\big)^2}{\big(\wp(s)-\wp(\kappa)\big)\big(\wp(z)-\wp(\kappa)\big)\big(\wp(\rho)-\wp(\kappa)\big)\big(\wp(z)-\wp(s)\big)}ds\\
				&=\int_0^\xi\frac{\wp'(z)\big(\wp(s)-\wp(\kappa)\big)}{\big(\wp(z)-\wp(\kappa)\big)\big(\wp(z)-\wp(s)\big)}-\frac{\wp'(s)\big(\wp(z)-\wp(\kappa)\big)}{\big(\wp(z)-\wp(s)\big)\big(\wp(s)-\wp(\kappa)\big)}ds\\
				&=\int_0^\xi\frac{\wp'(z)\big(\wp(s)-\wp(z)+\wp(z)-\wp(\kappa)\big)}{\big(\wp(z)-\wp(\kappa)\big)\big(\wp(z)-\wp(s)\big)}-\frac{\wp'(s)\big(\wp(z)-\wp(s)+\wp(s)-\wp(\kappa)\big)}{\big(\wp(z)-\wp(s)\big)\big(\wp(s)-\wp(\kappa)\big)}ds\\
				&=\int_0^\xi\frac{\wp'(z)}{\wp(\kappa)-\wp(z)}+\frac{\wp'(z)}{\wp(z)-\wp(s)}+\frac{\wp'(s)}{\wp(\kappa)-\wp(s)}+\frac{\wp'(s)}{\wp(s)-\wp(z)}ds\\
				&\xlongequal{\eqref{Eq-integration-formulas}}\left(\frac{\wp'(z)}{\wp(\kappa)-\wp(z)}-2\zeta(z)\right)\xi+\ln\left(\frac{\sigma^2(\xi+z)\sigma^2(\kappa)}{\sigma(\kappa-\xi)\sigma(\kappa+\xi)\sigma^2(z)}\right).
			\end{split}
		\end{equation}
	As a consequence of \eqref{Eq-II}-\eqref{Eq-II-2}, 
	\begin{equation}\label{Eq-II-integral}
		\begin{split}
\int_0^{\xi}II(\xi')d\xi'&=-\ln\left(\frac{\sigma(\xi+z)\sigma(\kappa)}{\sigma(\xi+\kappa)\sigma(z)}\right)-\frac{\mathrm{i}}{2}\left(2s_1+\nu(z)+\nu_0+2\mathrm{i}\zeta(z)-2\mathrm{i}\zeta(\kappa)-\frac{\mathrm{i}\wp'(z)}{\wp(\kappa)-\wp(z)}\right)\xi,
\end{split}
		\end{equation}
Moreover, $III$ can be directly integrated under the parameterization \eqref{Eq-parameterization-lambda-y} and \eqref{Eq-beta-plus}. Specifically, 
\begin{equation}\label{Eq-III-integral}
	\begin{split}
 \int_0^{\xi}III(\xi')d\xi'=\int_0^\xi\frac{\nu'(\xi')}{\nu(\xi')-\nu(z)}d\xi'=\ln\frac{\sigma^2(\kappa)\sigma(\xi-z)\sigma(\xi+z)}{\sigma^2(z)\sigma(\xi-\kappa)\sigma(\xi+\kappa)}.
	\end{split}
\end{equation}
 Ultimately, it suffices to simplify the coefficient of the linear part in the integration $\int_0^{\xi}I(\xi')+II(\xi')d\xi'$.
Actually, \eqref{Eq-nu-parameterization}, \eqref{Eq-nu1-parameterization},  \eqref{Eq-mu-parameterization}, \eqref{Eq-s1-parameterization}, \eqref{Eq-formulas of Weierstrass functions} and \eqref{Eq-half-argument-2} give rise to the following simplification:
\begin{equation}\label{Eq-linear-coefficient}
	\begin{split}
		&\quad -\mathrm{i}s_1-\frac{\mathrm{i}}{2}\nu(z)-\frac{\mathrm{i}}{2}\nu_0+\zeta(z)-\zeta(\kappa)-\frac{1}{2}\frac{\wp'(z)}{\wp(\kappa)-\wp(z)}+2\mathrm{i}\lambda^2(z)\\
		&\xlongequal[\eqref{Eq-mu-parameterization},\eqref{Eq-s1-parameterization}]{\eqref{Eq-uniform-expression-for-the-squared-modulus},\eqref{Eq-nu-0}}\left(-\frac{1}{2}\frac{\wp''(\kappa)}{\wp'(\kappa)}-\zeta(\kappa)+\frac{1}{2}\zeta(2\kappa)\right)+\frac{1}{2}\frac{\wp'(\kappa)-\wp'(z)}{\wp(\kappa)-\wp(z)}-\frac{1}{2}\left(\zeta(\kappa+z)-\zeta(\rho+z)+\zeta(\rho+\kappa)\right)+\zeta(z)\\
		&\xlongequal[\eqref{Eq-half-argument-2}]{\eqref{Eq-formulas of Weierstrass functions}}-\frac{1}{4}\frac{\wp''(\kappa)}{\wp'(\kappa)}+\left(\zeta(\kappa+z)-\zeta(\kappa)-\zeta(z)\right)-\frac{1}{2}\left(\zeta(\kappa+z)-\zeta(\rho+z)+\zeta(\rho+\kappa)\right)+\zeta(z)\\
		&\xlongequal{\eqref{Eq-half-argument-2}}\frac{1}{2}\left(\zeta(\kappa+z)+\zeta(\rho+z)-\zeta(\rho+\kappa)-\zeta(2\kappa)\right).
	\end{split}
\end{equation}
Combining \eqref{Eq-line-integral} with \eqref{Eq-II-integral}-\eqref{Eq-linear-coefficient}, we know that  $\big(\phi_+(\xi,t;z),r_+(\xi,t;z)\phi_+(\xi,t;z)\big)^T$ 
is a solution to the Lax pair \eqref{Eq-Lax pair} with the potential $u$ \eqref{Eq-DNLS-elliptic-solution} at $\lambda=\lambda(z)$, where $\phi_+$ is defined as \eqref{Eq-phi-p}.  
Another linearly independent solution  can be derived by substituting $z$ with $\hat{z}$. As a consequence, we complete the proof.
	\end{proof}
Henceforth, to establish the explicit form of the solutions to the Lax pair, it suffices to find the explicit expressions for $r_\pm(\xi,t;z)$.
	\begin{lemma}\label{Prop-r+}
		The function $r_+(\xi,t;z)$ admits the explicit representation
		\begin{align}\label{Eq-para-r+}
			r_+(\xi,t;z)=\mathrm{i}\frac{\sigma(-\hat{z}-\xi)\sigma(\kappa-\xi)d_0(\hat{z})}{\sigma(\xi-z)\sigma(\kappa+\xi)d_0(z)}e^{F(\xi,t)}.
		\end{align}
		The corresponding function $r_-(\xi,t;z)$ is obtained through the substitution $z \mapsto \hat{z}$.
	\end{lemma}
	
	\begin{proof}
		Using \eqref{Eq-nu-parameterization} and \eqref{Eq-mu-parameterization}, we know that
		$\nu(z)-\nu(\xi)$, $\nu(\hat{z})-\nu(\xi)$, $\lambda^2(z)-\mu(\xi)$ and  $\lambda^2(z)-\mu^*(\xi)$ are elliptic functions of $\xi$. Consequently,  within a periodic parallelogram, the number of their zeros equals that of their poles respectively. Using \eqref{Eq-nu-parameterization}, analysis of the zeros and poles yields
		\begin{align}\label{Eq-zero-pole-1}
			\nu(z)-\nu(\xi)\propto\frac{\sigma(\xi-z)\sigma(\xi+z)}{\sigma(\kappa+\xi)\sigma(\xi-\kappa)}.
		\end{align} Analogously, we derive
		\begin{align}\label{Eq-zero-pole-2}
			\nu(\hat{z})-\nu(\xi)\propto\frac{\sigma(\xi-\hat{z})\sigma(\xi+\hat{z})}{\sigma(\kappa+\xi)\sigma(\xi-\kappa)}.
		\end{align}
		Furthermore, based on the expression \eqref{Eq-mu-parameterization} and \eqref{Eq-shift}, we arrive at
		\begin{align}\label{Eq-zero-pole-3}
			\lambda^2(z)-\mu(\xi)\propto\frac{\sigma(\xi-z)\sigma(\xi-\hat{z})}{\sigma(\xi+\kappa)\sigma(\xi+\rho)}.
		\end{align}
		It follows directly from \eqref{Eq-mu-parameterization} that  the auxiliary spectrum $\mu$ admits the symmetry $\mu^*(\xi)=\mu(-\xi)$. Thus we obtain
		\begin{align}\label{Eq-zero-pole-4}
			\lambda^2(z)-\mu^*(\xi)\propto\frac{\sigma(\xi+z)\sigma(\xi+\hat{z})}{\sigma(\xi-\kappa)\sigma(\xi-\rho)}.
		\end{align}
		Moreover, we deduce from  \eqref{Eq-DNLS-elliptic-solution} that $u$ admits the symmetry $u^*(\xi,t)=u(-\xi,-t)$. Henceforth,
		\begin{align}\label{Eq-zero-pole-5}
			\frac{u^*(\xi,t)}{u(\xi,t)}=\frac{\sigma(\xi-\rho)\sigma(\xi-\kappa)^3}{\sigma(\xi+\rho)\sigma(\xi+\kappa)^3}e^{2F(\xi,t)}.
		\end{align}
		Combining \eqref{Eq-fgh}, \eqref{Eq-rpm} and \eqref{Eq-zero-pole-1}-\eqref{Eq-zero-pole-5} together, we deduce that 
		\begin{align}\label{Eq-r_+}
			r_{+}(\xi,t;z)=\hat{r}_+(z)\frac{\sigma(\xi+\hat{z})\sigma(\xi-\kappa)}{\sigma(\xi-z)\sigma(\xi+\kappa)}e^{F(\xi,t)},
		\end{align}
		where $\hat{r}_+(z)$ is independent of $\xi$ and $t$.
		By plugging $\xi=t=0$ into \eqref{Eq-r_+} and making utilize of \eqref{Eq-fgh}, \eqref{Eq-rpm}, \eqref{Eq-zero-pole-1} and \eqref{Eq-zero-pole-2} together with the first formula in \eqref{Eq-formulas of Weierstrass functions}, we have
		\begin{align}
			r_+(0,0;z) \xlongequal{\eqref{Eq-r_+}}\hat{r}_+(z)\frac{\sigma(\hat{z})}{\sigma(z)}\xlongequal[\eqref{Eq-rpm}]{\eqref{Eq-fgh}}\sqrt{-\frac{\nu(\hat{z})-\nu(0)}{\nu(z)-\nu(0)}}\xlongequal[\eqref{Eq-zero-pole-2}]{\eqref{Eq-zero-pole-1}}\sqrt{-\frac{\sigma(z+\kappa)\sigma(z-\kappa)}{\sigma(\hat{z}+\kappa)\sigma(\hat{z}-\kappa)}}\frac{\sigma(\hat{z})}{\sigma(z)},
		\end{align}
		which indicates that 
		\begin{align}
			\hat{r}_+(z)=\sqrt{-\frac{\sigma(z+\kappa)\sigma(z-\kappa)}{\sigma(\hat{z}+\kappa)\sigma(\hat{z}-\kappa)}}=\mathrm{i}\frac{d_0(\hat{z})}{d_0(z)}.
		\end{align}
		Henceforth we have proved the proposition.
	\end{proof}
	With these preparations in hand, we can establish the fundamental solution matrix to the Lax pair \eqref{Eq-Lax pair}, which is concluded in the following theorem.
	
	\begin{theorem}	[The fundamental solution matrix to the Lax pair]\label{thm:solution to the Lax pair}
		The fundamental solution matrix of the Lax pair \eqref{Eq-Lax pair} with the elliptic solution \eqref{Eq-DNLS-elliptic-solution} at $\lambda(z)$ can be written as
		\begin{align}\label{Eq-solution to the Lax pair-1}
			\Phi \big(\xi,t;\lambda(z)\big) = \begin{pmatrix} 
				d(\xi;z)d_0(z)E(\xi,t;z) & d(\xi;\hat{z})d_0(\hat{z})E(\xi,t;\hat{z}) \\
				\mathrm{i}d(-\xi;\hat{z})d_0(\hat{z})E(\xi,t;z)e^{F(\xi,t)} & \mathrm{i}d(-\xi;z)d_0(z)E(\xi,t;\hat{z})e^{F(\xi,t)}
			\end{pmatrix},
		\end{align}
		where the components are defined by
		\begin{equation}\label{Eq-solution to the Lax pair-2}
			\begin{split}
				E(\xi,t;z) &= \exp\Big(\frac{1}{2}\big(\zeta(\kappa+z)+\zeta(\rho+z)-\zeta(\rho+\kappa)-\zeta(2\kappa)\big)\xi-8\big(\mathrm{i}\alpha_4+y(z)\big)t\Big),\quad d(\xi;z) = \frac{\sigma(z-\xi)}{\sigma(\xi-\kappa)}.
			\end{split}
		\end{equation}
	\end{theorem}
	
	\begin{proof}
		Based on Lemma~\ref{Prop-psi+}-\ref{Prop-r+}, the vector $\big(d(\xi;z)d_0(z)E(z),\mathrm{i}d(-\xi;\hat{z})d_0(\hat{z})E(z)e^{F(\xi,t)}\big)^T$ is a solution to the Lax pair~\eqref{Eq-Lax pair}. Substituting $\hat{z}$ for $z$ yields another  solution. These two solutions are linearly independent since that they belong to the kernels of $\textbf{L}$ at distinct eigenvalues $\pm \mathrm{i}y(z)$ respectively. This completes the proof.
	\end{proof}

	\section{$N$-elliptic localized solutions}
	In this section, we construct the $N$-elliptic localized solution to the DNLS equation \eqref{Eq-DNLS-equation} through $BT_0$ \eqref{Eq-BT0-1}-\eqref{Eq-BT0-2}. The solution is written in compact form as the derivative of a ratio of two determinants. 
	A crucial tool for constructing $N$-elliptic localized solutions is the Sherman-Morrison-Woodbury-type matrix identity:
	\begin{equation}\label{Eq-SMW-identity}
		\begin{aligned}
			a+\textbf{q}^{\dagger} \mathbf{A}^{-1} \textbf{p}=\frac{a^{1-N} \operatorname{det}\left(a \mathbf{A}+\textbf{p} \textbf{q}^{\dagger}\right)}{\operatorname{det} \mathbf{A}},
		\end{aligned}
	\end{equation}
	where $\mathbf{A}\in\mathbb{C}^{n\times n}$, $a\in\mathbb{C}$ and $\textbf{p},\textbf{q}\in\mathbb{C}^{n\times 1}$. Assume that $z_i,i=1,2,\ldots,N$ are uniform parameters that correspond with the spectral parameters via $\lambda_i=\lambda(z_i)$ and $\phi_i=(\phi^{(1)}_i,\phi^{(2)}_i)^T,i=1,2,\ldots,N$ are solutions to the Lax pair \eqref{Eq-Lax pair} at the elliptic solution \eqref{Eq-DNLS-elliptic-solution}-\eqref{Eq-F} respectively. Each $\phi_i$ represents a linear combination of the two columns of the fundamental solution matrix $\Phi\left(\xi,t;\lambda(z_i)\right)$ defined in \eqref{Eq-solution to the Lax pair-1}. Specifically, 
	\begin{align}\label{Eq-phi-l}
		\begin{pmatrix}
			\phi_{i}^{(1)} \\ \phi_{i}^{(2)}
		\end{pmatrix}=	\begin{pmatrix} d(\xi;z_i)d_0(z_i)E(\xi,t;z_i)  \\
			\mathrm{i}d(-\xi;\hat{z}_i)d_0(\hat{z}_i)e^{F(\xi,t)}E(\xi,t;z_i) 	
		\end{pmatrix}
		+\alpha_i\begin{pmatrix}
			d(\xi;\hat{z}_i)d_0(\hat{z}_i)E(\xi,t;\hat{z}_i)\\
			\mathrm{i}d(-\xi;z_i)d_0(z_i)e^{F(\xi,t)}E(\xi,t;\hat{z}_i)
		\end{pmatrix},
	\end{align}
	where $\alpha_i\in\mathbb{C},i=1,2,\ldots,N$.
	
	Applying the Sherman-Morrison-Woodbury identity \eqref{Eq-SMW-identity} together with the $N$-fold Darboux-B\"acklund transformation $BT_0$ \eqref{Eq-BT0-1}-\eqref{Eq-BT0-2} yields the $N$-elliptic localized solution:
	\begin{equation}\label{Eq-uN-Sherman}
		\begin{split}
			u_N=\big(\tilde{u}+\mathrm{i}\phi^{(2)\dagger}\textbf{M}^{-1}\phi^{(1)}\big)_\xi=\left(\frac{\tilde{u}^{1-N}\det(\tilde{u}\textbf{M}+\mathrm{i}\phi^{(1)}\phi^{(2)\dagger})}{\det(\textbf{M})}\right)_\xi,
		\end{split}
	\end{equation}
	where $\phi^{(i)}=(\phi^{(i)}_1,\phi^{(i)}_2,\ldots,\phi^{(i)}_N)^T$ for $i=1,2$ and
	\begin{align}\label{Eq-tildeu}
		\tilde{u}=-\frac{\sqrt{\nu_0}\sigma(\kappa)\sigma(\rho+\kappa)\sigma(2\kappa)\sigma(\xi+\rho+2\kappa)}{\sigma(\rho)\sigma(\rho+3\kappa)\sigma(\xi-\kappa)}e^{-F(\xi,t)},
	\end{align}
	is the primitive function of the elliptic solution $u$. In fact, combining \eqref{Eq-DNLS-elliptic-solution}-\eqref{Eq-F}, \eqref{Eq-tildeu}, the derivative formula \eqref{Eq-derivative-sigma} and the third identity in \eqref{Eq-formulas of Weierstrass functions}, we derive
	\begin{equation}
		\begin{split}
			\tilde{u}_\xi 
			&= \left( \frac{\sigma'(\xi+\rho+2\kappa)}{\sigma(\xi+\rho+2\kappa)} - \frac{\sigma'(\xi-\kappa)}{\sigma(\xi-\kappa)} -\big(\zeta(\kappa+\rho)+\zeta(2\kappa)\big) \right)\tilde{u} \\
			&\xlongequal{\eqref{Eq-derivative-sigma}} \big( \zeta(\xi+\rho+2\kappa)-\zeta(\xi-\kappa)-\zeta(\kappa+\rho)-\zeta(2\kappa) \big)\tilde{u}\\
			&\xlongequal{\eqref{Eq-formulas of Weierstrass functions}}-\frac{\sigma(\xi+\rho)\sigma(\xi+\kappa)\sigma(3\kappa+\rho)}{\sigma(\xi-\kappa)\sigma(\kappa+\rho)\sigma(2\kappa)\sigma(\xi+\rho+2\kappa)}\tilde{u}=u,
		\end{split}
	\end{equation}
	which
	verifies that the $\xi$-derivative of $\tilde{u}$ in \eqref{Eq-tildeu} recovers the elliptic solution $u$ \eqref{Eq-DNLS-elliptic-solution}.

The squared moduli are classified into three types at the end of Subsection $3.1$ based on $\lambda^{(i)},i=1,2,3,4$. 
Since the representations of $\kappa$ and $\rho$ are different for squared moduli with $\nu_0=\nu_1$ and $\nu_0=\nu_4$, we analyze these scenarios separately. For squared moduli with $\nu_0=\nu_1$ (which may belong to any of the three types), $\kappa$ and $\rho$ are given by \eqref{Eq-kappa-rho}. Combining \eqref{Eq-ei-and-vi} and \eqref{Eq-kappa-rho} yields
\begin{equation}\label{Eq-kappa-expression1-in-proof}
	\wp(\kappa)=e_i-\frac{1}{4}\prod_{j=2,j\neq i+1}^4(\nu_1-\nu_j),\quad i=1,2,3.
\end{equation}
This implies $\wp(\kappa)<e_3$ and consequently $\Re(\kappa)=0$. Similarly, \eqref{Eq-ei-and-vi} and \eqref{Eq-kappa-rho} give
\begin{equation}\label{Eq-rho-expression1-in-proof}
	\wp(\rho)=e_i-\frac{\nu_{i+1}}{4\nu_1}\prod_{j=2,j\neq i+1}^{4}(\nu_1-\nu_i),\quad i=1,2,3.
\end{equation}
We find $e_2\leq \wp(\rho)\leq e_1$ and $\wp(\rho)\leq e_3$, leading to $\Re(\rho)=\omega_1$ for Type 1 and $\Re(\rho)=0$ for Type 2.  
For Type 3, $e_1\in\mathbb{R}$ while $e_{2,3}$ are complex conjugates. From \eqref{Eq-kappa-expression1-in-proof}-\eqref{Eq-rho-expression1-in-proof}, we deduce $\wp(\kappa)<e_1$ and $\wp(\rho)<e_1$, implying $\wp'(\kappa),\wp'(\rho)\in\mathrm{i}\mathbb{R}$ and restricting $\Re(\kappa),\Re(\rho)$ to $0$ or $\omega_1$. Specific values of $\Re(\kappa)$ and $\Re(\rho)$ follow from comparing $\wp(\kappa),\wp(\rho)$ with $\Re(e_3)$. Henceforth, the equations \eqref{Eq-kappa-expression1-in-proof}-\eqref{Eq-rho-expression1-in-proof} establish $\Re(\kappa)=0$ when $\Re(\nu_1-\nu_3)\geq0$ and $\Re(\kappa)=\omega_1$ when $\Re(\nu_1-\nu_3)<0$. Similarly, $\Re(\rho)=0$ when $\Re\big(\nu_4(\nu_1-\nu_3)\big)\geq0$ and $\Re(\rho)=\omega_1$ when $\Re\big(\nu_4(\nu_1-\nu_3)\big)<0$. 
Notably, the combination $\Re(\kappa)=\omega_1$ with $\Re(\rho)=0$ cannot occur, as demonstrated by
\begin{equation}
	\begin{split}
		\Re(\nu_1-\nu_3)&=\nu_1-\Re(\nu_3),\\
		\Re\big(\nu_4(\nu_1-\nu_3)\big)&=\Re(\nu_3)\nu_1-\Re^2(\nu_3)-\Im^2(\nu_3).
	\end{split}
\end{equation}
When $\Re(\nu_1-\nu_3)\leq0$, we immediately obtain $\Re\big(\nu_4(\nu_1-\nu_3)\big)<0$, which excludes $\Re(\kappa)=\omega_1$ and $\Re(\rho)=0$. Squared moduli with $\nu_0=\nu_4$ belong to Type 2, equations \eqref{Eq-kappa-expression1-in-proof} and \eqref{Eq-rho-expression1-in-proof} remain valid with $\nu_i,i=1,2$ exchanged by $\nu_{5-i}$. Consequently, $\Re(\kappa)=\Re(\rho)=0$ holds for this case.

In summary, the three types of squared moduli exhibit the following properties for $\kappa,\rho,\omega_1$ and $\omega_3$:
\begin{itemize}
	\item[-] Type 1: $\Re(\kappa)=0,\Re(\rho)=\omega_1,\Re(\omega_3)=0$,
	\item[-] Type 2: $\Re(\kappa)=\Re(\rho)=0,\Re(\omega_3)=0$,
	\item[-] Type 3: All the three situations can happen: $\Re(\kappa)=0,\Re(\rho)=\omega_1$, $\Re(\kappa)=\Re(\rho)=0$ and $\Re(\kappa)=\Re(\rho)=\omega_1$. Moreover, $\Re(\omega_3)\neq 0$.
\end{itemize}

To obtain the compact form of the $N$-elliptic localized solution, we express the determinant entries in \eqref{Eq-uN-Sherman} using Weierstrass functions. Subsequent computations involve frequent conjugation operations, where distinct values of $\Re(\kappa)$ and $\Re(\rho)$ yield different outcomes of the $N$-elliptic localized solutions and their asymptotics. To achieve uniform formulation, we introduce
	
	\begin{equation}\label{Eq-Ij}
		\begin{split}
			I_j=\begin{dcases}
				-1,&   \Re(\kappa)=\Re(\rho)=0,  \\
				e^{-2\zeta(\omega_1)(\kappa+\rho-\omega_1-z_j^*)},& \Re(\kappa)=0,\Re(\rho)=\omega_1,\\
				-e^{-2\zeta(\omega_1)(2\kappa+2\rho-4\omega_1-2z_j^*)},& \Re(\kappa)=\Re(\rho)=\omega_1,\\   
			\end{dcases}\quad    I_0(\xi)=\begin{dcases}
				1,&   \Re(\kappa)=\Re(\rho)=0,\\
				e^{2\zeta(\omega_1)\xi},&    \Re(\kappa)=0,\Re(\rho)=\omega_1,\\
				e^{4\zeta(\omega_1)\xi},&    \Re(\kappa)=\Re(\rho)=\omega_1.\\   
			\end{dcases}
		\end{split}
	\end{equation}
	These preparations yield a compact expression for the $N$-elliptic localized solutions to the DNLS equation \eqref{Eq-DNLS-equation}, as stated in the following theorem. The solution is derived from the $N$-fold Darboux-B\"acklund transformation $BT_0$ with the aforemention elliptic solution $u$ \eqref{Eq-DNLS-elliptic-solution} and $\phi_i=(\phi^{(1)}_i,\phi^{(2)}_i)^T,i=1,2,\ldots,N$.
	\begin{theorem}[$N$-elliptic localized solution]\label{thm:N-soliton solution}
		The $N$-elliptic localized solution of the DNLS equation \eqref{Eq-DNLS-equation} could be written as 
		\begin{equation}\label{Eq-N-elliptic-localized-solution}
			\begin{split}
				u_{N}(\xi,t)=\left(\Big(-\frac{\sigma(\xi+\rho+2\kappa)}{\sigma(\xi-\kappa)}\Big)^{1-N}\frac{\sqrt{\nu_0}\sigma(\kappa)\sigma(\rho+\kappa)\sigma(2\kappa)}{\sigma(\rho)\sigma(\rho+3\kappa)}\frac{\det\mathbf{B}_{N}^{(1)}}{\det\mathbf{B}_{N}^{(2)}}e^{-F(\xi,t)}\right)_\xi.
			\end{split}
		\end{equation}
		The matrices $\mathbf{B}_N^{(i)},i=1,2$ are $N\times N$, whose elements are given by
		\begin{equation}\label{Eq-N-elliptic-localized-solution-2}
			\begin{split}
				\big(\mathbf{B}_{N}^{(1)}\big)_{ij}&= \Sigma^{(1)}(\xi;z_i,z_j^*)r^{-1}(z_i)s(z_j^*)d_0(z_i)d_0^*(\hat{z}_j)E(\xi,t;z_i)E^*(\xi,t;z_j)I_j\\
				&\quad +\alpha_i\Sigma^{(1)}(\xi;\hat{z}_i,z_j^*)r(z_i)s(z_j^*)d_0(\hat{z}_i)d_0^*(\hat{z}_j)E(\xi,t;\hat{z}_i)E^*(\xi,t;z_j)I_j\\
				&\quad -\alpha_j^*\Sigma^{(1)}(\xi;z_i,\check{z}_j)r^{-1}(z_i)s^{-1}(z_j^*)d_0(z_i)d_0^*(z_j)E(\xi,t;z_i)E^*(\xi,t;\hat{z}_j)I_0(\xi)\\
				&\quad -\alpha_i\alpha_j^*\Sigma^{(1)}(\xi;\hat{z}_i,\check{z}_j)r(z_i)s^{-1}(z_j^*)d_0(\hat{z}_i)d_0^*(z_j)E(\xi,t;\hat{z}_i)E^*(\xi,t;\hat{z}_j)I_0(\xi),\\
				\big(\mathbf{B}_N^{(2)}\big)_{ij}&=\Sigma^{(2)}(\xi;z_i,z_j^*)d_0(z_i)d_0^*(\hat{z}_j)E(\xi,t;z_i)E^*(\xi,t;z_j)I_j\\
				&\quad +\alpha_i\Sigma^{(2)}(\xi;\hat{z}_i,z_j^*)d_0(\hat{z}_i)d_0^*(\hat{z}_j)E(\xi,t;\hat{z}_i)E^*(\xi,t;z_j)I_j\\
				&\quad -\alpha_j^*\Sigma^{(2)}(\xi;z_i,\check{z}_j)d_0(z_i)d_0^*(z_j)E(\xi,t;z_i)E^*(\xi,t;\hat{z}_j)I_0(\xi)\\
				&\quad -\alpha_i\alpha_j^*\Sigma^{(2)}(\xi;\hat{z}_i,\check{z}_j)d_0(\hat{z}_i)d_0^*(z_j)E(\xi,t;\hat{z}_i)E^*(\xi,t;\hat{z}_j)I_0(\xi),
			\end{split}
		\end{equation}
		where $\check{w}=\kappa+\rho-w^*$ and
		\begin{equation}\label{Eq-N-elliptic-localized-solution-3}
			\begin{split}
				r(\blacktriangle)&=\frac{\sigma(\kappa-\blacktriangle)}{\sigma(\rho+2\kappa+\blacktriangle)},\quad s(\blacktriangle)=\frac{\sigma(\kappa+\blacktriangle)}{\sigma(2\kappa+\rho-\blacktriangle)},\\
				\Sigma^{(1)}(\xi;\blacktriangle,\bullet) &= -\frac{\sigma(\blacktriangle-\kappa)\sigma(\blacktriangle+\bullet-\rho-2\kappa-\xi)\sigma(\rho+2\kappa-\bullet)}{\sigma(\blacktriangle+\bullet)}, \\
				\Sigma^{(2)}(\xi;\blacktriangle,\bullet) &= \frac{\sigma(\blacktriangle-\kappa)\sigma(\blacktriangle+\bullet+\kappa-\xi)\sigma(\rho+2\kappa-\bullet)}{\sigma(\blacktriangle+\bullet)}.\\	
			\end{split}
		\end{equation}
	\end{theorem}
	\begin{proof}
	The proof proceeds as follows. First, we evaluate the numerator and denominator of $\mathbf{M}_{ij}$ separately to obtain an expression for $\mathbf{M}_{ij}$. Direct evaluation of $(\mathrm{i}\phi^{(1)}\phi^{(2)\dagger})_{ij}$ via \eqref{Eq-phi-l} then yields the representation for $\left(\tilde{u}\mathbf{M}+\mathrm{i}\phi^{(1)}\phi^{(2)\dagger}\right)_{ij}$. Combining the computation results with \eqref{Eq-uN-Sherman} and the representation for $\tilde{u}$ \eqref{Eq-tildeu} ultimately gives the $N$-elliptic localized solutions.
	
		From \eqref{Eq-BT0-2},
we have $	\textbf{M}_{ij}
=\frac{2\lambda_i^{-1}\big(\phi_{j}^{(2)*}\phi_{i}^{(2)}\big)+2(\lambda_j^{*})^{-1}\big(\phi_{j}^{(1)*}\phi_{i}^{(1)}\big)}{\lambda_i^{-2}-(\lambda_j^{*})^{-2}},\quad 1\leq i,j\leq N$. Using the solution forms \eqref{Eq-phi-l}, the numerator of $\mathbf{M}_{ij}$ evaluates to
	\begin{equation}\label{Eq-Mij-numerator}
		\begin{split}
			&\quad\quad 2\left((\lambda_j^{-1})^*\phi_{j}^{(1)*}\phi_{i}^{(1)}+\lambda_i^{-1}\phi_{j}^{(2)*}\phi_{i}^{(2)}\right)\\
			&=2\left(	(\lambda_j^{-1})^*\left(E^*(\xi,t;z_j)d^*(\xi;z_j)d_0^*(z_j)+\alpha_j^*E^*(\xi,t;\hat{z}_j)d^*(\xi;\hat{z}_j)d_0^*(\hat{z}_j)\right)\right.\\
			&\quad \left(E(\xi,t;z_i)d(\xi;z_i)d_0(z_i)+\alpha_iE(\xi,t;\hat{z}_i)d(\xi;\hat{z}_i)d_0(\hat{z}_i)\right)\\
			&\quad +(\lambda_i^{-1})\left(E^*(\xi,t;z_j)d^*(-\xi;\hat{z}_j)d_0^*(\hat{z}_j)+\alpha_j^*E^*(\xi,t;\hat{z}_j)d^*(-\xi;z_j)d_0^*(z_j)\right)\\ &\left.\quad \left(E(\xi,t;z_i)d(-\xi;\hat{z}_i)d_0(\hat{z}_i)+\alpha_iE(\xi,t;\hat{z}_i)d(-\xi;z_i)d_0(z_i)\right)e^{F(\xi,t)+F^*(\xi,t)}\right).
		\end{split}
	\end{equation}
Expanding \eqref{Eq-Mij-numerator} expresses the numerator of $\mathbf{M}_{ij}$ as a sum of four terms. Different values of $\Re(\kappa)$ and $\Re(\rho)$ lead to distinct computation results for the components involving complex conjugation operations. We illustrate the computation for the first term as an example. Specifically, combining \eqref{Eq-F}, \eqref{Eq-nu0-d0} and \eqref{Eq-solution to the Lax pair-2} gives:
		\begin{equation*}\label{Eq-relations}
			\begin{split}
				2\Re(F(\xi,t))&=\begin{dcases}
					0,&   \Re(\kappa)=\Re(\rho)=0,\\
					2\zeta(\omega_1)\xi,&    \Re(\kappa)=0,\Re(\rho)=\omega_1,\\
					8\zeta(\omega_1)\xi,&    \Re(\kappa)=\Re(\rho)=\omega_1,\\
				\end{dcases}\quad 	d^*(\xi;z_j)=\begin{dcases}
					\frac{\sigma(\xi-z_j^*)}{\sigma(-\kappa-\xi)},&   \Re(\kappa)=0,\\
					-\frac{\sigma(\xi-z_j^*)}{\sigma(-\kappa-\xi)}e^{2\zeta(\omega_1)(\kappa+\xi-\omega_1)},&    \Re(\kappa)=\omega_1,\\
				\end{dcases}\\
				d^*(-\xi;\hat{z}_j)&=\begin{dcases}
					\frac{\sigma(-\xi-\kappa-\rho+z_j^*)}{\sigma(-\kappa+\xi)},&   \Re(\kappa)=\Re(\rho)=0,\\
					-\frac{\sigma(-\xi-\kappa-\rho+z_j^*)}{\sigma(-\kappa+\xi)}e^{2\zeta(\omega_1)(-\xi-\kappa-\rho+\omega_1+z_j^*)},&    \Re(\kappa)=0,\Re(\rho)=\omega_1,\\
					-\frac{\sigma(-\xi-\kappa-\rho+z_j^*)}{\sigma(-\kappa+\xi)}e^{2\zeta(\omega_1)(-3\xi-\kappa-2\rho+3\omega_1+2z_j^*)},&    \Re(\kappa)=\Re(\rho)=\omega_1.\\
				\end{dcases}\\
			\end{split}
		\end{equation*}
		When $\Re(\kappa)=0$, we further observe that 
		\begin{equation}\label{Eq-relations-2}
			\begin{split}
				2\Re(F(\xi,t))=\ln\big(I_0(\xi)\big),\quad 	
				d^*(-\xi;\hat{z}_j)=-\frac{\sigma(-\xi-\kappa-\rho+z_j^*)}{\sigma(-\kappa+\xi)}I_0^{-1}(\xi)I_j,
			\end{split}
		\end{equation}
		establishing uniform representations through $I_0(\xi)$ and $I_j$ \eqref{Eq-Ij}. Equations \eqref{Eq-parameterization-lambda-y} and \eqref{Eq-Ij} imply
		\begin{equation}\label{Eq-lambda-conjugate}
			\begin{split}
				&\quad (\lambda^{*}_j)^{-1}
				=-\frac{2\sqrt{\nu_0}\sigma(-\rho+z_j^*)\sigma(\rho+\kappa)\sigma(\kappa)d_0^*(\hat{z}_j)}{\sigma(-\kappa-z_j^*)\sigma(\rho)d_0^*(z_j)}I_j.
			\end{split}
		\end{equation}
		Using \eqref{Eq-relations-2}-\eqref{Eq-lambda-conjugate},
		we derive	\begin{equation}\label{Eq-Mij-denominator-1-part1}
			\begin{split}
				&\quad (\lambda_j^{-1})^*d^*(\xi;z_j)d(\xi;z_i)d_0^*(z_j)d_0(z_i)\\
				&=-\frac{2\sqrt{\nu_0}\sigma(-\rho+z_j^*)\sigma(\rho+\kappa)\sigma(\kappa)d_0^*(\hat{z}_j)I_j}{\sigma(-\kappa-z_j^*)\sigma(\rho)d_0^*(z_j)}	\frac{\sigma(\xi-z_j^*)}{\sigma(-\kappa-\xi)}\frac{\sigma(\xi-z_i)}{\sigma(\kappa-\xi)}d_0^*(z_j)d_0(z_i).
			\end{split}
		\end{equation}
		Applying \eqref{Eq-parameterization-lambda-y} and \eqref{Eq-relations-2} gives    
		\begin{equation}\label{Eq-Mij-denominator-1-part2}
			\begin{split}
				&\quad (\lambda_i^{-1})d^*(-\xi;\hat{z}_j)d(-\xi;\hat{z}_i)d_0^*(\hat{z}_j)d_0(\hat{z}_i)e^{F(\xi,t)+F^*(\xi,t)}\\
				&=\frac{2\sqrt{\nu_0}\sigma(\rho+z_i)\sigma(-\rho-\kappa)\sigma(\kappa)}{\sigma(z_i-\kappa)\sigma(\rho)}\left(\frac{d_0(\hat{z}_i)}{d_0(z_i)}\right)^2\frac{\sigma(-\xi-\kappa-\rho+z_j^*)}{\sigma(-\kappa+\xi)}\frac{\sigma(-\xi-\hat{z}_i)}{\sigma(\kappa+\xi)}I_jd_0(z_i)d_0^*(\hat{z}_j).
			\end{split}
		\end{equation}
		Combining \eqref{Eq-Mij-denominator-1-part1} and \eqref{Eq-Mij-denominator-1-part2} yields
		\begin{equation}\label{Eq-Mij-denominator-1}
			\begin{split}
				& \quad (\lambda_j^{-1})^*d^*(\xi;z_j)d(\xi;z_i)d_0^*(z_j)d_0(z_i)+(\lambda_i^{-1})d^*(-\xi;\hat{z}_j)d(-\xi;\hat{z}_i)d_0^*(\hat{z}_j)d_0(\hat{z}_i)I_0(\xi)\\
				&=-\frac{2\sqrt{\nu_0}\sigma(\kappa)\sigma(\rho+\kappa)d_0(z_i)d_0^*(\hat{z}_j)I_j}{\sigma(\rho)\sigma(\xi+\kappa)\sigma(\xi-\kappa)}\\
				&\quad \left(\frac{\sigma(-\rho+z_j^*)\sigma(z_j^*-\xi)\sigma(z_i-\xi)}{\sigma(-\kappa-z_j^*)}-\frac{\sigma(\rho+z_i)\sigma(\kappa+\rho+z_i-\xi)\sigma(\kappa+\rho-z_j^*+\xi)d_0^2(\hat{z}_i)}{\sigma(z_i-\kappa)d_0^2(z_i)}\right)\\
				&\xlongequal{\eqref{Eq-addition formulas of the sigma functions}}-\frac{2\sqrt{\nu_0}\sigma(\kappa)\sigma(\rho+\kappa)d_0(z_i)d_0^*(\hat{z}_j)I_j}{\sigma(\rho)\sigma(\xi+\kappa)\sigma(\xi-\kappa)\sigma(-\kappa-z_j^*)\sigma(z_i+2\kappa+\rho)}\\
				&\quad\left(\sigma(-\rho+z_j^*)\sigma(z_j^*-\xi)\sigma(z_i-\xi)\sigma(z_i+2\kappa+\rho)+\sigma(\kappa+z_i)\sigma(\xi-z_j^*+\kappa+\rho)\sigma(-\kappa-\rho-z_i+\xi)\sigma(-\kappa-z_j^*)\right)\\
				&=\frac{2\sqrt{\nu_0}\sigma^2(\rho+\kappa)\sigma(\kappa)\sigma(z_i+z_j^*+\kappa-\xi)\sigma(-z_i+z_j^*-\kappa-\rho)}{\sigma(\rho)\sigma(\xi-\kappa)\sigma(-\kappa-z_j^*)\sigma(z_i+2\kappa+\rho)}d_0(z_i)d_0^*(\hat{z}_j)I_j,
			\end{split}
		\end{equation}
		where the addition formula \eqref{Eq-addition formulas of the sigma functions} is applied. Analogously, we obtain
		\begin{equation}\label{Eq-Mij-denominator2-3-4}
			\begin{split}
				&	(\lambda_j^{-1})^*d^*(\xi;z_j)d(\xi;\hat{z}_i)d_0^*(z_j)d_0(\hat{z}_i)+(\lambda_i^{-1})d^*(-\xi;\hat{z}_j)d(-\xi;z_i)d_0(z_i)d_0^*(\hat{z}_j)I_0(\xi)\\
				&=\frac{2\sqrt{\nu_0}\sigma^2(\rho+\kappa)\sigma(\kappa)}{\sigma(\rho)\sigma(\xi-\kappa)}\frac{\sigma(z_j^*+z_i)\sigma(z_j^*-z_i-\rho-\xi)}{\sigma(\kappa-z_i)\sigma(-\kappa-z_j^*)}d_0(\hat{z}_i)d_0^*(\hat{z}_j)I_j,\\
				&	(\lambda_j^{-1})^*d^*(\xi;\hat{z}_j)d(\xi;z_i)d_0^*(\hat{z}_j)d_0(z_i)+(\lambda_i^{-1})d^*(-\xi;z_j)d(-\xi;\hat{z}_i)d_0(\hat{z}_i)d_0^*(z_j)I_0(\xi)\\
				&=\frac{2\sqrt{\nu_0}\sigma^2(\rho+\kappa)\sigma(\kappa)}{\sigma(\rho)\sigma(\xi-\kappa)}\frac{\sigma(z_j^*+z_i)\sigma(-z_j^*+z_i+\rho-\xi+2\kappa)}{\sigma(\rho+2\kappa+z_i)\sigma(-\rho-2\kappa+z_j^*)}d_0(z_i)d_0^*(z_j)I_0(\xi),\\
				&	(\lambda_j^{-1})^*d^*(\xi;\hat{z}_j)d(\xi;\hat{z}_i)d_0^*(\hat{z}_j)d_0(\hat{z}_i)+(\lambda_i^{-1})d^*(-\xi;z_j)d(-\xi;z_i)d_0(z_i)d_0^*(z_j)I_0(\xi)\\
				&=\frac{2\sqrt{\nu_0}\sigma^2(\rho+\kappa)\sigma(\kappa)}{\sigma(\rho)\sigma(\xi-\kappa)}\frac{\sigma(z_j^*-z_i-\rho-\kappa)\sigma(-z_j^*-z_i+\kappa-\xi)}{\sigma(\kappa-z_i)\sigma(-\rho-2\kappa+z_j^*)}d_0(\hat{z}_i)d_0^*(z_j)I_0(\xi).
			\end{split}
		\end{equation}
		The complex conjugate of \eqref{Eq-mu-parameterization} combined with \eqref{Eq-addition formulas of the sigma functions} provides
		\begin{equation}
			\begin{split}
				\lambda_i^{-2}-	(\lambda_j^*)^{-2}=\frac{4\mathrm{i}\sigma^2(\rho+\kappa)\sigma^2(2\kappa)\sigma(\kappa+\rho+z_i-z_j^*)\sigma(z_j^*+z_i)}{\sigma(\kappa-\rho)\sigma(2\kappa+\rho+z_i)\sigma(2\kappa+\rho-z_j^*)\sigma(\kappa-z_i)\sigma(\kappa+z_j^*)}.
			\end{split}
		\end{equation}
		Combining \eqref{Eq-Mij-numerator} and \eqref{Eq-Mij-denominator-1}-\eqref{Eq-Mij-denominator2-3-4} yields 
		\begin{equation}\label{Eq-Mij}
			\begin{split}
					\mathbf{M}_{ij}&=-\frac{\sqrt{\nu_0}\sigma(\kappa)\sigma(\kappa-\rho)}{\mathrm{i}\sigma(\rho)\sigma(\xi-\kappa)\sigma^2(2\kappa)}\mathbf{M}_{0,ij},\\
				\mathbf{M}_{0,ij}&=\Big(\frac{\sigma(-z_i-z_j^*-\kappa+\xi)\sigma(2\kappa+\rho-z_j^*)\sigma(\kappa-z_i)}{\sigma(z_i+z_j^*)}d_0(z_i)d_0^*(\hat{z}_j)E(\xi,t;z_i)E^*(\xi,t;z_j)I_j\\
				&\quad +\alpha_i\frac{\sigma(-z_i+z_j^*-\rho-\xi)\sigma(2\kappa+\rho+z_i)\sigma(2\kappa+\rho-z_j^*)}{\sigma(\kappa+\rho+z_i-z_j^*)}d_0(\hat{z}_i)d_0^*(\hat{z}_j)E(\xi,t;\hat{z}_i)E^*(\xi,t;z_j)I_j\\
				&\quad +\alpha_j^*\frac{\sigma(\kappa-z_i)\sigma(\kappa+z_j^*)\sigma(z_i-z_j^*+2\kappa+\rho-\xi)}{\sigma(\kappa+\rho+z_i-z_j^*)}d_0(z_i)d_0^*(z_j)E(\xi,t;z_i)E^*(\xi,t;\hat{z}_j)I_0(\xi)\\
				&\quad -\alpha_i\alpha_j^*\frac{\sigma(2\kappa+\rho+z_i)\sigma(\kappa+z_j^*)\sigma(-z_i-z_j^*+\kappa-\xi)}{\sigma(z_i+z_j^*)}d_0(\hat{z}_i)d_0^*(z_j)E(\xi,t;\hat{z}_i)E^*(\xi,t;\hat{z}_j)I_0(\xi)\Big).
			\end{split}
		\end{equation}
		Direct evaluation of \eqref{Eq-phi-l} provides
		\begin{equation}\label{Eq-phiab}
			\begin{split}
				\quad \left(\mathrm{i}\phi^{(1)}\phi^{(2)\dagger}\right)_{ij}
				&=e^{-F(\xi,t)}\left(-\frac{\sigma(z_i-\xi)\sigma(\kappa+\rho-z_j^*+\xi)}{\sigma(\xi-\kappa)\sigma(-\xi+\kappa)}d_0(z_i)d_0^*(\hat{z}_j)E(\xi,t;z_i)E^*(\xi,t;z_j)I_j\right.\\
				&\quad -\alpha_i\frac{\sigma(-\kappa-\rho-z_i-\xi)\sigma(\kappa+\rho-z_j^*+\xi)}{\sigma(\xi-\kappa)\sigma(-\xi+\kappa)}d_0(\hat{z}_i)d_0^*(\hat{z}_j)E(\xi,t;\hat{z}_i)E^*(\xi,t;z_j)I_j\\
				&	\quad +\alpha_j^*\frac{\sigma(z_i-\xi)\sigma(z_j^*+\xi)}{\sigma(\xi-\kappa)\sigma(-\xi+\kappa)}d_0(z_i)d_0^*(z_j)E(\xi,t;z_i)E^*(\xi,t;\hat{z}_j)I_0(\xi)\\
				&\quad \left.+\alpha_i\alpha_j^*\frac{\sigma(-\kappa-\rho-z_i-\xi)\sigma(z_j^*+\xi)}{\sigma(\xi-\kappa)\sigma(-\xi+\kappa)}d_0(\hat{z}_i)d_0^*(z_j)E(\xi,t;\hat{z}_i)E^*(\xi,t;\hat{z}_j)I_0(\xi)\right),
			\end{split}
		\end{equation}
		Combining with \eqref{Eq-tildeu}, we have \begin{equation}\label{Eq-uM-plus-phi12}
			\begin{split}
				&\quad \left(\tilde{u}\mathbf{M}+\mathrm{i}\phi^{(1)}\phi^{(2)\dagger}\right)_{ij}=\frac{e^{-F(\xi,t)}}{\sigma(\xi-\kappa)\sigma(\rho+3\kappa)}\\
				&\quad \Big(\frac{\sigma(\rho+2\kappa+z_i)\sigma(z_i+z_j^*-\rho-2\kappa-\xi)\sigma(\kappa+z_j^*)}{\sigma(z_i+z_j^*)}d_0(z_i)d_0^*(\hat{z}_j)E(\xi,t;z_i)E^*(\xi,t;z_j)I_j\\
				&\quad +\alpha_i\frac{\sigma(\kappa-z_i)\sigma(z_i-z_j^*+2\rho+3\kappa+\xi)\sigma(\kappa+z_j^*)}{\sigma(\kappa+\rho+z_i-z_j^*)}d_0(\hat{z}_i)d_0^*(\hat{z}_j)E(\xi,t;\hat{z}_i)E^*(\xi,t;z_j)I_j\\
				&\quad +\alpha_j^*\frac{\sigma(\xi+z_j^*-z_i+\kappa)\sigma(\rho+2\kappa-z_j^*)\sigma(\rho+2\kappa+z_i)}{\sigma(\kappa+\rho+z_i-z_j^*)}d_0(z_i)d_0^*(z_j)E(\xi,t;z_i)E^*(\xi,t;\hat{z}_j)I_0(\xi)\\
				&\quad +\alpha_i\alpha_j^*\frac{\sigma(\kappa-z_i)\sigma(z_j^*+z_i+\xi+\rho+2\kappa)\sigma(z_j^*-\rho-2\kappa)}{\sigma(z_i+z_j^*)}d_0(\hat{z}_i)d_0^*(z_j)E(\xi,t;\hat{z}_i)E^*(\xi,t;\hat{z}_j)I_0(\xi)\Big),
			\end{split}
		\end{equation}
	by applying \eqref{Eq-addition formulas of the sigma functions} for another time.  Combining  \eqref{Eq-tildeu} and \eqref{Eq-uM-plus-phi12} completes the proof for $\Re(\kappa)=0$. 
		The case $\Re(\kappa)=\omega_1$ follows analogously through similar computations. We omit further details due to procedural parallels.
	\end{proof}
	\begin{rmk}
		The addition formula for Weierstrass sigma functions \eqref{Eq-addition formulas of the sigma functions} plays a crucial role in proving Theorem \ref{thm:N-soliton solution}. We present a direct method for applying this formula. This algorithm determines whether the addition formula applies to summations of the form 
		$\sigma(\theta_1)\sigma(\theta_2)\sigma(\theta_3)\sigma(\theta_4)+\sigma(\vartheta_1)\sigma(\vartheta_2)\sigma(\vartheta_3)\sigma(\vartheta_4)$ 
		while simultaneously yielding the application result. Define two sets:
		\begin{equation}
			\begin{split}
				A:&=\left\{\frac{\theta_1+\theta_2+\theta_3+\theta_4}{2},\frac{-\theta_1-\theta_2+\theta_3+\theta_4}{2},\frac{-\theta_1+\theta_2-\theta_3+\theta_4}{2},\frac{-\theta_1+\theta_2+\theta_3-\theta_4}{2}\right\},\\
				B:&=\left\{\frac{-\theta_1+\theta_2+\theta_3+\theta_4}{2},\frac{\theta_1-\theta_2+\theta_3+\theta_4}{2},\frac{\theta_1+\theta_2-\theta_3+\theta_4}{2},\frac{\theta_1+\theta_2+\theta_3-\theta_4}{2}\right\}.
			\end{split}
		\end{equation}
		When $A=\{\vartheta_i,i=1,2,3,4\}$, the sum reduces to $\prod_{b\in B}\sigma(b)$. If three elements of $B$ and $\{\vartheta_i,i=1,2,3,4\}$ are identical with a single sign difference in the remaining element, the sum becomes $-\prod_{a\in A}\sigma(a)$. For example, consider the sum
		\begin{equation}\label{Eq-example-1}
			\sigma(-\rho+z_j^*)\sigma(z_j^*-\xi)\sigma(z_i-\xi)\sigma(z_i+2\kappa+\rho)+\sigma(\kappa+z_i)\sigma(-\xi+z_j^*-\kappa-\rho)\sigma(\kappa+\rho+z_i-\xi)\sigma(-\kappa-z_j^*),
		\end{equation}
		which appears in \eqref{Eq-Mij-denominator-1}. With 
		\begin{equation}
			\begin{split}
				\theta_1&=-\rho+z_j^*,\quad \theta_2=z_j^*-\xi,\quad \theta_3=z_i-\xi,\quad \theta_4=z_i+2\kappa+\rho,\\
				\vartheta_1&=\kappa+z_i,\quad \vartheta_2=-\xi+z_j^*-\kappa-\rho,\quad \vartheta_3=\kappa+\rho+z_i-\xi,\quad \vartheta_4=-\kappa-z_j^*,
			\end{split}
		\end{equation}	
		the sets $A$ and $B$ evaluate to 
		\begin{equation}\label{Eq-A-B-example-1}
			\begin{split}
				A&=\left\{z_i+z_j^*-\xi-\kappa,z_i-z_j^*+\kappa+\rho,\rho+\kappa,-\xi-\kappa\right\},\\
				B&=\left\{\kappa+\rho-\xi+z_i,z_i+\kappa,z_j^*+\kappa,-\kappa-\rho-\xi+z_j^*\right\}.
			\end{split}
		\end{equation}
		Comparing \eqref{Eq-example-1} and \eqref{Eq-A-B-example-1}, we find that
		the sum equals $\sigma(z_i+z_j^*-\xi-\kappa)\sigma(z_i-z_j^*+\kappa+\rho)\sigma(\xi+\kappa)\sigma(\rho+\kappa)$. As a second example, consider 
		\begin{equation}\label{Eq-example-2}
			\begin{split}
				&\quad\sigma(\rho+3\kappa)\sigma(\kappa+\rho-z_j^*+\xi)\sigma(\kappa+\rho+z_i-z_j^*)\sigma(-\kappa-\rho-z_i-\xi)\\
				&+\sigma(-\xi-\rho-2\kappa)\sigma(-2\kappa-\rho-z_i)\sigma(z_i-z_j^*+\rho+\xi)\sigma(2\kappa+\rho-z_j^*),
			\end{split}
		\end{equation}
		which arises in the evaluation of the second term of $\big(\tilde{u}\mathbf{M}+\mathrm{i}\phi^{(1)}\phi^{(2)\dagger}\big)_{ij}$. Taking 
		\begin{equation}
			\begin{split}
			\theta_1&=\rho+3\kappa,\quad \theta_2=\kappa+\rho-z_j^*+\xi,\quad \theta_3=\kappa+\rho+z_i-z_j^*,\quad \theta_4=-\kappa-\rho-z_i-\xi,\\
			\vartheta_1&=-\xi-\rho-2\kappa,\quad \vartheta_2=-2\kappa-\rho-z_i,\quad \vartheta_3=z_i-z_j^*+\rho+\xi,\quad \vartheta_4=2\kappa+\rho-z_j^*,
			\end{split}
		\end{equation}	
		the sets $A$ and $B$ are given by 
		\begin{equation}\label{Eq-A-B-example-2}
			\begin{split}
				A&=\left\{-\xi-\rho-2\kappa,-2\kappa-\rho-z_i,z_i-z_j^*+\rho+\xi,2\kappa+\rho-z_j^*\right\},\\
				B&=\left\{2\rho + 3\kappa + \xi + z_i - z_j^*,\kappa-z_i,\kappa-\xi,-\kappa-z_j^*\right\}.
			\end{split}
		\end{equation} 
		Comparing \eqref{Eq-example-2} and \eqref{Eq-A-B-example-2}, we obtain $\sigma(2\rho + 3\kappa + \xi + z_i - z_j^*)\sigma(\kappa-z_i)\sigma(\kappa-\xi)\sigma(-\kappa-z_j^*)$.
	\end{rmk}

	Previous works \cite{Ling-NLS,LS-mKdV-solution} have shown that $N$-elliptic localized solutions for many equations within the AKNS hierarchy can be expressed as ratios of determinants with theta function entries. Here, we establish that $N$-elliptic localized solutions to the DNLS equation \eqref{Eq-DNLS-equation} in the Kaup-Newell hierarchy take the form of derivatives of determinant ratios with Weierstrass function entries. This allows direct comparison with earlier results. We also conjecture that other systems in the Kaup-Newell hierarchy admit $N$-elliptic localized solutions in this form, which merits further study.

	\section{Asymptotic analysis of the $N$-elliptic localized solutions}
	In this section, we conduct the asymptotic analysis on the $N$-elliptic localized solution \eqref{Eq-N-elliptic-localized-solution}-\eqref{Eq-N-elliptic-localized-solution-3} derived from the $N$-fold Darboux-B\"acklund transformation $BT_0$ \eqref{Eq-BT0-1}-\eqref{Eq-BT0-2}. First, we establish Theorem \ref{thm:Cauchy matrix}, which provides a Weierstrass sigma version of Cauchy determinants essential for our analysis. As a generalization of Cauchy matrices, previous work \cite{Takahashi} introduced a theta-type Cauchy matrix that retains a structure similar to the classical Cauchy matrix with entries replaced by Jacobi theta functions, along with its determinant formula. This result was later applied to asymptotic analysis of elliptic localized solutions for the NLS equation \cite{Ling-NLS}. We introduce in the present work a Weierstrass sigma version of Cauchy matrix and rigorously prove its corresponding determinant formula by induction. This result is concluded in the following theorem.
	\begin{theorem}[	Sigma version of Cauchy determinants]\label{thm:Cauchy matrix}
	Assume $\tau$, $m_1,\ldots m_N,n_1,\ldots n_N\in\mathbb{C}$ and $N\in\mathbb{Z}$. Then the determinant identity
	\begin{align*}
		\det\left(\frac{\sigma(\tau+m_i+n_j)}{\sigma(m_i+n_j)}\right)_{1\leq i,j\leq N}=\frac{\sigma(\tau)^{N-1}\sigma\big(\tau+\sum_{i=1}^{N}(m_i+n_i)\big)\prod_{1\leq i<j\leq N}\sigma(m_i-m_j)\sigma(n_i-n_j)}{\prod_{i,j=1}^N\sigma(m_i+n_j)}
	\end{align*}
	holds.
	\end{theorem}
	
	\begin{proof}
		We prove this theorem by mathematical induction. 
		For $M=1$, the formula follows from direct computation. For
		$M=2$, the formula results from the addition formula  \eqref{Eq-addition formulas of the sigma functions}. Specifically, we have
		\begin{equation}\label{Eq-sigma-function-determinant-formulas}
			\begin{split}
				&\quad \det\left(\frac{\sigma(\tau+m_i+m_j)}{\sigma(m_i+n_j)}\right)_{1\leq i<j\leq 2}    \\
				&=\frac{\sigma(\tau+m_1+n_1)\sigma(\tau+m_2+n_2)}{\sigma(m_1+n_1)\sigma(m_2+n_2)}-\frac{\sigma(\tau+m_1+n_2)\sigma(\tau+m_2+n_1)}{\sigma(m_1+n_2)\sigma(m_2+n_1)}\\
				&=\frac{\sigma(\tau+m_1+n_1+m_2+n_2)\sigma(\tau)\sigma(m_1-m_2)\sigma(n_1-n_2)}{\sigma(m_1+n_2)\sigma(m_2+n_1)\sigma(m_1+n_1)\sigma(m_2+n_2)}.
			\end{split}
		\end{equation}
		Assume that the induction hypothesis holds for $M\leq N-1$, applying the Desnanot-Jacobi determinant identity yields
		\begin{equation}\label{Eq-Desnanot-Jacobi}
			\begin{split}
				&\quad D_N(\tau;m_1,\ldots,m_N,n_1,\ldots,n_N)D_{N-2}(\tau;m_1,\ldots,m_{N-2},n_1,\ldots,n_{N-2})\\
				&=D_{N-1}(\tau;m_1,\ldots,m_{N-1},n_1,\ldots,n_{N-1})D_{N-1}(\tau;m_1,\ldots,m_{N-2},m_{N},n_1,\ldots,n_{N-2},n_{N})\\
				&\quad  -D_{N-1}(\tau;m_1,\ldots,m_{N-1},n_1,\ldots,n_{N-2},n_N)D_{N-1}(\tau;m_1,\ldots,m_{N-2},m_{N},n_1,\ldots,n_{N-1}),
			\end{split}
		\end{equation}
		where $D_N(\tau;m_1,\ldots,m_N,n_1,\ldots,n_N)$ denotes the sigma version of Cauchy determinant with parameters $m_1,\ldots,m_N$ and $n_1,\ldots,n_N$. Based on the induction hypothesis, the determinants with subscripts $N-2$ and $N-1$ can be evaluated as  
		\begin{equation}\label{Eq-DN}
			\begin{split}
				&D_{N-2}(\tau;m_1,\ldots,m_{N-2},n_1,\ldots,n_{N-2})=\frac{\sigma(\tau)^{N-3}\sigma\big(\tau+\sum_{i=1}^{N-2}(m_i+n_i)\big)P^{(1)}}{P^{(2)}},\\
				&D_{N-1}(\tau;m_1,\ldots,m_{N-1},n_1,\ldots,n_{N-1})=\frac{\sigma(\tau)^{N-2}\sigma\big(\tau+\sum_{i=1}^{N-1}(m_i+n_i)\big)P^{(1)}P^{(3)}_{N-1,N-1}}{P^{(2)}P^{(4)}_{N-1,N-1}\sigma(m_{N-1}+n_{N-1})},\\
				&D_{N-1}(\tau;m_1,\ldots,m_{N-2},m_N,n_1,\ldots,n_{N-2},n_N)=\frac{\sigma(\tau)^{N-2}\sigma\big(\tau+\sum_{i=1}^{N-2}(m_i+n_i)+m_N+n_N\big)P^{(1)}P^{(3)}_{N,N}}{P^{(2)}P^{(4)}_{N,N}\sigma(m_{N}+n_N)},\\
				&D_{N-1}(\tau;m_1,\ldots,m_{N-1},n_1,\ldots,n_{N-2},n_N)=\frac{\sigma(\tau)^{N-2}\sigma\big(\tau+\sum_{i=1}^{N-2}(m_i+n_i)+m_{N-1}+n_N\big)P^{(1)}P^{(3)}_{N-1,N}}{P^{(2)}P^{(4)}_{N-1,N}\sigma(m_{N-1}+n_N)},\\
				&D_{N-1}(\tau;m_1,\ldots,m_{N-2},m_N,n_1,\ldots,n_{N-1})=\frac{\sigma(\tau)^{N-2}\sigma\big(\tau+\sum_{i=1}^{N-2}(m_i+n_i)+m_{N}+n_{N-1}\big)P^{(1)}P^{(3)}_{N,N-1}}{P^{(2)}P^{(4)}_{N,N-1}\sigma(m_{N}+n_{N-1})},
			\end{split}
		\end{equation}
		where
		\begin{equation}\label{Eq-Pi}
			\begin{split}
				P^{(1)}&=\prod_{1\leq i<j\leq N-2}\sigma(m_i-m_j)\sigma(n_i-n_j),\quad P^{(2)}=\prod_{i,j=1}^{N-2}\sigma(m_i+n_j),\\
				P^{(3)}_{k,l}&=\prod_{i=1}^{N-2}\sigma(m_i-m_k)\sigma(n_i-n_l),\quad 	P^{(4)}_{k,l}=\prod_{i=1}^{N-2}\sigma(m_k+n_i)\sigma(m_i+n_l),\quad k,l=N-1,N.
			\end{split}
		\end{equation}
		Combining with \eqref{Eq-DN}-\eqref{Eq-Pi}, we obtain that
		\begin{equation}\label{Eq-part1-a}
			\begin{split}
				&\quad  \frac{D_{N-1}(\tau;m_1,\ldots,m_{N-1},n_1,\ldots,n_{N-1})D_{N-1}(\tau;m_1,\ldots,m_{N-2},m_N,n_1,\ldots,n_{N-2},n_N)}{D_{N-2}(\tau;m_1,\ldots,m_{N-2},n_1,\ldots,n_{N-2})}\\
				&=\frac{\sigma^{N-1}(\tau)\sigma\big(\tau+\sum_{i=1}^{N-1}(m_i+n_i)\big)\sigma\big(\tau+\sum_{i=1}^{N-2}(m_i+n_i)+m_N+n_N\big)P^{(1)}P^{(3)}_{N-1,N-1}P^{(3)}_{N,N}}{\sigma\big(\tau+\sum_{i=1}^{N-2}(m_i+n_i)\big)\sigma(m_{N-1}+n_{N-1})\sigma(m_N+n_N)P^{(2)}P^{(4)}_{N-1,N-1}P^{(4)}_{N,N}},\\
				&\quad  \frac{D_{N-1}(\tau;m_1,\ldots,m_{N-1},n_1,\ldots,n_{N-2},n_N)D_{N-1}(\tau;m_1,\ldots,m_{N-2},m_N,n_1,\ldots,n_{N-1})}{D_{N-2}(\tau;m_1,\ldots,m_{N-2},n_1,\ldots,n_{N-2})}\\
				&=\frac{\sigma^{N-1}(\tau)\sigma\big(\tau+\sum_{i=1}^{N-2}(m_i+n_i)+m_{N-1}+n_N\big)\sigma\big(\tau+\sum_{i=1}^{N-2}(m_i+n_i)+m_N+n_{N-1}\big)P^{(1)}P^{(3)}_{N-1,N}P^{(3)}_{N,N-1}}{\sigma\big(\tau+\sum_{i=1}^{N-2}(m_i+n_i)\big)\sigma(m_{N-1}+n_N)\sigma(m_N+n_{N-1})P^{(2)}P^{(4)}_{N-1,N}P^{(4)}_{N,N-1}}.
			\end{split}
		\end{equation}
Using the addition formula \eqref{Eq-addition formulas of the sigma functions},
		we obtain
		\begin{equation}\label{Eq-adaption-of-addition-formula}
			\begin{split}
				&\quad\sigma\big(\tau+\sum_{i=1}^{N-1}(m_i+n_i)\big)\sigma\big(\tau+\sum_{i=1}^{N-2}(m_i+n_i)+m_N+n_N\big)\sigma(m_{N-1}+n_N)\sigma(m_N+n_{N-1})\\
				&-\sigma\big(\tau+\sum_{i=1}^{N-2}(m_i+n_i)+m_{N-1}+n_N\big)\sigma\big(\tau+\sum_{i=1}^{N-2}(m_i+n_i)+m_N+n_{N-1}\big)\sigma(m_{N-1}+n_{N-1})\sigma(m_N+n_N)\\
				&=\sigma\big(\tau+\sum_{i=1}^{N}(m_i+n_i)\big)\sigma\big(\tau+\sum_{i=1}^{N-2}(m_i+n_i)\big)\sigma(m_{N-1}-m_{N})\sigma(n_{N-1}-n_{N}).
			\end{split}
		\end{equation}
		Moreover, we notice that 
		\begin{equation}\label{Eq-Pi-equivalence}
			P^{(i)}_{N-1,N}	P^{(i)}_{N,N-1}=	P^{(i)}_{N-1,N-1}	P^{(i)}_{N,N},\quad i=3,4.
		\end{equation}
		Besides, we can verify
		\begin{equation}
			\begin{split}
				&\sigma(m_{N-1}-m_{N})\sigma(n_{N-1}-n_{N})P^{(1)}P^{(3)}_{N-1,N-1}P^{(3)}_{N,N}=\prod_{1\leq i<j\leq N}\sigma(m_i-m_j)\sigma(n_i-n_j),\\
				&\sigma(m_{N-1}+n_{N-1})\sigma(m_N+n_N)\sigma(m_{N-1}+n_{N})\sigma(m_N+n_{N-1})P^{(2)}P^{(4)}_{N-1,N-1}P^{(4)}_{N,N}=\prod_{i,j=1}^{N}\sigma(m_i+n_j),
			\end{split}
		\end{equation}
		through direct verification.
		Combining \eqref{Eq-Desnanot-Jacobi} with \eqref{Eq-part1-a}-\eqref{Eq-Pi-equivalence}, we obtain 
		\begin{equation}
			\begin{split}
				&\quad D_N(\tau;m_1,\ldots,m_N,n_1,\ldots,n_N)\\	&=\frac{\sigma^{N-1}(\tau)\sigma\big(\tau+\sum_{i=1}^{N}(m_i+n_i)\big)\sigma(m_{N-1}-m_N)\sigma(n_{N-1}-n_N)P^{(1)}P^{(3)}_{N-1,N-1}P^{(3)}_{N,N}}{\sigma(m_{N-1}+n_{N-1})\sigma(m_N+n_N)\sigma(m_{N-1}+n_{N})\sigma(m_N+n_{N-1})P^{(2)}P^{(4)}_{N-1,N-1}P^{(4)}_{N,N}}\\
				&=\frac{\sigma^{N-1}(\tau)\sigma\big(\tau+\sum_{i=1}^{N}(m_i+n_i)\big)\prod_{1\leq i<j\leq N}\sigma(m_i-m_j)\sigma(n_i-n_j)}{\prod_{i,j=1}^{N}\sigma(m_i+n_j)}.
			\end{split}
		\end{equation}
		Henceforth we prove that the proposition holds when $M=N$ under the induction hypothesis. As a consequence, we complete the inductive proof.
	\end{proof}
	
	With these preliminaries, we will establish the asymptotic behavior of the $N$-elliptic localized solution \eqref{Eq-N-elliptic-localized-solution}-\eqref{Eq-N-elliptic-localized-solution-3} for the DNLS equation \eqref{Eq-DNLS-equation} via the $N$-fold Darboux-B\"acklund transformation $BT_0$ \eqref{Eq-BT0-1}-\eqref{Eq-BT0-2}. 
	Asymptotic analysis is performed in the $(\xi,t)$ plane for simplicity. The domain exterior to the interaction region is partitioned into $2N$ regions by $N$ pairs of lines $L_i^{\pm},i=1,2,\ldots,N$. Regions bounded by $L_{i-1}^{\pm}$ and $L_{i}^{\pm}$ are denoted as $R_i^{\pm}$ for $i=2,\ldots,N$. Additionally, regions bounded by $L_1^{\pm}$ and $L_N^{\mp}$ are labeled as $R_1^{\pm}$, with superscripts $\pm$ indicating positive and negative temporal directions. It is verified in Theorem \ref{thm:AB1} that lines $L_i^{\pm}$ characterize propagation directions of the elliptic localized waves \eqref{Eq-N-elliptic-localized-solution}-\eqref{Eq-N-elliptic-localized-solution-3}, defined by
	\begin{align}\label{Eq-line}
		L_i^{\pm}:=\{\xi=v_it-c_i^{\pm}\},\quad i=1,2,\ldots,N,
	\end{align}
	where  $c_i\in\mathbb{R},i=1,2,\ldots,N$ are constants. The velocities $v_i$ of each elliptic localized wave are determined by
	\begin{equation}\label{Eq-vi}
		\begin{split}
			v_i&=-\frac{16\mathrm{Re}(y(z_i))}{\mathrm{Re}(\beta_i)},\quad 
			\beta_i=\big(\kappa+\rho+2z_i\big)\frac{\zeta(\omega_1)}{\omega_1}-\zeta(\kappa+z_i)-\zeta(\rho+z_i),
		\end{split}
	\end{equation}
	which will be verified in the proof of the following theorem.
	We conduct asymptotic analysis along $L_k^{\pm}$ and in the regions $R_k^{\pm}$ as $t\rightarrow \pm\infty$ in the following theorems.

	\begin{theorem}[The asymptotic solution along  $L_{k}^{\pm}$]\label{thm:AB1}
	Assume $\Re(\beta_i)>0,i=1,2,\ldots,N$ and $v_i<v_j,1\leq i<j\leq N$.
	As $t\rightarrow -\infty$, the asymptotic behavior of the solution $u_N$ along line $L_k^{-}$ is given by 
	\begin{equation}\label{Eq-asym-Lkm-1}
		\begin{split}
			u_{N}&= u_{N,L_k^-}+\mathcal{O}\big(\exp\big(\mathrm{min}_{i\neq k}|\beta_i(v_i-v_k)|t\big)\big),\\
			u_{N,L_k^-}&=\left(\frac{\sqrt{\nu_0}\sigma(\kappa)\sigma(\rho+\kappa)\sigma(2\kappa)}{\sigma(\rho)\sigma(\rho+3\kappa)}\frac{D_{N,L_{k}^{-}}^{(1)}}{D_{N,L_{k}^{-}}^{(2)}}C_{N,L_{k}^{-}}e^{-F(\xi,t)}\right)_\xi,
		\end{split}
	\end{equation}
	where 
	\begin{equation}\label{Eq-asym-Lkm-2}
		\begin{split}
			D_{N,L_{k}^{-}}^{(1)}&=\Sigma^{(1)}(\xi+z_{N,L_k^-};z_k,z_k^*)r^{-1}(z_k)s(z_k^*)d_0(z_k)d_0^*(\hat{z}_k)E(\xi,t;z_k)E^*(\xi,t;z_k)I_k\\
			& \quad +\alpha_{N,L_{k}^{-}}\Sigma^{(1)}(\xi+z_{N,L_k^-};\hat{z}_k,z_k^*)r(z_k)s(z_k^*)d_0(\hat{z}_k)d_0^*(\hat{z}_k)E(\xi,t;\hat{z}_k)E^*(\xi,t;z_k)I_k\\
			&\quad-\alpha_{N,L_{k}^{-}}^*\Sigma^{(1)}(\xi+z_{N,L_k^-};z_k,\check{z}_k)r^{-1}(z_k)s^{-1}(z_k^*)d_0(z_k)d_0^*(z_k)E(\xi,t;z_k)E^*(\xi,t;\hat{z}_k)I_0(\xi+z_{N,L_k^-})\\
			&\quad -|\alpha_{N,L_{k}^{-}}|^2\Sigma^{(1)}(\xi+z_{N,L_k^-};\hat{z}_k,\check{z}_k)r(z_k)s^{-1}(z_k^*)d_0(\hat{z}_k)d_0^*(z_k)E(\xi,t;\hat{z}_k)E^*(\xi,t;\hat{z}_k)I_0(\xi+z_{N,L_k^-}),\\
			D_{N,L_{k}^{-}}^{(2)}&=	\Sigma^{(2)}(\xi+z_{N,L_k^-};z_k,z_k^*)d_0(z_k)d_0^*(\hat{z}_k)E(\xi,t;z_k)E^*(\xi,t;z_k)I_k\\
			&\quad+\alpha_{N,L_{k}^{-}}\Sigma^{(2)}(\xi+z_{N,L_k^-};\hat{z}_k,z_k^*)d_0(\hat{z}_k)d_0^*(\hat{z}_k)E(\xi,t;\hat{z}_k)E^*(\xi,t;z_k)I_k\\
			&\quad-\alpha_{N,L_{k}^{-}}^*\Sigma^{(2)}(\xi+z_{N,L_k^-};z_k,\check{z}_k)d_0(z_k)d_0^*(z_k)E(\xi,t;z_k)E^*(\xi,t;\hat{z}_k)I_0(\xi+z_{N,L_k^-})\\
			&\quad -|\alpha_{N,L_{k}^{-}}|^2\Sigma^{(2)}(\xi+z_{N,L_k^-};\hat{z}_k,\check{z}_k)d_0(\hat{z}_k)d_0^*(z_k)E(\xi,t;\hat{z}_k)E^*(\xi,t;\hat{z}_k)I_0(\xi+z_{N,L_k^-}),\\
		\end{split}
	\end{equation}
	and
	\begin{equation}\label{Eq-zLkm}
		\begin{split} 
			z_{N,L_{k}^{-}}&=-\sum_{i=1}^{k-1}(z_i+z_i^*)+\sum_{i=k+1}^N(z_i+z_i^*),\quad \alpha_{N,L_{k}^{-}}=\alpha_k\Delta_{N,L_{k}^{-}},\\
			\Delta_{N,L_{k}^{-}}&=\prod_{i=1}^{k-1}\frac{\sigma(z_i-\hat{z}_k)\sigma(z_k+z_i^*)}{\sigma(\hat{z}_k+z_i^*)\sigma(z_i-z_k)}\prod_{i=k+1}^{N}\frac{\sigma(\hat{z}_k-\hat{z}_i)\sigma(z_k+\check{z}_i)}{\sigma({\hat{z}_k+\check{z}_i)\sigma(z_k-\hat{z}_i)}} ,\quad 
			C_{N,L_k^-}=\prod_{i=1}^{k-1}\frac{s(z_i^*)}{r(z_i)}\prod_{i=k+1}^{N}\frac{r(z_i)}{s(z_i^*)}.
		\end{split}
	\end{equation}
	For $t\rightarrow +\infty$, the corresponding asymptotic representation along $L_k^{+}$ holds with $z_{N,L_k^-}$, $\Delta_{N,L_k^-}$, $\alpha_{N,L_k^-}$, $C_{N,L_k^-}$, and $\mathcal{O}\big(\exp\big(\mathrm{min}_{i\neq k}|\beta_i(v_i-v_k)|t\big)\big)$ replaced by $z_{N,L_k^+}$, $\Delta_{N,L_k^+}$, $\alpha_{N,L_k^+}$, $C_{N,L_k^+}$, and $\mathcal{O}\big(\exp\big(-\mathrm{min}_{i\neq k}|\beta_i(v_i-v_k)|t\big)\big)$ respectively, where
	\begin{align}\label{Eq-z-plus}
		z_{N,L_k^+}=-z_{N,L_k^-}, \quad 	\Delta_{N,L_k^+}=(\Delta_{N,L_k^-})^{-1},
		\quad 	\alpha_{N,L_k^+}=\alpha_k\Delta_{N,L_k^+},
		\quad C_{N,L_k^+}=(C_{N,L_k^-})^{-1}.
	\end{align}
	\end{theorem}
	\begin{proof}
		We can rewrite \eqref{Eq-N-elliptic-localized-solution}-\eqref{Eq-N-elliptic-localized-solution-3} in the following equivalent form
		\begin{align}\label{Eq-uN-form2}
			u_{N}=\left(\frac{\tilde{u}\det\left(\tilde{u}\textbf{M}+\mathrm{i}\phi^{(1)}\phi^{(2)\dagger}\right)}{\det(\tilde{u}\textbf{M})}\right)_\xi.
		\end{align}
		The entries of $\tilde{u}\textbf{M}+\mathrm{i}\phi^{(1)}\phi^{(2)\dagger}$ and $\tilde{u}\textbf{M}$ can be rewritten as 
		\begin{equation}\label{Eq-entries-numerator}
			\begin{split}
				\left(\tilde{u}\textbf{M}+\mathrm{i}\phi^{(1)}\phi^{(2)\dagger}\right)_{ij}
				&=T_1(\xi,t)\begin{pmatrix}
					\tilde{E}(\xi,t;z_i)r^{-1}(z_i) & \alpha_i \tilde{E}(\xi,t;\hat{z}_i)r(z_i)
				\end{pmatrix}\\
				&\quad \begin{pmatrix}
					\Delta^{(1)}(\xi;z_i,z_j^*)d_0(z_i)d_0^*(\hat{z}_j)I_j &  -	\Delta^{(1)}(\xi;z_i,\check{z}_j)d_0(z_i)d_0^*(z_j)\\
					\Delta^{(1)}(\xi;\hat{z}_i,z_j^*)d_0(\hat{z}_i)d_0^*(\hat{z}_j)I_j &  -	\Delta^{(1)}(\xi;\hat{z}_i,\check{z}_j)d_0(\hat{z}_i)d_0^*(z_j)
				\end{pmatrix}	\begin{pmatrix}
					\tilde{E}^*(\xi,t;z_j)s(z_j^*)\\ \alpha_j^* \tilde{E}^*(\xi,t;\hat{z}_j)s^{-1}(z_j^*)
				\end{pmatrix},\\
				\left(\tilde{u}\textbf{M}\right)_{ij}
				&=
				T_2(\xi,t)\begin{pmatrix}
					\tilde{E}(\xi,t;z_i) & \alpha_i \tilde{E}(\xi,t;\hat{z}_i)
				\end{pmatrix}\\
				&\quad \begin{pmatrix}
					\Delta^{(2)}(\xi;z_i,z_j^*)d_0(z_i)d_0^*(\hat{z}_j)I_j &  -\Delta^{(2)}(\xi;z_i,\check{z}_j)d_0(z_i)d_0^*(z_j)\\
					\Delta^{(2)}(\xi;\hat{z}_i,z_j^*)d_0(\hat{z}_i)d_0^*(\hat{z}_j)I_j &  -\Delta^{(2)}(\xi;\hat{z}_i,\check{z}_j)d_0(\hat{z}_i)d_0^*(z_j)
				\end{pmatrix}	\begin{pmatrix}
					\tilde{E}^*(\xi,t;z_j)\\ \alpha_j^* \tilde{E}^*(\xi,t;\hat{z}_j)
				\end{pmatrix},
			\end{split}
		\end{equation}
		respectively, where 
		\begin{equation}\label{Eq-tildeE}
			\begin{split}
				\tilde{E}(\xi,t;z)&=E(\xi,t;z)e^{-\frac{\zeta(\omega_1)}{\omega_1}z\xi},\\
				T_1(\xi,t)&=\frac{e^{\frac{\zeta(\omega_1)}{2\omega_1}\xi^2+\frac{\zeta(\omega_1)}{\omega_1}(\rho+2\kappa)\xi-F(\xi,t)}}{\sigma(\rho+3\kappa)\sigma(\xi-\kappa)},\quad 
				T_2(\xi,t)=-\frac{\sigma(\xi+\rho+2\kappa)e^{\frac{\zeta(\omega_1)}{2\omega_1}\xi^2-\frac{\zeta(\omega_1)}{\omega_1}\kappa\xi-F(\xi,t)}}{\sigma(\xi-\kappa)^2\sigma(\rho+3\kappa)},\\
				\Delta^{(1)}(\xi;\blacktriangle,\bullet)&=\Sigma^{(1)}(\xi;\blacktriangle,\bullet)e^{-\frac{\zeta(\omega_1)}{2\omega_1}\left(\xi^2-2(\blacktriangle+\bullet-\rho-2\kappa)\xi\right)},\quad	\Delta^{(2)}(\xi;\blacktriangle,\bullet)=\Sigma^{(2)}(\xi;\blacktriangle,\bullet)e^{-\frac{\zeta(\omega_1)}{2\omega_1}\left(\xi^2-2(\blacktriangle+\bullet+\kappa)\xi\right)}.
			\end{split}
		\end{equation}
		The functions $\Delta^{(1,2)}(\xi;\blacktriangle,\bullet)$ are bounded and periodic in $\xi$ for arbitrary $\blacktriangle,\bullet\in\mathbb{C}$. Based on \eqref{Eq-entries-numerator}-\eqref{Eq-tildeE}, we obtain that
		\begin{equation}\label{Eq-asym-solution}
			\begin{split}
				\frac{\det\left(	\tilde{u}\textbf{M}+\mathrm{i}\phi^{(1)}\phi^{(2)\dagger}\right)}{\det\left(\tilde{u}\textbf{M}\right)}=\frac{\det\left(\sum_{m,n=1}^2(-1)^{n+1}T_1(\xi,t)\textbf{Y}_m\textbf{X}_m\mathbf{\Theta}_{m,n}^{(1)}\textbf{X}_n^\dagger\textbf{Z}_n\right)}{\det\left(\sum_{m,n=1}^2(-1)^{n+1}T_2(\xi,t)\textbf{X}_m\mathbf{\Theta}_{m,n}^{(2)}\textbf{X}_n^\dagger\right)},	\\
			\end{split}
		\end{equation}
		where 
		\begin{equation}\label{Eq-Aij-entries}
			\begin{split}
				\quad \big(\mathbf{\Theta}_{m,n}^{(l)}\big)_{1\leq i,j\leq N}&=\left(
				\Delta^{(l)}\left(\xi;\frac{z_i+\hat{z}_i-(-1)^{m}(z_i-\hat{z}_i)}{2},\frac{z_j^*+\check{z}_j-(-1)^{n}(z_j^*-\check{z}_j)}{2}\right)\right.\\	&\quad 	 \left. d_0\left(\frac{z_i+\hat{z}_i-(-1)^{m}(z_i-\hat{z}_i)}{2}\right)
				d_0^*\left(\frac{z_j+\hat{z}_j+(-1)^{n}(z_j-\hat{z}_j)}{2}\right)I_j^{2-n}\right),\quad   l,m,n=1,2,\\
				\textbf{X}_1&=\mathrm{diag}\begin{pmatrix}
					1,...,1,e^{-\tau_{k+1}},...,e^{-\tau_{N}}
				\end{pmatrix}
				,\quad \quad \quad \quad \quad \quad  \textbf{X}_2=\mathrm{diag}\begin{pmatrix}
					e^{\tau_{1}},...,e^{\tau_{k}},1,...,1
				\end{pmatrix},\\
				\textbf{Y}_1&=\mathrm{diag}\begin{pmatrix}
					r^{-1}(z_i)
				\end{pmatrix},\quad \quad \quad \quad  \textbf{Y}_2=\textbf{Y}_1^{-1},\quad\quad \quad \quad  \textbf{Z}_1=\mathrm{diag}\begin{pmatrix}
					s(z_j^*)
				\end{pmatrix},\quad\quad \quad \quad  \textbf{Z}_2=\textbf{Z}_1^{-1}.
			\end{split}
		\end{equation}     
		Here the functions 
		\begin{equation}\label{Eq-tau}
			\tau_i:=\ln(\alpha_i\tilde{E}(\xi,t;\hat{z}_i)\tilde{E}^{-1}(\xi,t;z_i))=\Re(\beta_i)(\xi-v_it)+\mathrm{i}\big(\Im(\beta_i)\xi+16\Im(y(z_i))t\big)+\ln(\alpha_i),\quad 1\leq i\leq N,
		\end{equation}
		determine the propagation directions of each elliptic localized waves, where we make utilize of \eqref{Eq-solution to the Lax pair-2} and \eqref{Eq-tildeE}. Henceforth, we have
		\begin{equation}\label{Eq-tau-2}
			\begin{split}
			\textbf{X}_1&\rightarrow\mathrm{diag}\Big(
			\overbrace{1,...,1}^{k},0,...,0
			\Big)
			,\quad  \textbf{X}_2\rightarrow\mathrm{diag}\Big(
			0,...,e^{\tau_{k}},	\overbrace{1,...,1}^{N-k}
			\Big),
		\end{split}
	\end{equation}
		along $L_k^{-}$ as $t\rightarrow -\infty$, since $\Re(\beta_i)>0,i=1,2,\ldots,N$ and $v_i<v_j,1\leq i<j\leq N$. Combining with \eqref{Eq-tildeu}, \eqref{Eq-uN-form2}-\eqref{Eq-Aij-entries} yields the asymptotic representation
		\begin{equation}\label{Eq-N-breather-asymptotic}
			\begin{split}
				u_{N}&= u_{N,L_k^-}+\mathcal{O}\big(\exp\big(\mathrm{min}_{i\neq k}|\beta_i(v_i-v_k)|t\big)\big),\quad  u_{N,L_{k}^{-}}=\Bigg(\tilde{u}\left(-\frac{\sigma(\xi-\kappa)}{\sigma(\xi+\rho+2\kappa)}\right)^NC_{N,L_k^{-}}\frac{U_{N,L_k^{-}}^{(1)}}{U_{N,L_k^{-}}^{(2)}}\Bigg)_\xi,
			\end{split}
		\end{equation}
		where
		\begin{equation}
			\begin{split}
				U_{N,L_k^{-}}^{(1)}&=I_k\left(r^{-1}(z_k)s(z_k^*)\det(\textbf{V}_{N,L_k^{-}}^{[(1),1,1]})+r(z_k)s(z_k^*)\det(\textbf{V}_{N,L_k^{-}}^{[(1),2,1]})e^{\tau_k}\right)\\
				&\quad -\left(r^{-1}(z_k)s^{-1}(z_k^*)\det(\textbf{V}_{N,L_k^{-}}^{[(1),1,2]})e^{\tau_k^*}+r(z_k)s^{-1}(z_k^*)\det(\textbf{V}_{N,L_k^{-}}^{[(1),2,2]})e^{\tau_k+\tau_k^*}\right),\\
				U_{N,L_k^{-}}^{(2)}&=I_k\left(\det(\textbf{V}_{N,L_k^{-}}^{[(2),1,1]})+\det(\textbf{V}_{N,L_k^{-}}^{[(2),2,1]})e^{\tau_k}\right)-\left(\det(\textbf{V}_{N,L_k^{-}}^{[(2),1,2]})e^{\tau_k^*}+\det(\textbf{V}_{N,L_k^{-}}^{[(2),2,2]})e^{\tau_k+\tau_k^*}\right),\\
				\left(\textbf{V}_{N,L_k^-}^{[(1),m,n]}\right)_{ij}&=	\frac{\sigma(2\kappa+\rho-\eta_j^{[m,n]})\sigma(\kappa-\delta_i^{[m,n]})\sigma(\delta_i^{[m,n]}+\eta_j^{[m,n]}-2\kappa-\rho-\xi)}{\sigma(\delta_i^{[m,n]}+\eta_j^{[m,n]})e^{-(\delta_i^{[m,n]}+\eta_j^{[m,n]})\frac{\zeta(\omega_1)\xi}{\omega_1}}} d_0(\delta_i^{[m,n]})d_0^*(\varepsilon_j^{[m,n]}),\\
				\left(\textbf{V}_{N,L_k^{-}}^{[(2),m,n]}\right)_{ij}&=	-\frac{\sigma(\delta_i^{[m,n]}+\eta_j^{[m,n]}+\kappa-\xi)}{\sigma(\delta_i^{[m,n]}+\eta_j^{[m,n]}-2\kappa-\rho-\xi)}	\left(\textbf{V}_{N,L_k^-}^{[(1),m,n]}\right)_{ij},\quad m,n=1,2,
			\end{split}
		\end{equation}
		and 
		\begin{equation}\label{Eq-N-breather-asymptotic-3}
			\begin{split}			&\left(\delta_1^{[1,n]}, \delta_2^{[1,n]}, \ldots, \delta_k^{[1,n]}, \delta_{k+1}^{[1,n]}, \ldots, \delta_N^{[1,n]}\right)=\left(z_1,z_2,\ldots,z_k,\hat{z}_{k+1},\ldots,\hat{z}_{N}\right),\\ 
				& \left(\delta_1^{[2,n]}, \delta_2^{[2,n]}, \ldots, \delta_{k-1}^{[2,n]}, \delta_k^{[2,n]}, \ldots, \delta_N^{[2,n]}\right)=\left(z_1,z_2,\ldots,z_{k-1},\hat{z}_k,\ldots,\hat{z}_{N}\right), \\	
				& \left(\eta_1^{[m,1]}, \eta_2^{[m,1]}, \ldots, \eta_k^{[m,1]}, \eta_{k+1}^{[m,1]}, \ldots, \eta_N^{[m,1]}\right)=\left(z_1^*,z_2^*,\ldots,z_k^*,\check{z}_{k+1},\ldots,\check{z}_{N}\right), \\
				&  \left(\eta_1^{[m,2]}, \eta_2^{[m,2]}, \ldots, \eta_{k-1}^{[m,2]}, \eta_k^{[m,2]}, \ldots, \eta_N^{[m,2]}\right)= \left(z_1^*,z_2^*,\ldots,z_{k-1}^*,\check{z}_{k},\ldots,\check{z}_{N}\right),\\
				& \left(\varepsilon_1^{[m,1]}, \varepsilon_2^{[m,1]}, \ldots, \varepsilon_k^{[m,1]}, \varepsilon_{k+1}^{[m,1]}, \ldots, \varepsilon_N^{[m,1]}\right)=\left(\hat{z}_1,\hat{z}_2,\ldots,\hat{z}_k,z_{k+1},\ldots,z_N\right), \\
				&  \left(\varepsilon_1^{[m,2]}, \varepsilon_2^{[m,2]}, \ldots, \varepsilon_{k-1}^{[m,2]}, \varepsilon_k^{[m,2]}, \ldots, \varepsilon_N^{[m,2]}\right)= \left(\hat{z}_1,\hat{z}_2,\ldots,\hat{z}_{k-1},z_k,\ldots,z_N\right).
			\end{split}
		\end{equation}
		Utilizing Theorem \ref{thm:Cauchy matrix}, we can evaluate the determinants appeared in the asymptotic formula \eqref{Eq-N-breather-asymptotic}-\eqref{Eq-N-breather-asymptotic-3}. Specifically, we have
		\begin{equation}\label{Eq-detV}
			\begin{split}
				\det(\textbf{V}_{N,L_k^-}^{[(1),m,n]})=\Pi_{N,L_k^-}^{[m,n]}(\xi;2\kappa+\rho),\quad \det(\textbf{V}_{N,L_k^-}^{[(2),m,n]})=(-1)^N\Pi_{N,L_k^-}^{[m,n]}(\xi;-\kappa),
			\end{split}
		\end{equation}
		where
		\begin{equation}
			\begin{split}
				&\quad \Pi_{N,L_k^-}^{[m,n]}(\xi;\bullet)\\
				&=\sigma(-\xi-\bullet)^{N-1}\sigma\big(-\xi-\bullet+\sum_{i=1}^N(\delta_i^{[m,n]}+\eta_i^{[m,n]})\big)\frac{ \prod_{1\leq i<j\leq N} \sigma(\delta_i^{[m,n]} - \delta_j^{[m,n]}) \sigma(\eta_i^{[m,n]} - \eta_j^{[m,n]})}{\prod_{i,j=1}^{N}\sigma(\delta_i^{[m,n]}+\eta_j^{[m,n]})}\\
				&\quad 	 \prod_{i=1}^N \sigma(\kappa - \delta_i^{[m,n]}) \sigma(2\kappa + \rho - \eta_i^{[m,n]}) d_0(\delta_i^{[m,n]}) d_0^*(\varepsilon_i^{[m,n]})e^{(\delta_i^{[m,n]}+\eta_i^{[m,n]})\frac{\zeta(\omega_1)}{\omega_1}\xi}.
			\end{split}
		\end{equation}
		Furthermore, by removing the common divisors of the numerator and dinominator, we obtain
		\begin{equation}\label{Eq-N-breather-asymptotic-2}
			\begin{split}
				u_{N,L_{k}^{-}}=\left(\frac{\sqrt{\nu_0}\sigma(\kappa)\sigma(\rho+\kappa)\sigma(2\kappa)}{\sigma(\rho)\sigma(\rho+3\kappa)}\frac{\tilde{U}_{N,L_k^{-}}^{(1)}}{\tilde{U}_{N,L_k^{-}}^{(2)}}C_{N,L_k^{-}}e^{-F(\xi,t)}\right)_\xi,
			\end{split}
		\end{equation}
		where 
		\begin{equation}\label{Eq-U12-tilde}
			\begin{split}
				\tilde{U}_{N,L_k^-}^{(1)}&=I_k\big(r^{-1}(z_k)s(z_k^*)\tilde{V}_{N,L_k^-}^{[(1),1,1]}+r(z_k)s(z_k^*)\tilde{V}_{N,L_k^-}^{[(1),2,1]}e^{\tau_k}\big)\\ &\quad -\big(r^{-1}(z_k)s^{-1}(z_k^*)\tilde{V}_{N,L_k^-}^{[(1),1,2]}e^{\tau_k^*}+r(z_k)s^{-1}(z_k^*)\tilde{V}_{N,L_k^-}^{[(1),2,2]}e^{\tau_k+\tau_k^*}\big), \\        
				\tilde{U}_{N,L_k^{-}}^{(2)}&=I_k\big(\tilde{V}_{N,L_k^-}^{[(2),1,1]}+\tilde{V}_{N,L_k^-}^{[(2),2,1]}e^{\tau_k}\big)-\big(\tilde{V}_{N,L_k^-}^{[(2),1,2]}e^{\tau_k^*} +\tilde{V}_{N,L_k^-}^{[(2),2,2]}e^{\tau_k+\tau_k^*}\big), \\
				\tilde{V}_{N,L_k^-}^{[(1),m,n]}&=\tilde{\Pi}_{N,L_k^-}^{[(1),m,n]}\tilde{\Pi}_{N,L_k^-}^{[(2),m,n]}\sigma\big(-\xi-2\kappa-\rho+\sum_{i=1}^N(\delta_i^{[m,n]}+\eta_i^{[m,n]})\big)e^{(\delta_k^{[m,n]}+\eta_k^{[m,n]})\frac{\zeta(\omega_1)}{\omega_1}\xi}, \\
				\tilde{V}_{N,L_k^-}^{[(2),m,n]}&=\tilde{\Pi}_{N,L_k^-}^{[(1),m,n]}\tilde{\Pi}_{N,L_k^-}^{[(2),m,n]}\sigma\big(\xi-\kappa-\sum_{i=1}^N(\delta_i^{[m,n]}+\eta_i^{[m,n]})\big)e^{(\delta_k^{[m,n]}+\eta_k^{[m,n]})\frac{\zeta(\omega_1)}{\omega_1}\xi},\\
				\tilde{\Pi}_{N,L_k^-}^{[(1),m,n]} &= \sigma(\kappa - \delta_k^{[m,n]}) \sigma(2\kappa + \rho - \eta_k^{[m,n]}) d_0(\delta_k^{[m,n]}) d_0^*(\varepsilon_k^{[m,n]}), \\
				\tilde{\Pi}_{N,L_k^-}^{[(2),m,n]} &=\frac{ \prod_{i=1}^{k-1}\sigma(\delta_i^{[m,n]}-\delta_k^{[m,n]})\sigma(\eta_i^{[m,n]}-\eta_k^{[m,n]})\prod_{i=k+1}^{n}\sigma(\delta_k^{[m,n]}-\delta_i^{[m,n]})\sigma(\eta_k^{[m,n]}-\eta_i^{[m,n]})}{\prod_{i=1}^{N}\sigma(\delta_i^{[m,n]}+\eta_k^{[m,n]})\prod_{j=1,j\neq k}^{N}\sigma(\delta_k^{[m,n]}+\eta_j^{[m,n]})}.	
			\end{split}
		\end{equation}
		Using \eqref{Eq-N-breather-asymptotic-3}, for each $m,n=1,2$, 
		\begin{equation}
			z_{N,L_k^-}=-\sum_{i=1}^{k-1}(\delta_i^{[m,n]}+\eta_i^{[m,n]})-	\sum_{i=k+1}^{N}(\delta_i^{[m,n]}+\eta_i^{[m,n]}).
		\end{equation}
		Moreover, the representation of $\alpha_{N,L_k^-}$ in \eqref{Eq-asym-Lkm-2} can be evaluated as 
		\begin{equation}
			\alpha_{N,L_k^-}=\frac{	\tilde{\Pi}_{N,L_k^-}^{[(2),2,1]}\sigma(\delta_k^{[2,1]}+\eta_k^{[2,1]})}{	\tilde{\Pi}_{N,L_k^-}^{[(2),1,1]}\sigma(\delta_k^{[1,1]}+\eta_k^{[1,1]})}.
		\end{equation}
		Making use of \eqref{Eq-N-breather-asymptotic-3} and the last identity in \eqref{Eq-U12-tilde}, we obtain the expression for $\alpha_{N,L_k^-}$ in \eqref{Eq-asym-Lkm-2}. Henceforth, via a further direct simplification of \eqref{Eq-N-breather-asymptotic-2} and \eqref{Eq-U12-tilde}, the proof is complete.
	\end{proof}

	Furthermore, the asymptotic formulas of the solution \eqref{Eq-N-elliptic-localized-solution}-\eqref{Eq-N-elliptic-localized-solution-3} in the regions $R_{k}^{\pm}$ as $t\rightarrow\pm\infty$ are concluded in the following theorem.

	\begin{theorem}	[The asymptotic solution in $R_{k}^{\pm}$]\label{thm:AB2}
When $\Re(\beta_i)>0,i=1,2,\ldots,N$ and $v_i<v_j,1\leq i<j\leq N$, the asymptotic behavior of the solution $u_N$ as $t\rightarrow -\infty$ in region $R_k^{-}$ is given by 
\begin{equation}\label{Eq-asymptotic-Rk}
	\begin{split}
		u_N= u_{N,R_k^-}+\mathcal{O}\big(\exp\big(\mathrm{min}_{i\neq k}|\beta_i(v_i-v_k)|t\big)\big),
	\end{split}
\end{equation}
where
\begin{equation}\label{Eq-uNRk-}
	u_{N,R_k^-}=	\left(-\frac{\sqrt{\nu_0}\sigma(\kappa)\sigma(\rho+\kappa)\sigma(2\kappa)\sigma\big(\xi+2\kappa+\rho+z_{N,R_k^-}\big)}{\sigma(\rho)\sigma(\rho+3\kappa)\sigma\big(\xi-\kappa+z_{N,R_k^-}\big)}C_{N,R_k^-}e^{-F(\xi,t)}\right)_\xi,
\end{equation}
with $C_{N,R_k^-}=\prod_{i=1}^{k-1}\frac{s(z_i^*)}{r(z_i)}\prod_{i=k}^{N}\frac{r(z_i)}{s(z_i^*)}$ and $z_{N,R_k^-}=	-\sum_{i=1}^{k-1}(z_i+z_i^*)+\sum_{i=k}^N(z_i+z_i^*)$.	
For $t\rightarrow +\infty$ in region $R_k^{+}$, the corresponding asymptotic representation holds with $C_{N,R_k^-}$, $z_{N,R_k^-}$, and $\mathcal{O}\big(\exp\big(\mathrm{min}_{i\neq k}|\beta_i(v_i-v_k)|t\big)\big)$ replaced by $ C_{N,R_k^+}$, $z_{N,R_k^+}$, and $\mathcal{O}\big(-\exp\big(\mathrm{min}_{i\neq k}|\beta_i(v_i-v_k)|t\big)\big)$ respectively, where
\begin{align}
	z_{N,R_k^+}=-z_{N,R_k^-},\quad C_{N,R_k^+}=(C_{N,R_k^-})^{-1}.
\end{align}
Furthermore, both asymptotic solutions admit the derivative-free representation  
\begin{equation}\label{Eq-asymptotic-Rkp}
	\begin{split}
		u_{N,R_k^\pm}	=\frac{\sqrt{\nu_0}\sigma(\kappa)\sigma\big(\xi+\rho+z_{N,R_k^{\pm}}\big)\sigma\big(\xi+\kappa+z_{N,R_k^{\pm}}\big)}{\sigma(\rho)\sigma^2\big(\xi-\kappa+z_{N,R_k^{\pm}}\big)}C_{N,R_k^{\pm}}e^{-F(\xi,t)},
	\end{split}
\end{equation}
respectively.	
	\end{theorem}
	\begin{proof}
		We also rewrite the $N$-elliptic localized solution as \eqref{Eq-uN-form2}. 
		In the region $R_k^{-}$, the velocity $v$ satisfies $v_{k-1}<v<v_{k}$. The ratios of the determinants $\det(\tilde{u}\textbf{M}+\mathrm{i}\phi^{(1)}\phi^{(2)\dagger})$ and $\det(\tilde{u}\textbf{M})$ can also be written in the form of \eqref{Eq-asym-solution}-\eqref{Eq-Aij-entries} with $\textbf{X}_{1,2}$ replaced by 
		\begin{align}
			\mathbf{X}_1 = \mathrm{diag}(1,\ldots,1,e^{-\tau_k},\ldots,e^{-\tau_N}), \quad 
			\mathbf{X}_2 = \mathrm{diag}(e^{\tau_1},\ldots,e^{\tau_{k-1}},1,\ldots,1).
		\end{align}
		When $\Re(\beta_i)>0,i=1,2,\ldots,N$ and $v_i<v_j,1\leq i<j\leq N$, we have
		\begin{align}
			\mathbf{X}_1\rightarrow\mathrm{diag}\Big(
			\overbrace{1,...,1}^{k-1},0,...,0
			\Big)
			,\quad  	\mathbf{X}_2\rightarrow\mathrm{diag}\Big(
			0,...,0,	\overbrace{1,...,1}^{N-k+1}
			\Big),
		\end{align}
		as $t\rightarrow -\infty$.
		Therefore, we obtain
		\begin{equation}\label{Eq-reduced-expression}
			\begin{split}
				u_N&= u_{N,R_k^-}+\mathcal{O}\big(\exp\big(\mathrm{min}_{i\neq k}|\beta_i(v_i-v_k)|t\big)\big),\\
				u_{N,R_{k}^{-}}&=\left(\frac{\sqrt{\nu_0}\sigma(\kappa)\sigma(\rho+\kappa)\sigma(2\kappa)}{\sigma(\rho)\sigma(\rho+3\kappa)}\frac{\det(\mathbf{V}_{N,L_k^-}^{[(1),2,2]})}{\det(\mathbf{V}_{N,L_k^-}^{[(2),2,2]})}C_{N,R_k^{-}}e^{-F(\xi,t)}\right)_\xi.
			\end{split}
		\end{equation}
		We can evaluate the determinants involved in \eqref{Eq-reduced-expression} using \eqref{Eq-detV}.  By removing common divisors in the denominator and numerator, we arrive at \eqref{Eq-asymptotic-Rk}-\eqref{Eq-uNRk-}. Similarly, we obtain the asymptotic behavior in the region $R_k^+$. Moreover, using \eqref{Eq-derivative-sigma} and the third formula in \eqref{Eq-formulas of Weierstrass functions}, we obtain \eqref{Eq-asymptotic-Rkp}. This completes the proof.
	\end{proof}
	
	Using \eqref{Eq-N-elliptic-localized-solution-3}, we deduce that 
	\begin{equation}
		\begin{split}
			r^*(z_i)=\begin{dcases}	s(z_i^*), & \Re(\kappa)=\Re(\rho)=0,\\
				-s(z_i^*)e^{-2\zeta(\omega_1)(-\rho+\omega_1+z_i^*)}, & \Re(\kappa)=0,\Re(\rho)=\omega_1,\\
				s(z_i^*)e^{-2\zeta(\omega_1)(-5\kappa-3\rho+8\omega_1+4z_i^*)}, & \Re(\kappa)=\Re(\rho)=\omega_1.
			\end{dcases}
		\end{split}
	\end{equation}
	Thus it follows from  \eqref{Eq-zLkm} that 
	\begin{equation}
		\begin{split}
			\left|C_{N,L_k^-}\right|
			=\begin{dcases}
				1, & \Re(\kappa)=\Re(\rho)=0,\\
				e^{-\zeta(\omega_1)z_{N,L_k^-}},& \Re(\kappa)=0,\Re(\rho)=\omega_1,\\
				e^{-4\zeta(\omega_1)z_{N,L_k^-}},& \Re(\kappa)=\Re(\rho)=\omega_1.
			\end{dcases}	\end{split}
	\end{equation}
	Henceforth, when $\Re(\kappa)=\Re(\rho)=0$, the asymptotic solution can be written as
	\begin{equation}\label{Eq-asym-rewritten}
		\begin{split}
			u_{N,L_k^{-}}=
			\left(\frac{\sqrt{\nu_0}\sigma(\kappa)\sigma(\rho+\kappa)\sigma(2\kappa)}{\sigma(\rho)\sigma(\rho+3\kappa)}\frac{D_{N,L_{k}^{-}}^{(1)}C_{N,L_k^-}}{D_{N,L_{k}^{-}}^{(2)}|C_{N,L_k^-}|}e^{-F(\xi+z_{N,L_k^-},t)+(\zeta(\rho+\kappa)+\zeta(2\kappa))z_{N,L_k^-}}\right)_\xi.
		\end{split}
	\end{equation}
	Since that the DNLS equation \eqref{Eq-DNLS-equation} is rotational invariant, we deduce that 
	\begin{equation}\label{Eq-asym-rewritten}
		\begin{split}
			\hat{u}_{N,L_k^{-}}=
			\left(\frac{\sqrt{\nu_0}\sigma(\kappa)\sigma(\rho+\kappa)\sigma(2\kappa)}{\sigma(\rho)\sigma(\rho+3\kappa)}\frac{D_{N,L_{k}^{-}}^{(1)}}{D_{N,L_{k}^{-}}^{(2)}}e^{-F(\xi+z_{N,L_k^-},t)}\right)_\xi,
		\end{split}
	\end{equation}
	is also a solution to the DNLS equation \eqref{Eq-DNLS-equation}.
	Furthermore, we observe that
	replacing $\alpha_{N,L_k^-}$, $E(\xi,t;z_k)$ and $E(\xi,t;\hat{z}_k)$ with $\alpha_{N,L_k^-}e^{(\zeta(\kappa+z_k)+\zeta(\rho+z_k))z_{N,L_k^-}}$ , $E(\xi+z_{N,L_k^-},t;z_k)$ and $E(\xi+z_{N,L_k^-},t;\hat{z}_k)$ respectively preserves \eqref{Eq-asym-rewritten}. Therefore, the asymptotic solution \eqref{Eq-asym-Lkm-1}-\eqref{Eq-asym-Lkm-2} constitutes a first-order elliptic localized solution derived from \eqref{Eq-DNLS-elliptic-solution} via the one-fold Darboux-B\"acklund transformation $BT_0$ at $\lambda_k=\lambda(z_k)$, with a shift $\xi\rightarrow \xi+z_{N,L_k^-}$ and a phase factor $C_{N,L_k^-}e^{(\zeta(\rho+\kappa)+\zeta(2\kappa))z_{N,L_k^-}}$ when $\Re(\kappa)=\Re(\rho)=0$. For other values of $\Re(\kappa)$ and $\Re(\rho)$, we could obtain the same conclusions in a similar approach. Additionally, \eqref{Eq-asymptotic-Rkp} demonstrates that the asymptotic solutions in regions $R_{k}^{\pm},k=1,2,\ldots,N$ correspond to shifted versions of the elliptic solution \eqref{Eq-DNLS-elliptic-solution}.

	Previous studies \cite{LS-mKdV-solution} established that $N$-elliptic localized solutions of the mKdV equation possess symmetry with respect to the origin when specific conditions are satisfied. Here, we extend this symmetry to the DNLS equation. The asymptotic analysis in Theorems \ref{thm:AB1} and \ref{thm:AB2} demonstrates that wave collisions are elastic. This result extends to strictly elastic collisions between elliptic localized waves when solutions exhibit symmetry about the origin.

	\begin{theorem}[The symmetry property of the $N$-elliptic localized solutions]\label{thm:symmetry}
When $\alpha_i=1,i=1,2,\ldots,N$, the $N$-elliptic localized solution $u_N$ satisfies the symmetry:
\begin{align}\label{Eq-uN-symmetry}
	u_N(\xi,t)= u_N^*(-\xi,-t).
\end{align}
	\end{theorem}
	\begin{proof}
		Utilizing \eqref{Eq-F} and  \eqref{Eq-solution to the Lax pair-2}, we could verify that
		\begin{align}
			E(-\xi,-t;z)=E(x,t;\hat{z})e^{F(\xi,t)},
		\end{align}
		is satisfied for arbitrary uniform parameter $z$. Then we obtain the symmetry
		\begin{align}\label{Eq-symmetry-entry}
			\begin{pmatrix}
				\phi_i^{(1)}(-\xi,-t)\\
				\phi_i^{(2)}(-\xi,-t)
			\end{pmatrix}:=\Phi(-\xi,-t;\lambda(z_i))\begin{pmatrix}
				1\\
				1
			\end{pmatrix}=\begin{pmatrix}
				0 & -\mathrm{i}\\
				\mathrm{i} & 0
			\end{pmatrix}\Phi(\xi,t;\lambda(z_i))\begin{pmatrix}
				0 & 1\\
				1 & 0
			\end{pmatrix}\begin{pmatrix}
				1\\
				1
			\end{pmatrix}=\begin{pmatrix}
				-\mathrm{i}\phi_{i}^{(2)}(\xi,t)\\
				\mathrm{i}\phi_{i}^{(1)}(\xi,t)
			\end{pmatrix},
		\end{align}
		which leads to
		\begin{equation}
			\begin{split}
				\phi_j^\dagger(\xi,t)\phi_i(\xi,t)&=\phi_j^\dagger(-\xi,-t)\phi_i(-\xi,-t),\\
				\phi_j^\dagger(\xi,t)\sigma_3\phi_i(\xi,t)&=-\phi_j^\dagger(-\xi,-t)\sigma_3\phi_i(-\xi,-t).
			\end{split}
		\end{equation}		
	Moreover, we have
		\begin{equation}\label{Eq-M-symmetry}
			\mathbf{M}(\xi,t)=-\mathbf{M}^\dagger(-\xi,-t).
		\end{equation}
		Similarly, it follows directly from \eqref{Eq-symmetry-entry} that 
		\begin{equation}
			\phi_i^{(1)}\big(\phi_j^{(2)}\big)^*(\xi,t)=	-\phi_i^{(2)}\big(\phi_j^{(1)}\big)^*(-\xi,-t).
		\end{equation}
		As a consequence, we have 
		\begin{equation}\label{Eq-phi-symmetry}
			\left(\phi^{(1)}\big(\phi^{(2)}\big)^\dagger\right)(\xi,t)=-\left(\phi^{(1)}\big(\phi^{(2)}\big)^\dagger\right)^\dagger(-\xi,-t).
		\end{equation}
		Using the explicit form of the elliptic solution \eqref{Eq-DNLS-elliptic-solution}, we arrive at 
		\begin{equation}
			u(\xi,t)=u^*(-\xi,-t).
		\end{equation}
		Henceforth, it holds that
		\begin{equation}\label{Eq-utilde-symmetry}
			\tilde{u}(\xi,t)=-\tilde{u}^*(-\xi,-t).
		\end{equation}
		Combining \eqref{Eq-uN-Sherman}, \eqref{Eq-M-symmetry}, \eqref{Eq-phi-symmetry} and \eqref{Eq-utilde-symmetry}, utilmately we verify  \eqref{Eq-uN-symmetry} and complete the proof.
	\end{proof}

	Combining Theorems \ref{thm:AB1}, \ref{thm:AB2}, and \ref{thm:symmetry}, we conclude that the collisions under the specific condition are strictly elastic.

	\section{The equivalent derivative-free form of the $N$-elliptic localized solutions}
	In this section, we present an equivalent derivative-free representation of the $N$-elliptic localized solutions \eqref{Eq-N-elliptic-localized-solution}-\eqref{Eq-N-elliptic-localized-solution-3}, originally obtained via $BT_\infty$
	transformation defined in \eqref{Eq-DT-infty}. Additionally, we establish asymptotic analysis for these solutions using this formulation.   The form of $N$-elliptic function solutions presented in Theorem \ref{thm:N-soliton solution} contains a derivative. As a supplyment, we present below an equivalent derivative-free representation based on the $BT_\infty$ transformation \eqref{Eq-DT-infty}. Notice that 
	\begin{equation}\label{Eq-Morrison-simplification}
		\begin{split}
			&\quad \left(1+2\phi^{(1)\dagger}\mathbf{\Lambda}^{\dagger}\mathbf{M}^{-1}\phi^{(1)}\right)u+4\phi^{(2) \dagger}\left(\mathbf{\Lambda}^{\dagger}\right)^2\mathbf{M}^{-1}\phi^{(1)}\\
			&=u+2\left(u\phi^{(1)\dagger}+2\phi^{(2)\dagger}\mathbf{\Lambda}^{\dagger}\right)\mathbf{\Lambda}^{\dagger}\textbf{M}^{-1}\phi^{(1)}\\
			&:=u+2\phi_A^\dagger\mathbf{\Lambda}^{\dagger}\textbf{M}^{-1}\phi^{(1)}.\\
		\end{split}
	\end{equation}
	With the help of the Sherman-Morrison-Woodbury-type matrix identity \eqref{Eq-SMW-identity}, combining with the $N$-fold Darboux-B\"acklund transformation $BT_{\infty}$ \eqref{Eq-DT-infty}, the $N$-elliptic localized solution $u_N$ \eqref{Eq-N-elliptic-localized-solution}-\eqref{Eq-N-elliptic-localized-solution-3} can be rewritten as
	\begin{align}\label{Eq-uN-BT-infty}
		u_N=\frac{u^{1-N}\det\left(-\mathbf{M}^\dagger\right)}{\det\left(\textbf{M}\right)}	\frac{\det\left(\frac{1}{2}u\textbf{M}\left(\mathbf{\Lambda}^\dagger\right)^{-1}+\phi^{(1)}\phi_A^\dagger\right)}{\det\left(-\frac{1}{2}u\mathbf{M}^\dagger\left(\mathbf{\Lambda}^\dagger\right)^{-1}+\phi^{(2)}\phi^{(2)\dagger}\right)},
	\end{align}
	where $u$ is the elliptic solution \eqref{Eq-DNLS-elliptic-solution}. Expressing the entries of the determinants involved in \eqref{Eq-uN-BT-infty} in terms of Weierstrass functions yields the following theorem.
	\begin{theorem}[Equivalent derivative-free form of $N$-elliptic localized solutions]\label{thm:solution-2}  
	The $N$-elliptic localized solution \eqref{Eq-N-elliptic-localized-solution}-\eqref{Eq-N-elliptic-localized-solution-3} to the DNLS equation \eqref{Eq-DNLS-equation} can be expressed as
	\begin{equation}\label{Eq-uN-2}
		\begin{split}
			u_N(\xi,t)=\left(-\frac{\sigma^2(\xi-\kappa)}{\sigma(\xi+\kappa)\sigma(\xi+\rho)}\right)^{N}\left(\frac{\sqrt{\nu_0}\sigma(\kappa)\sigma(\xi+\rho)\sigma(\xi+\kappa)}{\sigma(\rho)\sigma^2(\xi-\kappa)}e^{-F(\xi,t)}\right)\prod_{i=1}^{2}\frac{\det\left(\mathbf{B}_{N}^{(2i+1)}\right)}{\det\left(\mathbf{B}_{N}^{(2i+2)}\right)},
		\end{split}
	\end{equation}
	where the matrix elements are defined by
	\begin{equation}\label{Eq-asym-Lkm-3}
		\begin{split}
			\left(\mathbf{B}_{N}^{(s)}\right)_{ij}&=\sum_{m,n=0}^{1}(-1)^{n}\alpha_i^{m}\left(\alpha_j^*\right)^{n}I_0^{n}(\xi)I_j^{\frac{1+(-1)^s}{2}+(-1)^{s+1}n}\\
			&\quad \Sigma^{(s)}\left(\xi;\frac{z_i+\hat{z}_i+(-1)^{m}(z_i-\hat{z}_i)}{2},\frac{z_j^*+\check{z}_j+(-1)^{n}\big(z_j^*-\check{z}_j\big)}{2}\right)\\
			&\quad d_0\left(\frac{z_i+\hat{z}_i+(-1)^{m+(s\quad \mathrm{mod}5)}(z_i-\hat{z}_i)}{2}\right) d_0^*\left(\frac{z_j+\hat{z}_j-(-1)^{s+n}(z_j-\hat{z}_j)}{2}\right)\\
			&\quad 
			E\left(\frac{z_i+\hat{z}_i+(-1)^{m}(z_i-\hat{z}_i)}{2}\right)E^*\left(\frac{z_j+\hat{z}_j+(-1)^{n}(z_j-\hat{z}_j)}{2}\right),\quad s=3,4,5,6,
		\end{split}
	\end{equation}
	and the $\Sigma$-functions are given by
	\begin{equation}
		\begin{split}
			&\Sigma^{(3)}(\xi;\blacktriangle,\bullet)=\frac{\sigma( - \blacktriangle-\bullet + \kappa + \xi)\sigma(-2\kappa - \rho - \blacktriangle)\sigma(\kappa + \bullet)}{\sigma(\blacktriangle + \bullet)},\\
			&	\Sigma^{(4)}(\xi;\blacktriangle,\bullet)=\frac{\sigma(-\blacktriangle - \bullet - \kappa + \xi)\sigma(2\kappa + \rho - \bullet)\sigma(\kappa - \blacktriangle)}{\sigma(\blacktriangle + \bullet)},\\
			&	\Sigma^{(5)}(\xi;\blacktriangle,\bullet)=\frac{\sigma(\blacktriangle + \bullet -\rho- \xi  )\sigma(\kappa - \bullet)\sigma(\kappa + \bullet)\sigma(\blacktriangle + \rho)}{\sigma(\blacktriangle + \bullet)\sigma(\bullet - \rho)},\\
			&	\Sigma^{(6)}(\xi;\blacktriangle,\bullet)=\frac{\sigma(-\blacktriangle - \bullet - \kappa + \xi)\sigma(2\kappa + \rho - \bullet)\sigma(-\rho - \blacktriangle)}{\sigma(\blacktriangle + \bullet)}.
		\end{split}
	\end{equation}
	\end{theorem}

	\begin{proof}	
		We focus on solutions with $\Re(\kappa)=0$. Cases where $\Re(\kappa)=\omega_1$ can be analyzed similarly with direct but more complicated computations.
		It follows directly from \eqref{Eq-Mij} that
		\begin{equation}\label{Eq-Mtilde}
			\begin{split}
				&\quad \left(-\textbf{M}^{\dagger}\right)_{ij}=\frac{\sqrt{\nu_0}\sigma(\kappa)\sigma(\kappa-\rho)}{\mathrm{i}\sigma(\rho)\sigma(\xi+\kappa)\sigma^2(2\kappa)}\left( \frac{\sigma(-z_j^*-z_i+\kappa+\xi)\sigma(-2\kappa-\rho-z_i)\sigma(\kappa+z_j^*)}{\sigma(z_i+z_j^*)}d_0(\hat{z}_i)d_0^*(z_j)E(z_i)E^*(z_j)\right.\\
				&\quad +\alpha_i\frac{\sigma(-\kappa-z_j^*)\sigma(-\kappa+z_i)\sigma(z_j^*-z_i-2\kappa-\rho-\xi)}{\sigma(-\kappa-\rho+z_j^*-z_i)}d_0(z_i)d_0^*(z_j)E(\hat{z}_i)E^*(z_j)\\
				&\quad +\alpha_j^*\frac{\sigma(-2\kappa-\rho+z_j^*)\sigma(-2\kappa-\rho-z_i)\sigma(-z_j^*+z_i+\rho-\xi)}{\sigma(-\kappa-\rho+z_j^*-z_i)}d_0(\hat{z}_i)d_0^*(\hat{z}_j)E(z_i)E^*(\hat{z}_j)I_jI_0(\xi)\\
				&\quad \left.+\alpha_i\alpha_j^*\frac{\sigma(-2\kappa-\rho+z_j^*)\sigma(-\kappa+z_i)\sigma(-z_j^*-z_i-\kappa-\xi)}{\sigma(z_j^*+z_i)}d_0(z_i)d_0^*(\hat{z}_j)E(\hat{z}_i)E^*(\hat{z}_j)I_jI_0(\xi)\right).\\
			\end{split}
		\end{equation}
		Based on \eqref{Eq-uN-2}, it suffices to determine the entries of the matrices $\frac{1}{2}u\textbf{M}\left(\mathbf{\Lambda}^\dagger\right)^{-1}+\phi^{(1)}\phi_A^\dagger$ and $-\frac{1}{2}u\textbf{M}^\dagger\left(\mathbf{\Lambda}^\dagger\right)^{-1}+\phi^{(2)}\phi^{(2)\dagger}$. Using  \eqref{Eq-DNLS-elliptic-solution} and \eqref{Eq-parameterization-lambda-y}, we obtain
		\begin{equation}\label{Eq-uML}
			\begin{split}
				\left(\frac{1}{2}u\mathbf{M}\left(\mathbf{\Lambda}^{-1}\right)^{\dagger}\right)_{ij}=\frac{\sqrt{\nu_0}\sigma(\kappa)\sigma(\xi+\rho)\sigma(\xi+\kappa)\sigma(-\rho+z_j^*)d_0^*(\hat{z}_j)}{\sigma(\rho)\sigma(2\kappa)\sigma^3(\xi-\kappa)\sigma(\kappa+z_j^*)d_0^*(z_j)}e^{-F(\xi,t)}I_j\mathbf{M}_{0,ij}.
			\end{split}
		\end{equation}
		Besides, using \eqref{Eq-solution to the Lax pair-1}-\eqref{Eq-phi-l}, \eqref{Eq-Morrison-simplification} and the addition formula \eqref{Eq-addition formulas of the sigma functions}, we arrive at
		\begin{equation}
			\begin{split}
				&\quad \left(\phi_A^\dagger\right)_i=\frac{\sqrt{\nu_0}\sigma(\kappa)\sigma(\kappa+\rho)\sigma(\xi+\kappa)}{\sigma(2\kappa)\sigma(\rho)\sigma^2(\xi-\kappa)}e^{-F(\xi,t)}\\
				&\quad \left(\frac{\sigma(z_i^*-\kappa)\sigma(z_i^*+\kappa-\rho-\xi)}{\sigma(z_i^*-\rho)}d_0^*(z_i)E^*(z_i)-\alpha_i^*\frac{\sigma(2\kappa-z_i^*-\xi)\sigma(\rho-z_i^*)}{\sigma(\kappa-z_i^*)}d_0^*(\hat{z}_i)E^*(\hat{z}_i)I_0(\xi)I_i\right).
			\end{split}
		\end{equation}
		As a consequence of \eqref{Eq-phi-l}, the entries of the matrix $\phi^{(1)}\phi_A^\dagger$ read as
		\begin{equation*}
			\begin{split}
				\quad\left(\phi^{(1)}\phi_A^\dagger\right)_{ij}&=\frac{\sqrt{\nu_0}\sigma(\kappa)\sigma(\kappa+\rho)\sigma(\xi+\kappa)}{\sigma(2\kappa)\sigma(\rho)\sigma^3(\xi-\kappa)}e^{-F(\xi,t)}\\
				&\quad\left( \frac{\sigma(z_i-\xi)\sigma(z_j^*-\kappa)\sigma(z_j^*+\kappa-\rho-\xi)}{\sigma(z_j^*-\rho)}d_0(z_i)d_0^*(z_j)E(z_i)E^*(z_j)\right.\\
				&\quad+\alpha_i\frac{\sigma(-\kappa-\rho-z_i-\xi)\sigma(z_j^*-\kappa)\sigma(z_j^*+\kappa-\rho-\xi)}{\sigma(z_j^*-\rho)}d_0(\hat{z}_i)d_0^*(z_j)E(\hat{z}_i)E^*(z_j)\\
				&\quad-\alpha_j^*\frac{\sigma(z_i-\xi)\sigma(2\kappa-z_j^*-\xi)\sigma(\rho-z_j^*)}{\sigma(\kappa-z_j^*)}d_0(z_i)d_0^*(\hat{z}_j)E(z_i)E^*(\hat{z}_j)I_0(\xi)I_j\\
				&\left.\quad-\alpha_i\alpha_j^*\frac{\sigma(-\kappa-\rho-z_i-\xi)\sigma(2\kappa-z_j^*-\xi)\sigma(\rho-z_j^*)}{\sigma(\kappa-z_j^*)}d_0(\hat{z}_i)d_0^*(\hat{z}_j)E(\hat{z}_i)E^*(\hat{z}_j)I_0(\xi)I_j\right).
			\end{split}
		\end{equation*}
		Making utilize of the addition formula \eqref{Eq-addition formulas of the sigma functions} for another time, we obtain 
		\begin{equation}\label{Eq-det-B5}
			\begin{split}
				&  \left(\frac{1}{2}u\textbf{M}(\mathbf{\Lambda}^\dagger)^{-1}+\phi^{(1)}\phi_A^\dagger\right)_{ij}= \frac{\sqrt{\nu_0} \sigma(\kappa) \sigma(\xi+\kappa)}{\sigma(\rho) \sigma(2 \kappa) \sigma^2(\xi-\kappa)} e^{-F(\xi,t)}\\
				&\left(\frac{\sigma(\kappa-z_j^*) \sigma(\kappa+z_j^*) \sigma(z_i+\rho) \sigma(z_i+z_j^*-\xi-\rho)}{\sigma(z_i+z_j^*) \sigma(z_j^*-\rho)} d_0(z_i) d_0^*(z_j) E(z_i) E^*(z_j)\right. \\
				&\quad  +\alpha_i \frac{\sigma(-\kappa+z_j^*) \sigma(\kappa+z_j^*) \sigma(-z_i-\kappa) \sigma(-2 \rho-\kappa-\xi-z_i+z_j^*)}{\sigma(\kappa+\rho+z_i-z_j^*) \sigma(z_j^*-\rho)} d_0(\hat{z}_i) d_0^*(z_j) E(\hat{z}_i) E^*(z_j)\\
				&\quad  +\alpha_j^* \frac{\sigma(\rho-z_j^*) \sigma(2 \kappa+\rho-z_j^*) \sigma(-\rho-z_i) \sigma(-z_i+z_j^*-\kappa+\xi)}{\sigma(\kappa+\rho+z_i-z_j^*) \sigma(\kappa-z_j^*)} d_0(z_i) d_0^*(\hat{z}_j) E(z_i) E^*(\hat{z}_j)I_0(\xi)I_j \\
				& \quad \left.+\alpha_i \alpha_j^* \frac{\sigma(\rho-z_j^*) \sigma(2 \kappa+\rho-z_j^*) \sigma(-\kappa-z_i) \sigma(z_i+z_j^*+\rho+\xi)}{\sigma(z_i+z_j^*) \sigma(\kappa-z_j^*)} d_0(\hat{z}_i) d_0^*(\hat{z}_j) E(\hat{z}_i) E^*(\hat{z}_j)I_0(\xi)I_j\right).
			\end{split}
		\end{equation}
	Ultimately, combining with \eqref{Eq-DNLS-elliptic-solution}, \eqref{Eq-solution to the Lax pair-1}, \eqref{Eq-solution to the Lax pair-2}, \eqref{Eq-Mtilde} and \eqref{Eq-addition formulas of the sigma functions}, we have
		\begin{equation}\label{Eq-det-B6}
			\begin{split}
				&\quad \left(-\frac{1}{2}u\textbf{M}^\dagger\left(\mathbf{\Lambda}^\dagger\right)^{-1}+\phi^{(2)}\phi^{(2)\dagger}\right)_{ij}\\
				&=\frac{1}{\sigma(2\kappa)\sigma(\xi-\kappa)} \left(\frac{\sigma(-z_j^*-z_i-\kappa+\xi)\sigma(2\kappa+\rho-z_j^*)\sigma(-\rho-z_i)}{\sigma(z_i+z_j^*)}d_0(\hat{z}_i)d_0^*(\hat{z}_j)E(z_i)E^*(z_j)I_j\right.\\
				&\quad +\alpha_i\frac{\sigma(\rho+z_i-z_j^*+\xi)\sigma(\kappa+z_i)\sigma(2\kappa+\rho-z_j^*)}{\sigma(-\kappa-\rho-z_i+z_j^*)}d_0(z_i)d_0^*(\hat{z}_j)E(\hat{z}_i)E^*(z_j)I_j\\
				&\quad +\alpha_j^*\frac{\sigma(z_j^*-z_i-2\kappa-\rho+\xi)\sigma(\kappa+z_j^*)\sigma(\rho+z_i)}{\sigma(\kappa+\rho+z_i-z_j^*)}d_0(\hat{z}_i)d_0^*(z_j)E(z_i)E^*(\hat{z}_j)I_0(\xi)\\
				&\quad \left.+\alpha_i\alpha_j^*\frac{\sigma(z_j^*+z_i-\kappa+\xi)\sigma(\kappa+z_i)\sigma(\kappa+z_j^*)}{\sigma(z_i+z_j^*)}d_0(z_i)d_0^*(z_j)E(\hat{z}_i)E^*(\hat{z}_j)I_0(\xi)\right).
			\end{split}
		\end{equation}
		Therefore, after further simplification using  \eqref{Eq-DNLS-elliptic-solution}, \eqref{Eq-Mij}, \eqref{Eq-uN-BT-infty}, \eqref{Eq-Mtilde}, 
		\eqref{Eq-det-B5} and \eqref{Eq-det-B6}, we establish the theorem's conclusions.
	\end{proof}
	As a by-product, we obtain that the modulus of the $N$-elliptic localized solution $u_N$ can be expressed as 
	\begin{equation}\label{Eq-modulus}
		|u_N|=\left|\frac{\sqrt{\nu_0}\sigma(\kappa)}{\sigma(\rho)}\left(\frac{\sigma(\xi-\kappa)}{\sigma(\xi+\rho)}\right)^{N-1}\frac{\det(\mathbf{B}_N^{(5)})}{\det(\mathbf{B}_N^{(6)})}e^{-F(\xi,t)}\right|.
	\end{equation}
	Actually, using \eqref{Eq-DNLS-elliptic-solution} and \eqref{Eq-uN-BT-infty}, we arrive at
	\begin{equation}
		|u_N|=|u^{1-N}|\left|\frac{\det\big(\frac{1}{2}u\textbf{M}(\mathbf{\Lambda}^\dagger)^{-1}+\phi^{(1)}\phi_A^\dagger\big)}{\det\big(-\frac{1}{2}u\textbf{M}^\dagger\left(\mathbf{\Lambda}^\dagger\right)^{-1}+\phi^{(2)}\phi^{(2)\dagger}\big)}\right|,
	\end{equation}
	where $u$ denotes the elliptic solution given in \eqref{Eq-DNLS-elliptic-solution}. Combining \eqref{Eq-DNLS-elliptic-solution}, \eqref{Eq-det-B5} and \eqref{Eq-det-B6}, direct evaluation yields \eqref{Eq-modulus}.

	Building upon the derivative-free form of the $N$-elliptic localized solution presented in Theorem \ref{thm:solution-2}, we can directly derive the corresponding asymptotic formula in the derivative-free form.
	\begin{theorem}[Equivalent derivative-free form of asymptotic solutions]\label{thm:asym-BT-infty}
	When $\Re(\beta_i)>0,i=1,2,\ldots,N$ and $v_i<v_j,1\leq i<j\leq N$,
	the asymptotic behavior of the $N$-elliptic localized solution along propagation directions $L_k^{-}$ as $t\rightarrow -\infty$ is given by
	\begin{equation}\label{Eq-asym-Lkm-infty}
		\begin{split}
			u_N&= u_{N,L_k^-}+\mathcal{O}\big(\exp\big(\mathrm{min}_{i\neq k}|\beta_i(v_i-v_k)|t\big)\big),\quad
			u_{N,L_k^-}=\sqrt{\nu_0}\frac{\sigma(\kappa)}{\sigma(\rho)}\prod_{i=1}^2\frac{D^{(2i+1)}_{N,L_k^-}}{D^{(2i+2)}_{N,L_k^-}}C_{N,L_k^-}e^{-F(\xi,t)},\\
		\end{split}
	\end{equation}	
	where				
	\begin{equation}
		\begin{split}		
			D_{N,L_{k}^-}^{(s)}&=\sum_{m,n=0}^{1}(-1)^{n}\alpha_{N,L_k^-}^{m}\big(\alpha_{N,L_k^-}^*\big)^{n}I_0^{n}\big(\xi+z_{N,L_k^-}\big)I_k^{\frac{1+(-1)^s}{2}+(-1)^{s+1}n}\\
			&\quad\Sigma^{(s)}\left(\xi+z_{N,L_k^-};\frac{z_k+\hat{z}_k+(-1)^{m}(z_k-\hat{z}_k)}{2},\frac{z_k^*+\check{z}_k+(-1)^{n}\big(z_k^*-\check{z}_k\big)}{2}\right)\\
			&\quad d_0\left(\frac{z_k+\hat{z}_k+(-1)^{m+(s\quad\mathrm{mod}5)}(z_k-\hat{z}_k)}{2}\right) d_0^*\left(\frac{z_k+\hat{z}_k-(-1)^{s+n}(z_k-\hat{z}_k)}{2}\right)\\
			&\quad 
			E\left(\xi;\frac{z_k+\hat{z}_k+(-1)^{m}(z_k-\hat{z}_k)}{2}\right)E^*\left(\xi;\frac{z_k+\hat{z}_k+(-1)^{n}(z_k-\hat{z}_k)}{2}\right), \quad s=3,4,5,6.
		\end{split}
	\end{equation}
	This representation remains valid along $L_k^{+}$ as $t\rightarrow +\infty$, with $z_{N,L_k^-}$, $\Delta_{N,L_k^-}$, $\alpha_{N,L_k^-}$, $C_{N,L_k^-}$,  \quad and $\mathcal{O}\big(\exp\big(\mathrm{min}_{i\neq k}|\beta_i(v_i-v_k)|t\big)\big)$ replaced by $z_{N,L_k^+}$, $\Delta_{N,L_k^+}$, $\alpha_{N,L_k^+}$, $C_{N,L_k^+}$, and $\mathcal{O}\big(\exp\big(-\mathrm{min}_{i\neq k}|\beta_i(v_i-v_k)|t\big)\big)$, respectively.
	\end{theorem}
	\begin{proof}
		We begin by analyzing the asymptotic behavior of the determinants in the solution expression \eqref{Eq-uN-2}, which requires separating bounded and unbounded components of each matrix entry. These entries can be expressed as 
		\begin{equation}\label{Eq-MNij}
			\begin{split}
				\quad (\mathbf{B}^{(3)}_N)_{ij}&=e^{\frac{\zeta(\omega_1)}{2\omega_1}(\xi^2+2\kappa\xi)}\begin{pmatrix}
					\tilde{E}(\xi;z_i) & \alpha_i \tilde{E}(\xi;\hat{z}_i)
				\end{pmatrix} \\
				&\quad \begin{pmatrix}
					\Delta^{(3)}(\xi;z_i,z_j^*)d_0(\hat{z}_i)d_0^*(z_j) &      -\Delta^{(3)}(\xi;z_i,\check{z}_j)d_0(\hat{z}_i)d_0^*(\hat{z}_j)I_j\\
					\Delta^{(3)}(\xi;\hat{z}_i,z_j^*)d_0(z_i)d_0^*(z_j)
					&  -    \Delta^{(3)}(\xi;\hat{z}_i,\check{z}_j)d_0(z_i)d_0^*(\hat{z}_j)I_j
				\end{pmatrix}\begin{pmatrix}
					\tilde{E}^*(\xi;z_j)\\ \alpha_j^* \tilde{E}^*(\xi;\hat{z}_j)
				\end{pmatrix},\\
				(\mathbf{B}^{(4)}_N)_{ij}&=e^{\frac{\zeta(\omega_1)}{2\omega_1}(\xi^2-2\kappa\xi)}\begin{pmatrix}
					\tilde{E}(\xi;z_i) & \alpha_i \tilde{E}(\xi;\hat{z}_i)
				\end{pmatrix}\\
				&\quad \begin{pmatrix}
					\Delta^{(4)}(\xi;z_i,z_j^*)d_0(z_i)d_0^*(\hat{z}_j)I_j &  -\Delta^{(4)}(\xi;z_i,\check{z}_j)d_0(z_i)d_0^*(z_j)\\
					\Delta^{(4)}(\xi;\hat{z}_i,z_j^*)d_0(\hat{z}_i)d_0^*(\hat{z}_j)I_j
					&  -\Delta^{(4)}(\xi;\hat{z}_i,\check{z}_j)d_0(\hat{z}_i)d_0^*(z_j)
				\end{pmatrix}\begin{pmatrix}
					\tilde{E}^*(\xi;z_j)\\ \alpha_j^* \tilde{E}^*(\xi;\hat{z}_j)
				\end{pmatrix},\\ (\mathbf{B}^{(5)}_N)_{ij}&=e^{\frac{\zeta(\omega_1)}{2\omega_1}(\xi^2+2\rho\xi)}
				\begin{pmatrix}
					\tilde{E}(\xi;z_i) & \alpha_i \tilde{E}(\xi;\hat{z}_i)
				\end{pmatrix}\\
				&\quad \begin{pmatrix}
					\Delta^{(5)}(\xi;z_i,z_j^*)d_0(z_i)d_0^*(z_j) &  -\Delta^{(5)}(\xi;z_i,\check{z}_j)d_0(z_i)d_0^*(\hat{z}_j)I_
					j\\
					\Delta^{(5)}(\xi;\hat{z}_i,z_j^*)d_0(\hat{z}_i)d_0^*(z_j) 
					&  -\Delta^{(5)}(\xi;\hat{z}_i,\check{z}_j)d_0(\hat{z}_i)d_0^*(\hat{z}_j)I_
					j
				\end{pmatrix}    \begin{pmatrix}
					\tilde{E}^*(\xi;z_j)\\ \alpha_j^* \tilde{E}^*(\xi;\hat{z}_j)
				\end{pmatrix},\\
				(\mathbf{B}^{(6)}_N)_{ij}&=e^{\frac{\zeta(\omega_1)}{2\omega_1}(\xi^2-2\kappa\xi)}
				\begin{pmatrix}
					\tilde{E}(\xi;z_i) & \alpha_i \tilde{E}(\xi;\hat{z}_i)
				\end{pmatrix}\\ &\quad \begin{pmatrix}
					\Delta^{(6)}(\xi;z_i,z_j^*)d_0(\hat{z}_i)d_0^*(\hat{z}_j)I_j &  -\Delta^{(6)}(\xi;z_i,\check{z}_j)d_0(\hat{z}_i)d_0^*(z_j)\\
					\Delta^{(6)}(\xi;\hat{z}_i,z_j^*)d_0(z_i)d_0^*(\hat{z}_j)I_j 
					&  -\Delta^{(6)}(\xi;\hat{z}_i,\check{z}_j)d_0(z_i)d_0^*(z_j)
				\end{pmatrix}    \begin{pmatrix}
					\tilde{E}^*(\xi;z_j)\\ \alpha_j^* \tilde{E}^*(\xi;\hat{z}_j)
				\end{pmatrix},\\
			\end{split}
		\end{equation}
		where
		\begin{equation}
			\begin{split}
				\Delta^{(3)}(\xi;\blacktriangle,\bullet)&=\Sigma^{(3)}(\xi;\blacktriangle,\bullet)e^{-\frac{\zeta(\omega_1)}{2\omega_1}\big(\xi^2-2(\blacktriangle+\bullet-\kappa)\xi\big)},\\
				\Delta^{(4)}(\xi;\blacktriangle,\bullet)&=\Sigma^{(4)}(\xi;\blacktriangle,\bullet)e^{-\frac{\zeta(\omega_1)}{2\omega_1}\big(\xi^2-2(\blacktriangle+\bullet+\kappa)\xi\big)},\\
				\Delta^{(5)}(\xi;\blacktriangle,\bullet)&=\Sigma^{(5)}(\xi;\blacktriangle,\bullet)e^{-\frac{\zeta(\omega_1)}{2\omega_1}\big(\xi^2-2(\blacktriangle+\bullet-\rho)\xi\big)},\\
				\Delta^{(6)}(\xi;\blacktriangle,\bullet)&=\Sigma^{(6)}(\xi;\blacktriangle,\bullet)e^{-\frac{\zeta(\omega_1)}{2\omega_1}\big(\xi^2-2(\blacktriangle+\bullet+\kappa)\xi\big)}.
			\end{split}
		\end{equation}    
		The functions $\Delta^{(i)}$, $i=3,4,5,6$, remain bounded in $\xi$ for arbitrary $\blacktriangle,\bullet\in\mathbb{C}$. Following a procedure analogous to that in the proof of Theorem \ref{thm:AB1}, we derive the asymptotic expression for the solution \eqref{Eq-uN-2}:
		\begin{align}\label{Eq-asymptotic-formula-1-infty}
			u_{N,L_k^-}=\left(\frac{\sqrt{\nu_0}\sigma(\kappa)\sigma(\xi+\rho)\sigma(\xi+\kappa)}{\sigma(\rho)\sigma^2(\xi-\kappa)}e^{-F(\xi,t)}\right) \left(-\frac{\sigma(\xi+\rho)\sigma(\xi+\kappa)}{\sigma^2(\xi-\kappa)}\right)^{-N}\frac{U^{(3)}_{N,L_k^{-}}}{U^{(4)}_{N,L_k^{-}}},
		\end{align}
		where 
		\begin{equation}
			\begin{split}
				U^{(3)}_{N,L_k^{-}}
				&=\prod_{i=1}^{2}\sum_{m,n=1}^2(-1)^{n-1}I_k^{n-1}e^{(m-1)\tau_k+(n-1)\tau_k^*}\det\left(\textbf{V}_{N,L_k^-}^{[(2i+1),m,n]}\right),\\
				U^{(4)}_{N,L_k^{-}}
				&=\prod_{i=1}^{2}\sum_{m,n=1}^2(-1)^{n-1}I_k^{2-n}e^{(m-1)\tau_k+(n-1)\tau_k^*}\det\left(\textbf{V}_{N,L_k^-}^{[(2i+2),m,n]}\right).\\
			\end{split}
		\end{equation}
		The determinants appearing in $U^{(3)}_{N,L_k^{-}}$ and $U^{(4)}_{N,L_k^{-}}$ can be evaluated as 
		\begin{equation}
			\begin{split}
				\det\left(\textbf{V}_{N,L_k^-}^{[(3),m,n]}\right)&=\left(\prod_{i=1}^N \frac{d_0(\hat{\delta}_i^{[m,n]})d_0^*(\hat{\varepsilon}_i^{[m,n]})}{d_0(\delta_i^{[m,n]})d_0^*(\varepsilon_i^{[m,n]})}r^{-1}(\delta_{i}^{[m,n]})s(\eta_{i}^{[m,n]})\right)	\Pi_{N,L_k^-}^{[m,n]}(\xi;\kappa)I_{k+},\\  
				\det\left(\textbf{V}_{N,L_k^-}^{[(4),m,n]}\right)&= (-1)^N\Pi_{N,L_k^-}^{[m,n]}(\xi;-\kappa)I_{k-},\\
				\det\left(\textbf{V}_{N,L_k^-}^{[(5),m,n]}\right)&=(-1)^N\left(\prod_{i=1}^N\frac{d_0^*(\hat{\varepsilon}_i^{[m,n]})}{d_0^*(\varepsilon_i^{[m,n]})}\frac{\sigma(-\rho-\delta_i^{[m,n]})\sigma(\kappa-\eta_i^{[m,n]})\sigma(\kappa+\eta_i^{[m,n]})}{\sigma(\kappa-\delta_i^{[m,n]})\sigma(\eta_i^{[m,n]}-\rho)\sigma(2\kappa+\rho-\eta_i^{[m,n]})}\right)\Pi_{N,L_k^-}^{[m,n]}(\xi;\rho)I_{k+},\\
				\det\left(\textbf{V}_{N,L_k^-}^{[(6),m,n]}\right)&=(-1)^N\prod_{i=1}^N \frac{d_0(\hat{\delta}_i^{[m,n]})}{d_0(\delta_i^{[m,n]})}\frac{\sigma(-\rho-\delta_i^{[m,n]})}{\sigma(\kappa-\delta_i^{[m,n]})}\Pi_{N,L_k^-}^{[m,n]}(\xi;-\kappa)I_{k-},\\
			\end{split}
		\end{equation}
		where $	I_{k-}=\prod_{j=1}^{k-1}I_j,\quad I_{k+}=\prod_{j=k+1}^{N}I_j$.
		We observe that 
		\begin{align}
			\left(\frac{d_0^*(\hat{\varepsilon}_i^{[m,n]})}{d_0^*(\varepsilon_i^{[m,n]})}\right)^2\frac{\sigma(\kappa-\eta_i^{[m,n]})\sigma(\kappa+\eta_i^{[m,n]})}{\sigma(\eta_i^{[m,n]}-\rho)\sigma(2\kappa+\rho-\eta_i^{[m,n]})}=\left\{\begin{aligned}
				&I_i^2,\quad\quad  1\leq i\leq k-1,\\
				&I_i^{-2},\quad k+1\leq i\leq N.
			\end{aligned}
			\right.
		\end{align}
		Consequently, 
		\begin{align}
			\prod_{i=1,i\neq k}^{N}    \left(\frac{d_0^*(\hat{\varepsilon}_i^{[m,n]})}{d_0^*(\varepsilon_i^{[m,n]})}\right)^2\frac{\sigma(\kappa-\eta_i^{[m,n]})\sigma(\kappa+\eta_i^{[m,n]})}{\sigma(\eta_i^{[m,n]}-\rho)\sigma(2\kappa+\rho-\eta_i^{[m,n]})}=\left(\frac{I_{k-}}{I_{k+}}\right)^2.
		\end{align}
		This leads to the simplified asymptotic solution
		\begin{align}\label{Eq-u-Lk--infty}
			u_{N,L_k^-}= \frac{\sqrt{\nu_0}\sigma(\kappa)}{\sigma(\rho)}\frac{\tilde{U}^{(3)}_{N,L_k^{-1}}}{\tilde{U}^{(4)}_{N,L_k^{-1}}}C_{N,L_k^-}e^{-F(\xi,t)},
		\end{align}
		where
		\begin{equation}\label{Eq-u-Lk--infty-2}
			\begin{split}
				\tilde{U}^{(3)}_{N,L_k^{-}}
				&=\prod_{i=1}^{2}\sum_{m,n=1}^2(-1)^{n-1}I_k^{n-1}e^{(m-1)\tau_k+(n-1)\tau_k^*}\tilde{V}_{N,L_k^-}^{[(2i+1),m,n]},\\
				\tilde{U}^{(4)}_{N,L_k^{-}}
				&=\prod_{i=1}^{2}\sum_{m,n=1}^2(-1)^{n-1}I_k^{2-n}e^{(m-1)\tau_k+(n-1)\tau_k^*}\tilde{V}_{N,L_k^-}^{[(2i+2),m,n]},
					\end{split}
			\end{equation}
		and	
				\begin{equation}\label{Eq-u-Lk--infty-3}
				\begin{split}	
				\tilde{V}_{N,L_k^-}^{[(3),m,n]}&=\frac{d_0(\hat{\delta}_k^{[m,n]})d_0^*(\hat{\varepsilon}_k^{[m,n]})}{d_0(\delta_k^{[m,n]})d_0^*(\varepsilon_k^{[m,n]})}r^{-1}(\delta_{k}^{[m,n]})s(\eta_{k}^{[m,n]})\tilde{\Pi}_{N,L_k^-}^{[(1),m,n]}\tilde{\Pi}_{N,L_k^-}^{[(2),m,n]}\\
				&\sigma\big(-\xi-\kappa+\sum_{i=1}^N(\delta_i^{[m,n]}+\eta_i^{[m,n]})\big)e^{(\delta_k^{[m,n]}+\eta_k^{[m,n]})\frac{\zeta(\omega_1)}{\omega_1}\xi},\\  
				\tilde{V}_{N,L_k^-}^{[(4),m,n]}&= \tilde{\Pi}_{N,L_k^-}^{[(1),m,n]}\tilde{\Pi}_{N,L_k^-}^{[(2),m,n]}\sigma\big(-\xi+\kappa+\sum_{i=1}^N(\delta_i^{[m,n]}+\eta_i^{[m,n]})\big)e^{(\delta_k^{[m,n]}+\eta_k^{[m,n]})\frac{\zeta(\omega_1)}{\omega_1}\xi},\\
				\tilde{V}_{N,L_k^-}^{[(5),m,n]}&=\frac{d_0^*(\hat{\varepsilon}_k^{[m,n]})}{d_0^*(\varepsilon_k^{[m,n]})}\frac{\sigma(-\rho-\delta_k^{[m,n]})\sigma(\kappa-\eta_k^{[m,n]})\sigma(\kappa+\eta_k^{[m,n]})}{\sigma(\kappa-\delta_k^{[m,n]})\sigma(\eta_k^{[m,n]}-\rho)\sigma(2\kappa+\rho-\eta_k^{[m,n]})}\tilde{\Pi}_{N,L_k^-}^{[(1),m,n]}\tilde{\Pi}_{N,L_k^-}^{[(2),m,n]}\\
				&\quad \sigma\big(-\xi-\rho+\sum_{i=1}^N(\delta_i^{[m,n]}+\eta_i^{[m,n]})\big)e^{(\delta_k^{[m,n]}+\eta_k^{[m,n]})\frac{\zeta(\omega_1)}{\omega_1}\xi},\\
				\tilde{V}_{N,L_k^-}^{[(6),m,n]}&=\frac{d_0(\hat{\delta}_k^{[m,n]})}{d_0(\delta_k^{[m,n]})}\frac{\sigma(-\rho-\delta_k^{[m,n]})}{\sigma(\kappa-\delta_k^{[m,n]})}\tilde{\Pi}_{N,L_k^-}^{[(1),m,n]}\tilde{\Pi}_{N,L_k^-}^{[(2),m,n]}\sigma\big(-\xi+\kappa+\sum_{i=1}^N(\delta_i^{[m,n]}+\eta_i^{[m,n]})\big)\\
				&\quad e^{(\delta_k^{[m,n]}+\eta_k^{[m,n]})\frac{\zeta(\omega_1)}{\omega_1}\xi}.\\
			\end{split}
		\end{equation}
		Based on \eqref{Eq-tildeE}-\eqref{Eq-tau}, \eqref{Eq-N-breather-asymptotic-3}, \eqref{Eq-U12-tilde}, \eqref{Eq-u-Lk--infty}-\eqref{Eq-u-Lk--infty-3}, we arrive at the conclusion of Theorem \ref{thm:AB1} with a further direct simplification.
	\end{proof}
	Building on Theorem \ref{thm:asym-BT-infty}, the derivative-free form of the asymptotic formula \eqref{Eq-asymptotic-Rkp} can be directly derived.

	\section{Dynamic behaviors of the one and two-elliptic localized solutions}
	In this section, we investigate the dynamics of the one and two-elliptic localized solutions to gain deeper insight into the constructed solutions and their analytical properties. The dynamics are systematically explored, with analytical findings illustrated graphically.

	\subsection{The dynamics of the one-elliptic localized solutions}
	In this subsection, we display one-elliptic localized solutions given as \eqref{Eq-N-elliptic-localized-solution}-\eqref{Eq-N-elliptic-localized-solution-3} with
	$	\kappa = 1.57\mathrm{i},  \rho = 4.61 - 1.57\mathrm{i}, \omega_1 = 4.61, \omega_3 = -3.14\mathrm{i}$. Setting $z_1 =\dfrac{\omega _1- \omega_3}{2}$ and $\alpha_1 = 1$ in \eqref{Eq-parameterization-lambda-y} yields the spectral parameter value $
	\lambda_1 = \lambda(z_1) = -0.17 - 0.17\mathrm{i}$. Substituting these into \eqref{Eq-N-elliptic-localized-solution}-\eqref{Eq-N-elliptic-localized-solution-3} produces a one-elliptic stationary localized solution, shown in the left panel of Figure~\ref{Fig-first-order-solution-1}. As verified by Theorem~\ref{thm:symmetry}, this solution exhibits symmetry with respect to the origin. Selecting $\alpha_1 = 0.05$ while keeping other parameters unchanged yields another stationary solution lacking this symmetry, displayed in the right panel of Figure~\ref{Fig-first-order-solution-1}. Using \eqref{Eq-tau}, the fundamental period $T$ along the $t$-axis for the stationary one-elliptic localized solutions is evaluated as
	\begin{align}
		T = \frac{\pi}{8|\Im\big(y(z_1)\big)|}.
	\end{align}
	For both solutions in Figure~\ref{Fig-first-order-solution-1}, we evaluate $T \approx 30.2$, consistent with the illustrations. Numerical results shows that the maximum modulus $|u_1|$ for the solution with $\alpha_1=1$ (left panel) is $2.02$, attained at $(\xi,t)=(0,\frac{(2n+1)T}{2}),n\in \mathbb{Z}$. The minimum value is $0.02$, occurring at $(\xi,t)=(0,nT),n\in \mathbb{Z}$.
	Using \eqref{Eq-vi}, we deduce that the speeds $v_1$ of both solutions are $0$, analytically confirming their stationary nature. Substituting the parameters to $\kappa=1.9\mathrm{i},\rho={4.47-2.2\mathrm{i}},\omega_1=4.47,\omega_3=-3.51\mathrm{i}$ with $z_1$ invariant, we obtain two non-stationary solutions shown in Figure~\ref{Fig-first-order-solution-2}. The solution in the left panel, with $\alpha_1=1$, exhibits symmetry with respect to the origin. However, the solution in the right panel, where $\alpha_1=0.05$, lacks this symmetry. All solutions presented in Figures~\ref{Fig-first-order-solution-1} and~\ref{Fig-first-order-solution-2} represent one breather propagating on the elliptic background.
	
	\subsection{The dynamics of the two-elliptic localized solutions}
	The left panel of Figure \ref{Fig-second-order-solution-1} illustrates a two-elliptic localized solution given as \eqref{Eq-N-elliptic-localized-solution}-\eqref{Eq-N-elliptic-localized-solution-3} with parameters $\kappa = 1.57\mathrm{i},\rho = 4.6 - 1.57\mathrm{i},\omega_1 = 4.6, \omega_3 = -3.14\mathrm{i}$, $ z_1 =\dfrac{\omega_1}{2} - \dfrac{\omega_3}{2},z_2 = -\dfrac{\omega_1}{9} + \dfrac{\omega_3}{8}$ and $\alpha_1 = \alpha_2 = 1$. This solution features two interacting breathers propagating on the elliptic background, with one static and the other propagating at velocity $v_1=-1.72$, computed via \eqref{Eq-vi}. The choice $\alpha_1 = \alpha_2 = 1$ imposes the symmetry with respect to the origin as established in Theorem \ref{thm:symmetry}, constraining the characteristic trajectory lines to intersect precisely at the origin. With identical parameters but substituting $z_1 = \dfrac{\omega_1}{2}-\dfrac{2\omega_3}{3}, z_2 = -\dfrac{\omega_1}{9}+\dfrac{\omega_3}{6}$, we generate another two-elliptic solution featuring two nonstationary breathers, illustrated in the right panel of Figure \ref{Fig-second-order-solution-1}.
	
	\subsection{The asymptotic behaviors of the one and two-elliptic localized solutions}
	Using \eqref{Eq-asymptotic-Rk}-\eqref{Eq-asymptotic-Rkp}, we display the asymptotic solution of an one-elliptic localized solution given by \eqref{Eq-N-elliptic-localized-solution}-\eqref{Eq-N-elliptic-localized-solution-3} when $\kappa=2.08\mathrm{i},\rho=3.25-1.73\mathrm{i},\omega_1=3.25,\omega_3=-3.31\mathrm{i}$ and $z_1=1+\mathrm{i}$ in Figure \ref{Fig-Asym-first-order}. The moduli $|u_1|$ and $|u_{1,R_1^\pm}|$ are illustrated with orange solid curves and green dash dot curves in both panels of Figure \ref{Fig-Asym-first-order} respectively when $t=0$. 
	
	We visualize the asymptotic behaviors of a two-elliptic localized solution given by \eqref{Eq-N-elliptic-localized-solution}-\eqref{Eq-N-elliptic-localized-solution-3} with parameters $\kappa=2.08\mathrm{i},\rho={3.25-1.73\mathrm{i}},\omega_1=3.25,\omega_3=-3.31\mathrm{i}$, $z_1=-0.46-2.06\mathrm{i},z_2=-0.65+4.41\mathrm{i}$ and $\alpha_1=\alpha_2=1$. Using \eqref{Eq-vi}, the velocities of the two elliptic localized waves are evaluated as $v_1=-0.14$ and $v_2=0.58$, respectively. Consequently, this solution describes two elliptic localized waves propagating in opposite directions. The asymptotic solutions between the propagation directions are given by \eqref{Eq-asymptotic-Rk}-\eqref{Eq-asymptotic-Rkp}. Figure \ref{Fig-Asym-second-order-positive} displays the moduli of the asymptotic solutions together with $|u_2|$ in regions $R_1^-$, $R_2^+$, and $R_1^+$ at $t=27$. The orange solid curve represents $|u_2|$, while the green dash dot curves depict $|u_{2,R_1^{-}}|$, $|u_{2,R_2^{+}}|$, and $|u_{2,R_1^{+}}|$, respectively. This demonstrates the solution's asymptotic properties in regions $R_1^-$, $R_2^+$, and $R_1^+$ as $t\rightarrow +\infty$. Additionally, Figure \ref{Fig-Asym-second-order-line-lines} illustrate the asymptotic behaviors along lines $L_1^\pm$ and $L_2^\pm$ based on 	\eqref{Eq-asym-Lkm-infty}, respectively.

	\begin{figure}[ht]
		\centering
		\includegraphics[width=0.48
		\linewidth]{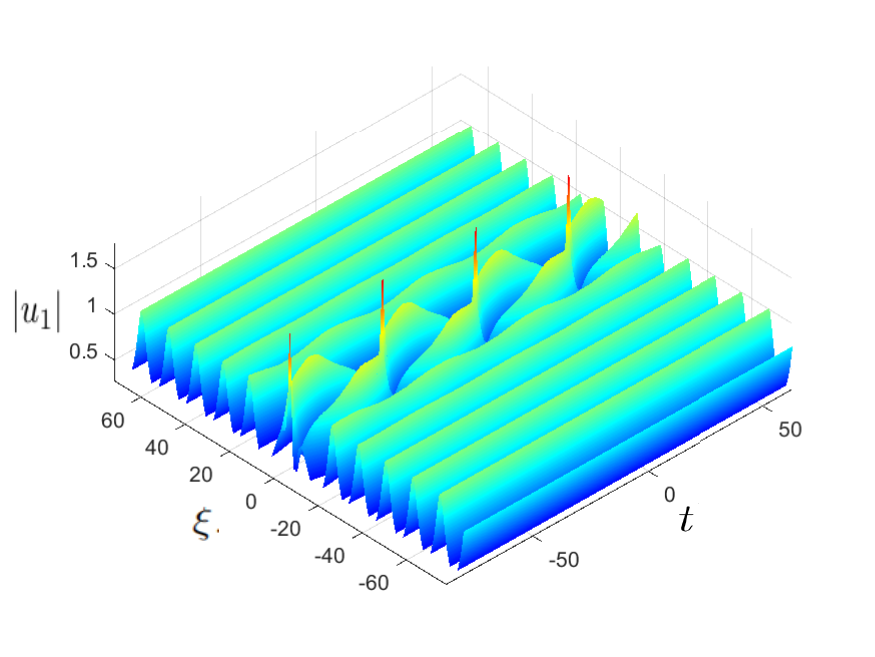}
		\includegraphics[width=0.48
		\linewidth]{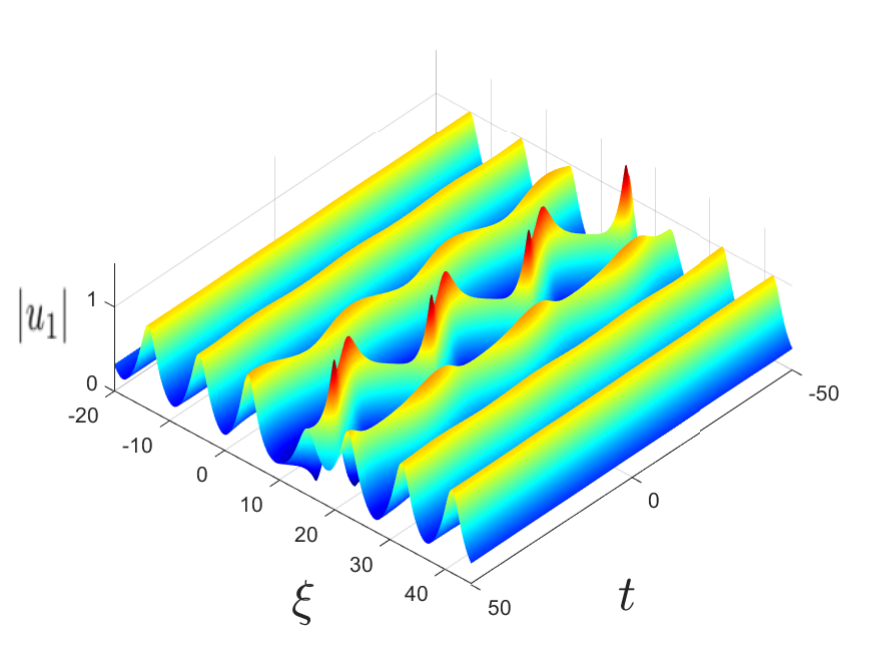}
		\caption{The modulus of the one-elliptic localized solution $u_1$ with parameters $\kappa=1.57\mathrm{i},\rho={4.61-1.57\mathrm{i}},\omega_1=4.61,\omega_3=-3.14\mathrm{i}$, $z_1=\frac{\omega_1}{2}-\frac{\omega_3}{2}$. The left panel: $\alpha_1=1$. The right panel: $\alpha_1=0.05$. }
		\label{Fig-first-order-solution-1}
	\end{figure}
	
	\begin{figure}[ht]
		\centering
		\includegraphics[width=0.48
		\linewidth]{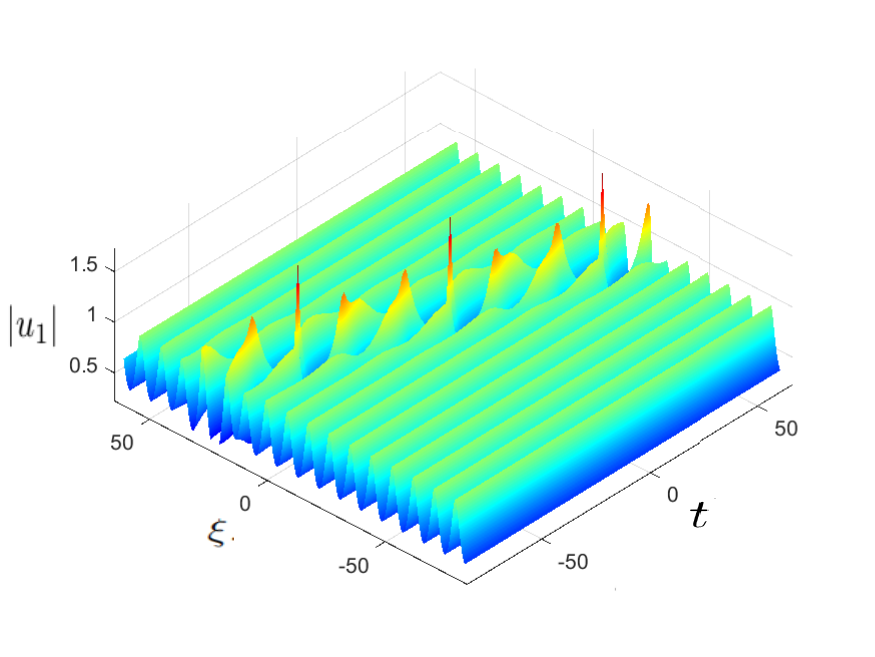}
		\includegraphics[width=0.48
		\linewidth]{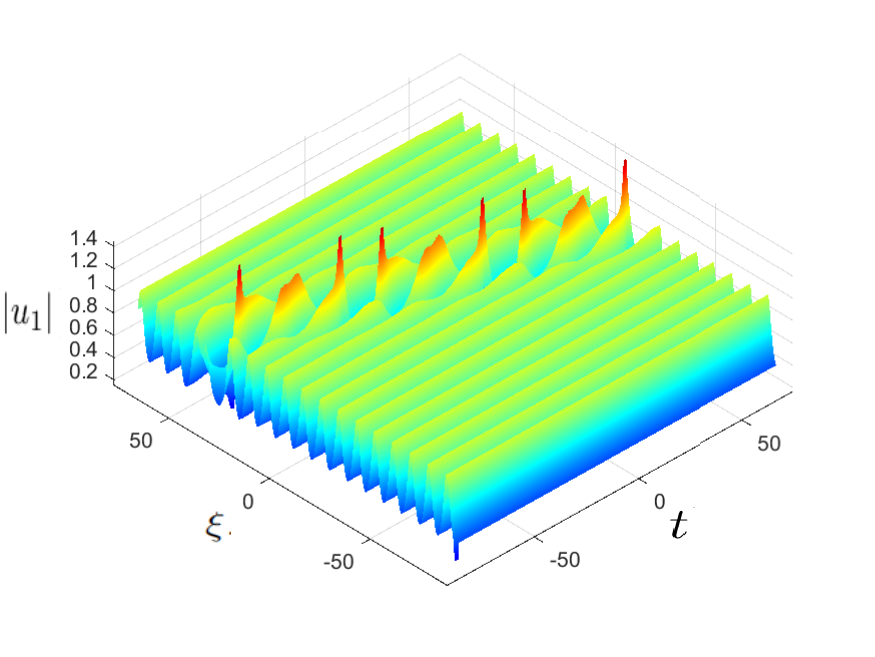}
		\caption{The modulus of the one-elliptic localized solution $u_1$ with parameters $\kappa=1.9\mathrm{i},\rho={4.47-2.2\mathrm{i}},\omega_1=4.47,\omega_3=-3.51\mathrm{i}$, $z_1=\frac{\omega_1}{2}-\frac{\omega_3}{2}$. The left panel: $\alpha_1=1$. The right panel: $\alpha_1=0.05$.}
		\label{Fig-first-order-solution-2}
	\end{figure}

	\begin{figure}[ht]
		\centering
		\includegraphics[width=0.48
		\linewidth]{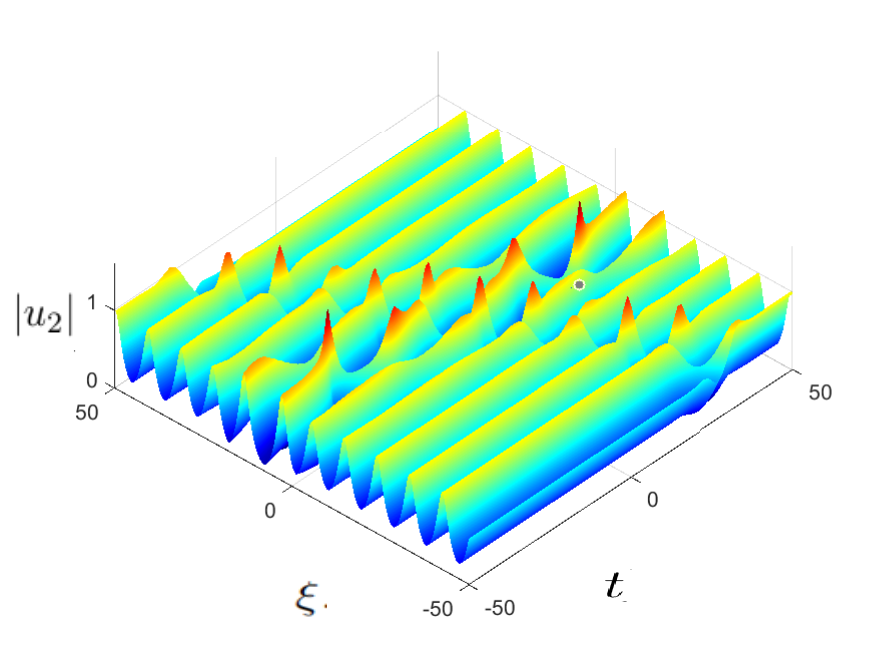}
		\includegraphics[width=0.48\linewidth]{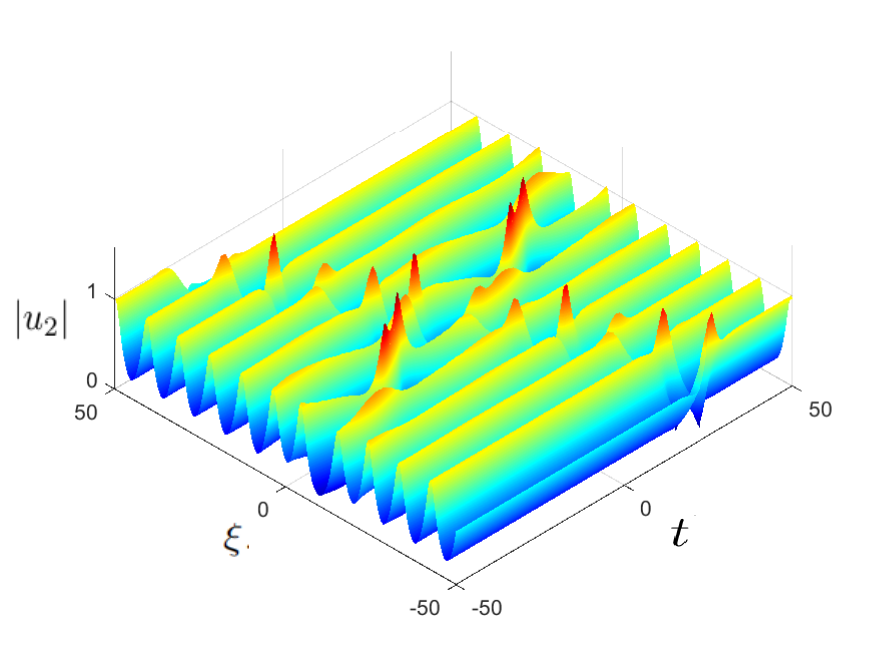}
		\caption{The modulus of the two-elliptic localized solutions $u_2$ with parameters $\kappa=1.57\mathrm{i},\rho={4.6-1.57\mathrm{i}},\omega_1=4.6,\omega_3=-3.14\mathrm{i}$ and $\alpha_1=\alpha_2=1$. The left panel: $z_1=\frac{\omega_1}{2}-\frac{\omega_3}{2},z_2=-\frac{\omega_1}{9}+\frac{\omega_3}{8}$. The right panel: $z_1=\frac{\omega_1}{2}-\frac{2\omega_3}{3},z_2=-\frac{\omega_1}{9}+\frac{\omega_3}{6}$. }
		\label{Fig-second-order-solution-1}
	\end{figure}

	\begin{figure}[ht]
		\centering
		\includegraphics[width=0.44
		\linewidth]{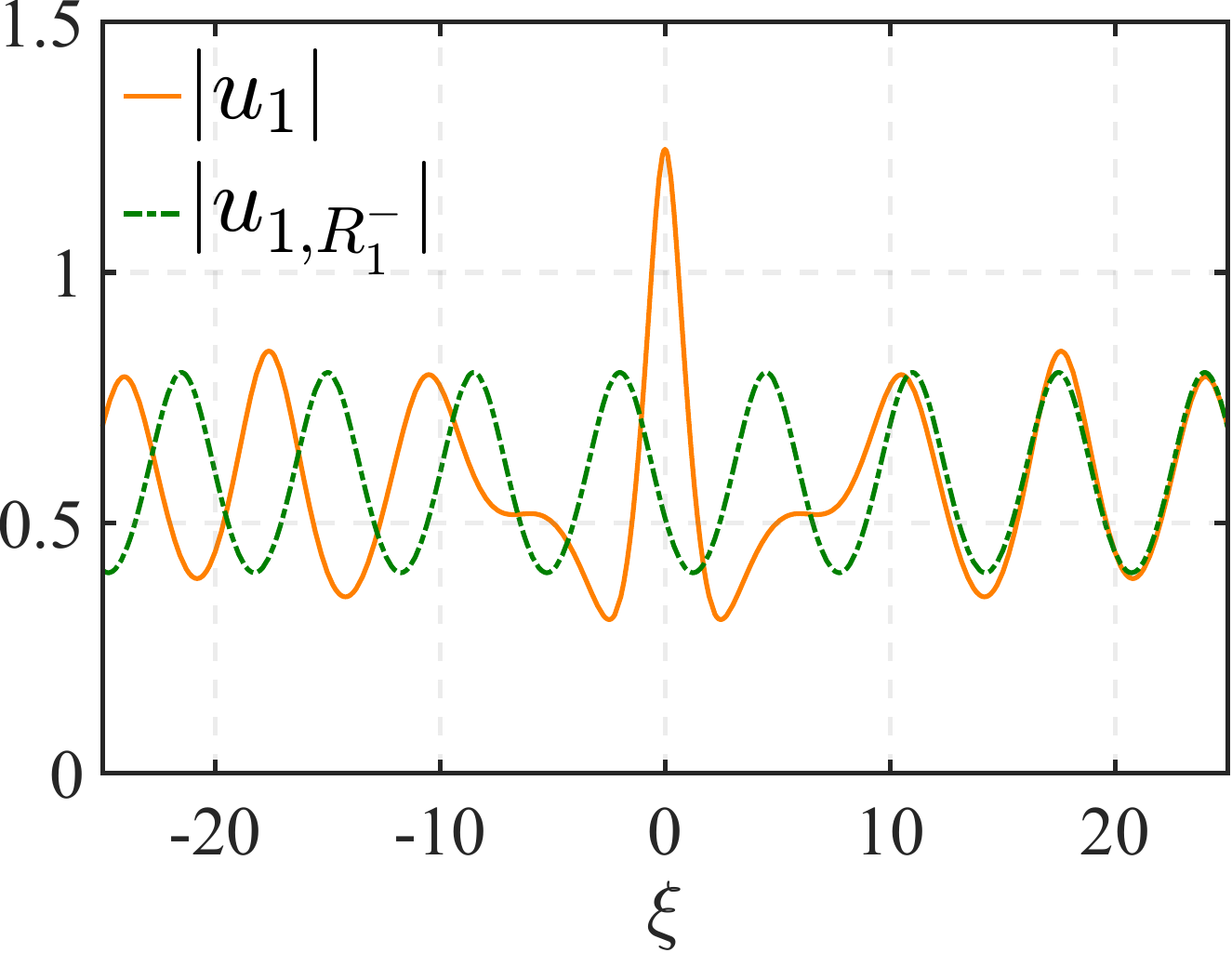}
		\includegraphics[width=0.44
		\linewidth]{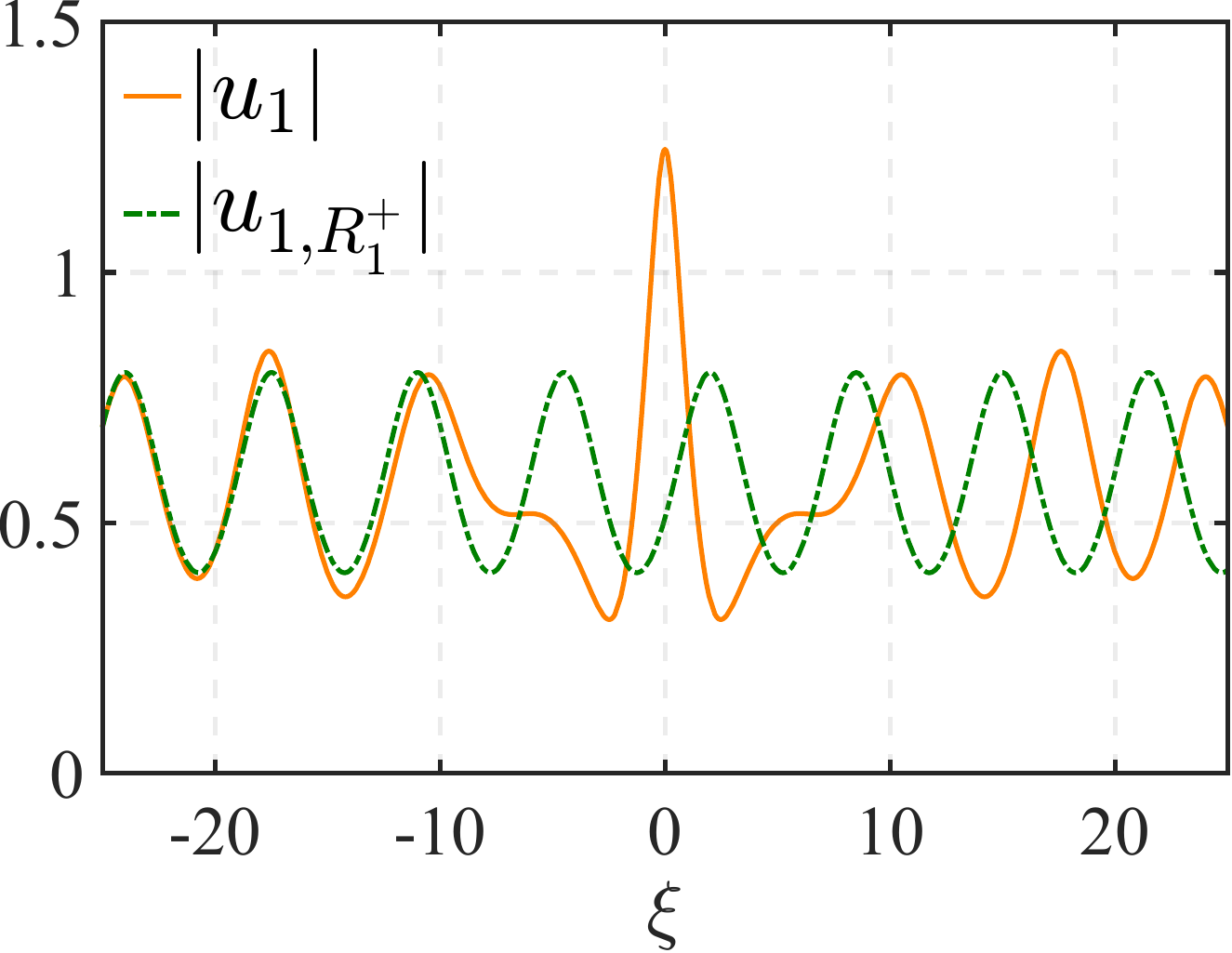}
		\caption{The asymptotic behaviors for the one-elliptic localized solution with parameters $\kappa=2.08\mathrm{i},\rho={3.25-1.73\mathrm{i}},\omega_1=3.25,\omega_3=-3.31\mathrm{i},z_1=1+\mathrm{i}$ and $\alpha_1=1$ when $t=0$. The left panel:  The orange solid curve and the green dash curve represent $|u_1|$ and $|u_{1,R_{1}^-}|$ respectively. The right panel: The orange solid curve and the green dash curve represent $|u_1|$ and $|u_{1,R_{1}^+}|$ respectively. }
		\label{Fig-Asym-first-order}
	\end{figure}

	\begin{figure}[ht]
		\centering
		(A)\includegraphics[width=0.35
		\linewidth]{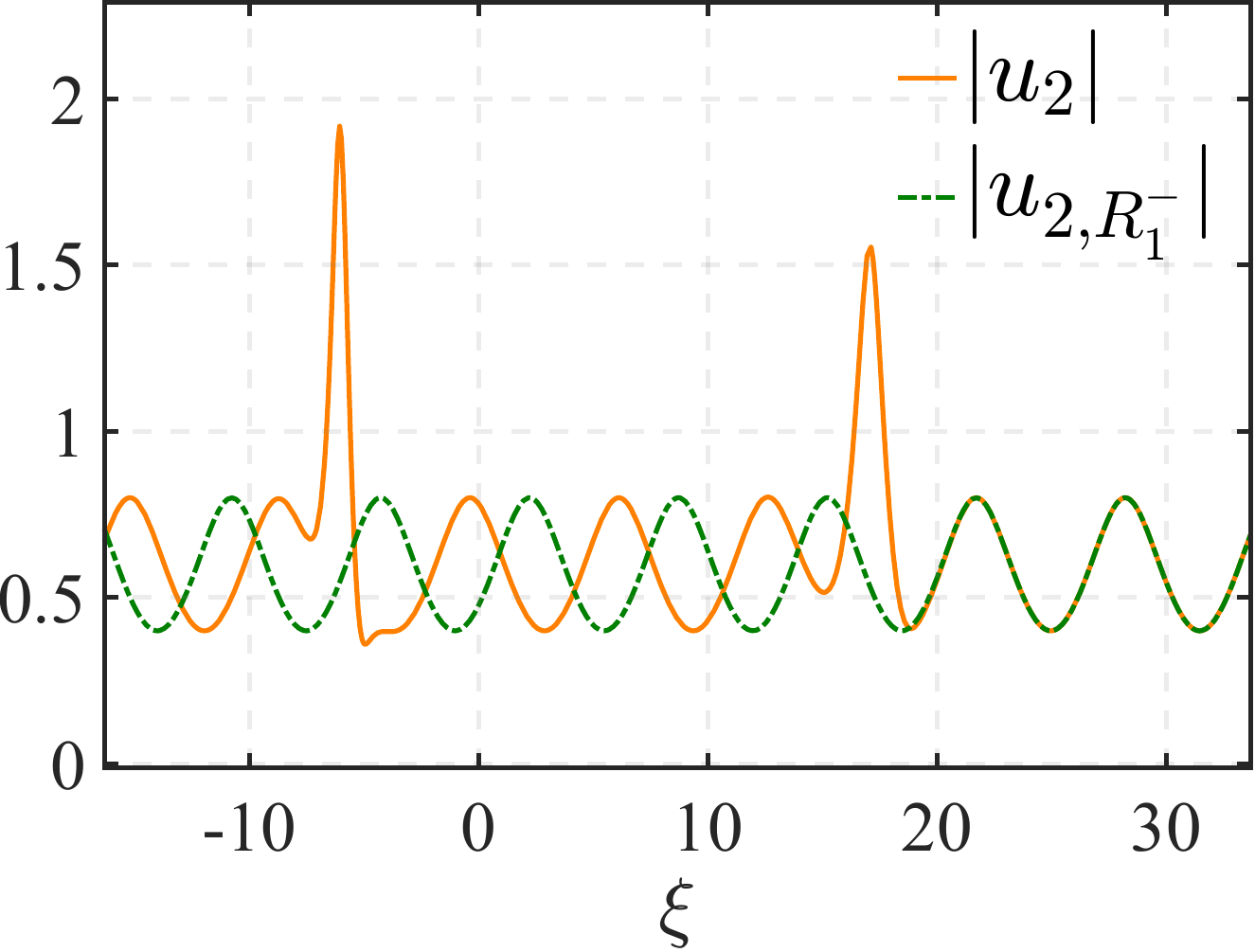} 
		(B)\includegraphics[width=0.35
		\linewidth]{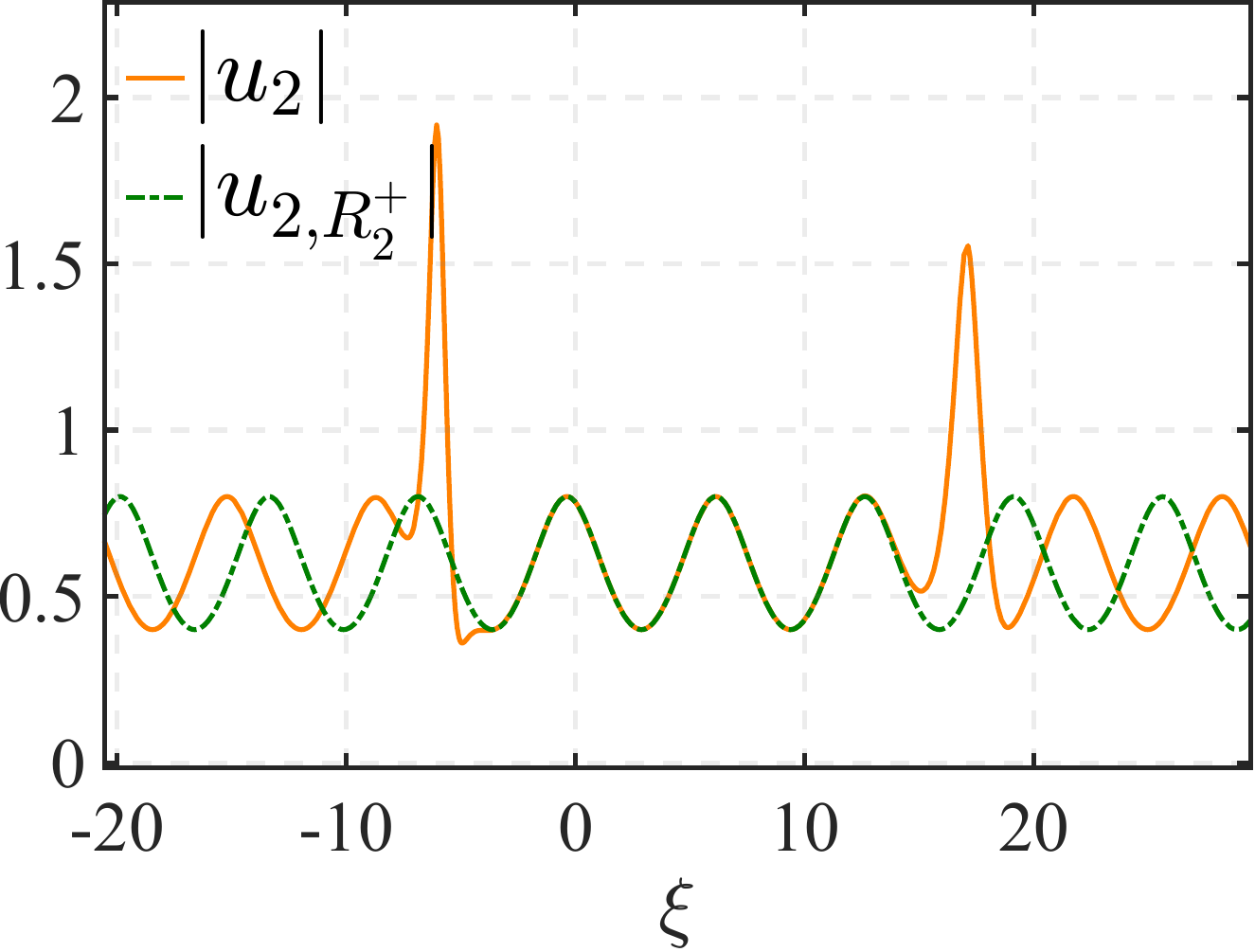}
		(C)\includegraphics[width=0.35
		\linewidth]{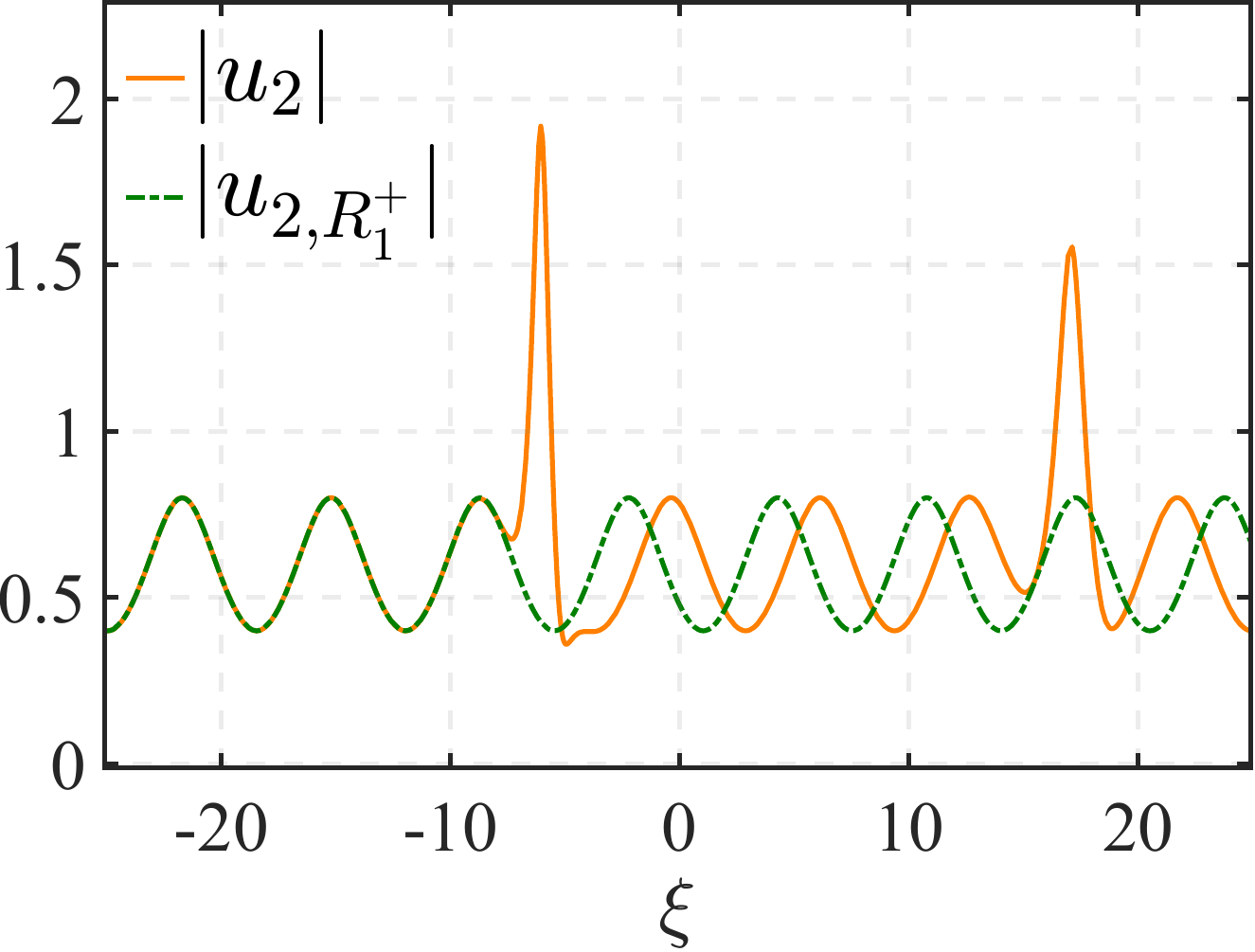}
		\caption{The asymptotic behaviors of the two-elliptic localized solution $|u_2|$ with parameters $\kappa=2.08\mathrm{i},\rho={3.25-1.73\mathrm{i}},\omega_1=3.25,\omega_3=-3.31\mathrm{i}$, $z_1=-0.46-2.06\mathrm{i},z_2=-0.65+4.41\mathrm{i}$ and $\alpha_1=\alpha_2=1$ in the regions $R_1^-,R_2^+$ and $R_1^+$ when $t=27$. (A)  The orange solid curve and the green dash dot curve represent $|u_2|$ and $|u_{2,R_1^-}|$ respectively. (B) The orange solid curve and the green dash dot curve represent $|u_2|$ and $|u_{2,R_2^+}|$  respectively. (C) The orange solid curve and the green dash dot curve represent $|u_2|$ and $|u_{2,R_1^+}|$  respectively.}
		\label{Fig-Asym-second-order-positive}
	\end{figure}
	
	\begin{figure}[ht]
		\centering
		(A)\includegraphics[width=0.35
		\linewidth]{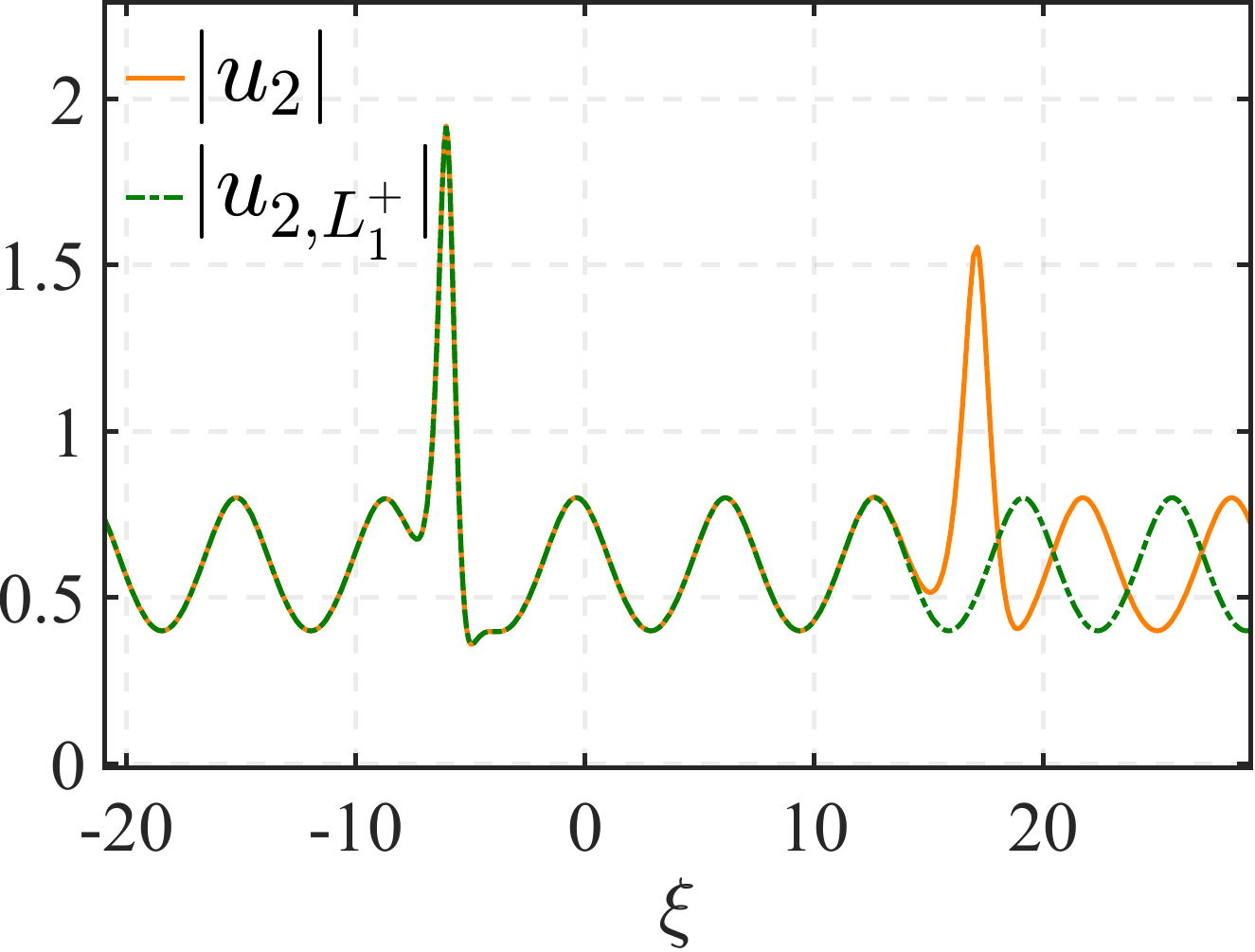}
		(B)\includegraphics[width=0.35
		\linewidth]{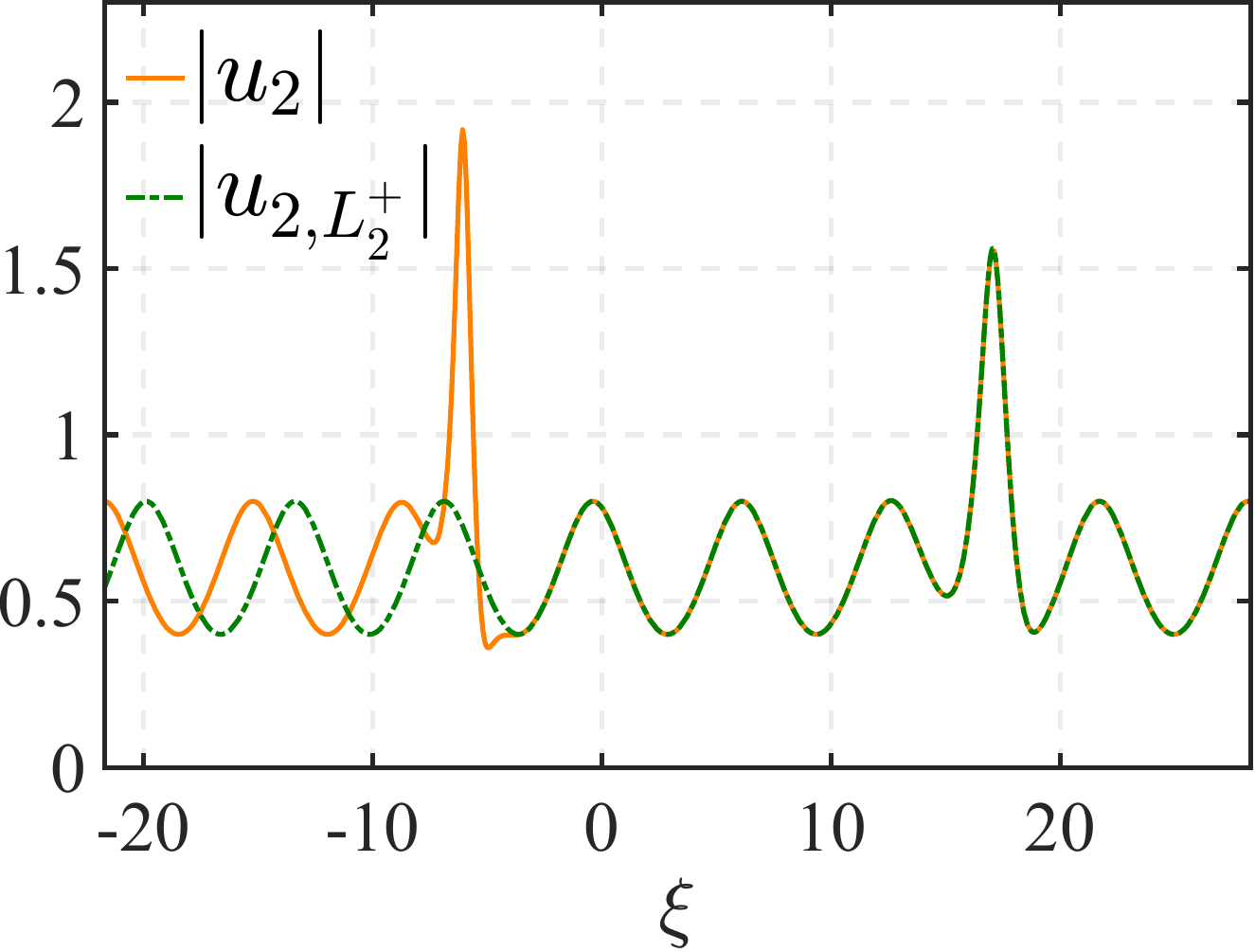}\\
		(C)\includegraphics[width=0.35
		\linewidth]{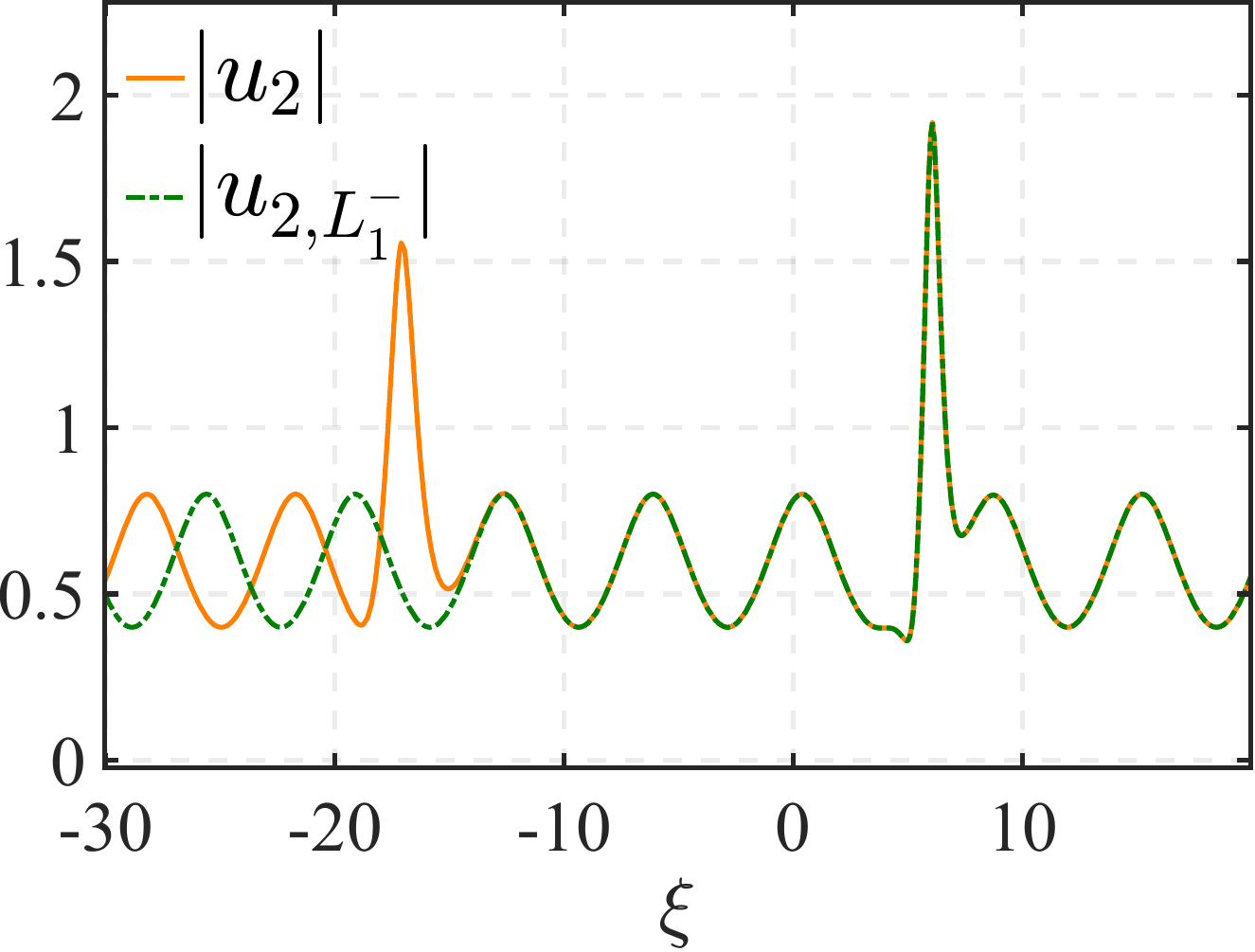}
		(D)\includegraphics[width=0.35
		\linewidth]{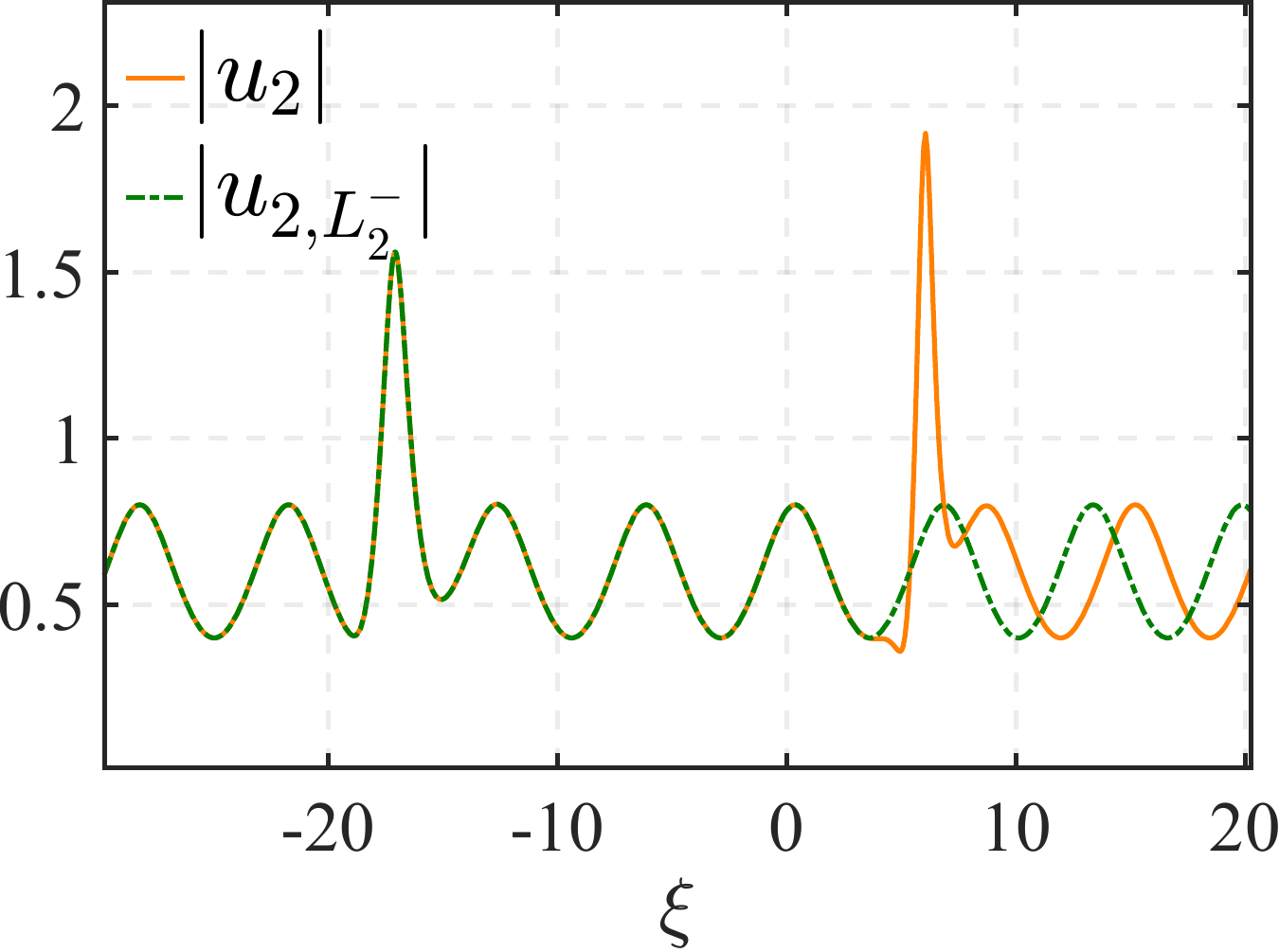}	
		\caption{The asymptotic behaviors of the two-elliptic localized solution $|u_2|$ with parameters $\kappa=2.08\mathrm{i},\rho={3.25-1.73\mathrm{i}},\omega_1=3.25,\omega_3=-3.31\mathrm{i}$, $z_1=-0.46-2.06\mathrm{i},z_2=-0.65+4.41\mathrm{i}$ and $\alpha_1=\alpha_2=1$ along the lines $L_{1,2}^{\pm}$ when $t=\pm 27$ respectively. (A) The orange solid curve and the green dash dot curve represent $|u_2|$ and $|u_{2,L_1^+}|$ when $t=27$ respectively. (B) The orange solid curve and the green dash dot curve represent $|u_2|$ and $|u_{2,L_2^+}|$ when $t=27$ respectively. (C)  The orange solid curve and the green dash dot curve represent $|u_2|$ and $|u_{2,L_1^-}|$ when $t=-27$ respectively. (D) The orange solid curve and the green dash dot curve represent $|u_2|$ and $|u_{2,L_2^-}|$ when $t=-27$ respectively.}
		\label{Fig-Asym-second-order-line-lines}
	\end{figure}
	
	\section{Conclusions}
	In this work, we investigate $N$-elliptic localized solutions to the DNLS equation. We streamline the process of the effective integration method and obtain a uniform parameterization for the squared moduli. Based on this representation, elliptic function solutions to the DNLS equation are generated using Weierstrass functions. By introducing a uniform parameter $z$ as a substitution for the spectral parameter $\lambda$, we construct the fundamental solution matrix to the associated Lax pair. Furthermore, by employing the Darboux-B\"acklund transformation and the addition formula for Weierstrass sigma functions, we obtain $N$-elliptic localized solutions to the DNLS equation in two equivalent forms. Based on these forms, equivalent asymptotic analyses are derived both along and between the propagation directions of the elliptic localized waves.
	
	This work extends the framework established in \cite{Ling-NLS} to the Kaup-Newell hierarchy, obtaining the compact form of $N$-elliptic localized solutions for Kaup-Newell equations for the first time. In the solution construction process, we have introduced some novel computational techniques. The methodology presented here shows potential for extension to other Kaup-Newell equations, including the Fokas-Lenells equation, Chen-Lee-Liu equation, the Gerdjikov-Ivanov equation and so on.
	
	Within Hirota's bilinear framework, the essence of $N$-soliton solutions is understood as determinant identities. Previous works \cite{LS-NLS,Ling-NLS,LS-mKdV-solution,LS-sG} demonstrate the effectiveness of addition formulas for the elliptic functions in constructing $N$-elliptic localized solutions for soliton equations. The compatibility between these addition formulas and the structures of AKNS and Kaup-Newell equations suggests a parallel assertion: the essence of $N$-elliptic localized solutions for integrable equations lies in addition formulas. However, confirming this requires substantial work. First, current results only cover specific equations within AKNS and Kaup-Newell systems. Whether this approach applies to equations in other hierarchies requires further exploration. Second, we aim to establish deeper connections between the additive group structure of elliptic curves and integrable structures, potentially revealing more profound algebraic relationships \cite{Wright2015EffectiveIO,ZHAO2020132213}. Building on this foundation, we seek to simplify $N$-elliptic localized solution construction using algebraic geometry, ultimately developing methods extendable to higher-genus cases. These investigations are currently underway.

	\section*{Acknowledgements}
	Liming Ling is supported by the National Natural Science Foundation of China (No. 12471236), the Guangzhou Municipal Science and Technology Project (Guangzhou Science and Technology Plan, No. 2024A04J6245) and Guangdong Natural Science Foundation grant (No. 2025A1515011868).
	
	\section*{Conflict of interest}
	The authors declare that they have no conflict of interest with other people or organizations that may inappropriately influence the author's work.
	
	\section*{Data Availability}
	
	Data sharing is not applicable to this article as no new data were
	created or analyzed in this study.

	\appendix
	
	\titleformat{\section}[display]
	{\centering\LARGE\bfseries}{ }{11pt}{\LARGE}
	
	\titleformat{\subsection}[display]
	{\large\bfseries}{ }{10pt}{\large}
	
	\renewcommand{\appendixname}{Appendix \, \Alph{section}}

	\section{\appendixname. Weierstrass Functions}
	
	\setcounter{equation}{0}
	\setcounter{prop}{0}
	\setcounter{lemma}{0}
	
	\renewcommand\theequation{\Alph{section}.\arabic{equation}}
	\renewcommand\theprop{\Alph{section}.\arabic{prop}}
	\renewcommand\thelemma{\Alph{section}.\arabic{lemma}}
	
	Let \(\Omega = \{2m\omega_1 + 2n\omega_3 \mid m,n \in \mathbb{Z}\}\) denote a lattice in the complex plane with $\omega_1>0$ and $\Im(\omega_3)<0$. The fundamental Weierstrass functions are defined as follows:
	
	\subsection{ Weierstrass elliptic functions}
	
	The Weierstrass elliptic function is defined as
	\begin{equation}\label{Eq-wp}
		\wp(\theta;\Omega) = \frac{1}{\theta^2} + \sum_{\substack{\omega \in \Omega \\ \omega \neq 0}} \left( \frac{1}{(\theta-\omega)^2} - \frac{1}{\omega^2} \right).
	\end{equation}
	The function $\wp(\theta)$ is doubly periodic on the complex plane with fundamental periods $2\omega_1$ and $2\omega_3$. Denote $\omega_2=-\omega_1-\omega_3$, then $\wp(\theta)$
	is the solution to the following equation 
	\begin{equation}
		\big(\wp'(\theta)\big)^2=4\big(\wp(\theta)-e_1\big)\big(\wp(\theta)-e_2\big)\big(\wp(\theta)-e_3\big),
	\end{equation}
	where $\wp(\omega_i)=e_i,i=1,2,3$ and the prime denotes the derivative with respect to $\theta$. It is satisfied that
	\begin{equation}\label{Eq-special-values}
		\wp(\theta \pm \omega_i) = e_i + \dfrac{\prod_{l=1,l\neq i}^3(e_i-e_l)}{\wp(\theta) - e_i}.
	\end{equation}
	The Weierstrass elliptic function $\wp(\theta)$ is associated with the Jacobi snoidal function via 
	\begin{equation}\label{Eq-relation-wp-sn}
		\wp(\theta)=e_3+\frac{e_1-e_3}{\ssn^2(\alpha_0\theta,k_0)},
	\end{equation}
	where $\alpha_0=\sqrt{e_1-e_3},k_0=\sqrt{\frac{e_2-e_3}{e_1-e_3}}$.
	\subsection{ Weierstrass zeta functions}
	
	The Weierstrass zeta function is defined by
	\begin{equation}\label{Eq-zeta}
		\zeta(\theta;\Omega) = \frac{1}{\theta} + \int_0^\theta \left( \frac{1}{\chi^2} - \wp(\chi) \right) d\chi.
	\end{equation}
	This function satisfies
	\begin{align}\label{Quasi-periodic-zeta}
		\zeta(\theta+2\omega_i) = \zeta(\theta) + 2\zeta(\omega_i), \quad i=1,3.
	\end{align}
	Furthermore, it is an odd function with the derivative formula
	\begin{equation}\label{Eq-derivative-zeta}
		\zeta'(\theta) = -\wp(\theta).
	\end{equation}
	\subsection{Weierstrass sigma functions}
	
	The Weierstrass sigma function is defined by the canonical product
	\begin{equation}\label{Eq-sigma}
		\sigma(\theta;\Omega) := \theta \prod_{\substack{\omega \in \Omega\\ \omega \neq 0}} \left(1-\frac{\theta}{\omega}\right)\exp\left(\frac{\theta}{\omega} + \frac{\theta^2}{2\omega^2}\right),
	\end{equation}
	which is an entire function in the complex plane with simple zeros at lattice points $\omega\in\Omega$. This function exhibits quasi-periodicity characterized by
	\begin{align}\label{Eq-quasi-periodic-sigma}
		\sigma(\theta+2\omega_i) = -\sigma(\theta)e^{2\zeta(\omega_i)(\theta+\omega_i)}, \quad i=1,3.
	\end{align}
	As an odd function, its logarithmic derivative satisfies
	\begin{equation}\label{Eq-derivative-sigma}
		\frac{\sigma'(\theta)}{\sigma(\theta)} = \zeta(\theta).
	\end{equation}
		\subsection{Integration formulas}
	
	The following integrals are essential for constructing $N$-elliptic localized solutions:
	\begin{align}\label{Eq-integration-formulas}
		\begin{split}
			& \int \frac{\wp^{\prime}(\theta) d \theta}{\wp(\theta)-\wp(\theta_1)} = \ln\big(\wp(\theta)-\wp(\theta_1)\big),\\
			& \int \frac{d \theta}{\wp(\theta)-\wp(\theta_1)} = \frac{1}{\wp^{\prime}(\theta_1)}\left(\ln \frac{\sigma(\theta-\theta_1)}{\sigma(\theta+\theta_1)} + 2 \theta \zeta(\theta_1)\right).
		\end{split}
	\end{align}
	\subsection{Fundamental formulas}
	
	The following formulas for Weierstrass functions are crucial for deriving the $N$-elliptic localized solution:
	\begin{equation}\label{Eq-formulas of Weierstrass functions}
		\begin{split}
			&\wp(\theta_1)-\wp(\theta_2) = -\frac{\sigma(\theta_1+\theta_2)\sigma(\theta_1-\theta_2)}{\sigma^2(\theta_1)\sigma^2(\theta_2)}, \\ 
			& \zeta(\theta_1+\theta_2)-\zeta(\theta_2)-\zeta(\theta_1) = \frac{1}{2}\frac{\wp'(\theta_2)-\wp'(\theta_1)}{\wp(\theta_2)-\wp(\theta_1)}, \\
			& \zeta(\theta_1)+\zeta(\theta_2)+\zeta(\theta_3)-\zeta(\theta_1+\theta_2+\theta_3) = \frac{\sigma(\theta_1+\theta_2)\sigma(\theta_1+\theta_3)\sigma(\theta_2+\theta_3)}{\sigma(\theta_1)\sigma(\theta_2)\sigma(\theta_3)\sigma(\theta_1+\theta_2+\theta_3)}.
		\end{split}
	\end{equation}
	The first formula is actually addition formula for the Weierstrass $\wp$- functions. Taking the limit $\theta_1\rightarrow \theta_2$ yields
	\begin{align}\label{Eq-half-argument-1}
		\wp'(\theta) = -\frac{\sigma(2\theta)}{\sigma^4(\theta)},
	\end{align}
	where we utilize the asymptotic behavior $\sigma(\theta)\sim\theta$ as $\theta\rightarrow 0$ from \eqref{Eq-sigma}. Similarly, the limit $\theta_1\rightarrow \theta_2$ in the second identity gives
	\begin{equation}\label{Eq-half-argument-2}
		\zeta(2\theta) = 2\zeta(\theta) + \frac{\wp''(\theta)}{2\wp'(\theta)}.
	\end{equation}
	Equations \eqref{Eq-half-argument-1}-\eqref{Eq-half-argument-2} are recognized as the half-argument formulas for Weierstrass functions.
	
	\subsection{Addition formula for sigma functions}
	
	The addition formula for the sigma function, fundamental to our construction, follows from the first identity in \eqref{Eq-formulas of Weierstrass functions}:
	
	\begin{prop}
		The Weierstrass sigma function satisfies
		\begin{equation}\label{Eq-addition formulas of the sigma functions}
			\begin{split}
				&\sigma(\theta_1+\theta_2)\sigma(\theta_1-\theta_2)\sigma(\theta_3+\theta_4)\sigma(\theta_3-\theta_4) \\
				&\quad + \sigma(\theta_1+\theta_3)\sigma(\theta_1-\theta_3)\sigma(\theta_4+\theta_2)\sigma(\theta_4-\theta_2) \\
				= &\sigma(\theta_3+\theta_2)\sigma(\theta_3-\theta_2)\sigma(\theta_1+\theta_4)\sigma(\theta_1-\theta_4),
			\end{split}
		\end{equation}
		for arbitrary $\theta_i \in \mathbb{C}$, $i=1,2,3,4$.
	\end{prop}
	\begin{proof}
		From \eqref{Eq-formulas of Weierstrass functions} we obtain
		\begin{equation}\label{Eq-first} 
			\sigma(\theta_1+\theta_2)\sigma(\theta_1-\theta_2)\sigma(\theta_3+\theta_4)\sigma(\theta_3-\theta_4) = \big(\wp(\theta_1) - \wp(\theta_2)\big)\big(\wp(\theta_3) - \wp(\theta_4)\big)\sigma^2(\theta_1)\sigma^2(\theta_2)\sigma^2(\theta_3)\sigma^2(\theta_4).
		\end{equation}
		Similarly,
		\begin{equation}\label{Eq-second} 
			\begin{split}
				\sigma(\theta_1+\theta_3)\sigma(\theta_1-\theta_3)\sigma(\theta_4+\theta_2)\sigma(\theta_4-\theta_2) &= \big(\wp(\theta_1) - \wp(\theta_3)\big)\big(\wp(\theta_4) - \wp(\theta_2)\big)\sigma^2(\theta_1)\sigma^2(\theta_2)\sigma^2(\theta_3)\sigma^2(\theta_4),\\
				\sigma(\theta_3+\theta_2)\sigma(\theta_3-\theta_2)\sigma(\theta_1+\theta_4)\sigma(\theta_1-\theta_4) &= \big(\wp(\theta_3) - \wp(\theta_2)\big)\big(\wp(\theta_1) - \wp(\theta_4)\big)\sigma^2(\theta_1)\sigma^2(\theta_2)\sigma^2(\theta_3)\sigma^2(\theta_4).
			\end{split}
		\end{equation}
		Substituting these into the algebraic identity
		\[
		\big(\wp(\theta_1)-\wp(\theta_2)\big)\big(\wp(\theta_3)-\wp(\theta_4)\big) + \big(\wp(\theta_1)-\wp(\theta_3)\big)\big(\wp(\theta_4)-\wp(\theta_2)\big) = \big(\wp(\theta_3)-\wp(\theta_2)\big)\big(\wp(\theta_1)-\wp(\theta_4)\big)
		\]
		yields \eqref{Eq-addition formulas of the sigma functions}.
	\end{proof}

	\section{References}
	\bibliographystyle{siam}
	
	\bibliography{Reference}

\end{document}